\documentclass[preprint,12pt]{elsarticle}

\usepackage{graphs,mathrsfs}
\usepackage{amsmath,amsfonts,amsthm,amssymb}
\usepackage{hyperref,xcolor,subfig}

\newtheorem{tlem}{Lemma}
\newtheorem{tthm}[tlem]{Theorem}

\newtheorem{tcor}[tlem]{Corollary}
\newtheorem{tdef}[tlem]{Definition}
\newtheorem{trem}[tlem]{Remark}
\def \ie {i.e.~}
\def \lhs {l.h.s.~}
\def \rhs {r.h.s.~}
\def \CP { \mathcal{CP} }
\def \ECP { \mathcal{ECP} }
\def \F { \mathcal{F} }

\def \S { \mathscr{S} }

\journal{Discrete Applied Mathematics}

\begin{document}

\begin{frontmatter}

\title{A polyhedral approach for the Equitable Coloring Problem\tnoteref{grant}}

\author[a]{Isabel M\'endez-D\'iaz\fnref{correspon}}
\author[b,c]{Graciela Nasini}
\author[b,c]{Daniel Sever\'in}

\address[a]{ FCEyN, Universidad de Buenos Aires,
	Argentina }

\address[b]{ FCEIA, Universidad Nacional de Rosario,
	Argentina }

\address[c]{ CONICET, Argentina }

\tnotetext[grant]{Partially supported by grants UBACyT X143 (2008-2010), PID-CONICET 204 (2010-2012)
and PICT 2006-1600.\\
\emph{E-mail addresses}: \texttt{imendez@dc.uba.ar} (I. M\'endez-D\'iaz),
\texttt{nasini@fceia.unr.edu.ar} (G. Nasini), \texttt{daniel@fceia.unr.edu.ar} (D. Sever\'in)}
\fntext[correspon]{Corresponding author at Departamento de Ciencias de la Computaci\'on, Facultad de Ciencias Exactas y Naturales, Universidad de Buenos Aires,
Intendente Guiraldes 2160 (Ciudad Universitaria, Pabell\'on 1), Argentina}

\begin{abstract}
In this work we study the polytope associated with a 0,1-integer programming formulation for the
Equitable Coloring Problem. We find several families of valid inequalities and derive sufficient
conditions in order to be facet-defining inequalities. 
We also present computational evidence that shows the efficacy of these inequalities used in a
cutting-plane algorithm.
\end{abstract}


\end{frontmatter}

%%%%%%%%%%%%%%%%%%%%%%%%%%%%%%%%%%%%%%%%%%%%%%%%%%%%%%%%%%%%%%%%%%%%%%%%%%%%%%%%

\section{Introduction} \label{SINTRO}

In graph theory, there is a large family of optimization
problems having relevant practical importance, besides its theoretical
interest. One of the most representative problem of this family is the
\emph{Graph Coloring Problem} (GCP), which arises in many applications
such as scheduling, timetabling, electronic bandwidth allocation and sequencing problems.

Given a simple graph $G = (V, E)$, where $V$  is the set of vertices
and $E$ is the set of edges, a coloring of $G$ is an assignment of colors to each vertex such that the endpoints of any edge have different colors.
A $k$-coloring is a coloring that uses $k$ colors. Equivalently, a $k$-coloring can be defined as a partition of $V$ into $k$ subsets, called \emph{color classes}, such that adjacent vertices belong to different classes.
Given a $k$-coloring, color classes are denoted by $C_1, \ldots, C_k$ assuming that, for each $i \in \{1,\ldots,k\}$, vertices in $C_i$ are colored with color $i$.
We can also define a $k$-coloring of $G$ as a mapping $c: V \rightarrow \{1,\ldots,k\}$ such that
$c(u) \ne c(v)$ for all $(u, v) \in E$.
% and $c(u) = j$ if $u \in C_j$
The GCP consists of finding the minimum number of colors such that a coloring exists. This minimum
number of colors is called the \emph{chromatic number} of the graph $G$ and is denoted by $\chi(G)$.

Some applications impose additional restrictions arising variations of GCP.
%such as \emph{List Coloring}, \emph{Total Coloring}, etc.
For instance, in scheduling problems, it may be required to ensure the uniformity of the distribution of workload employees.
The addition of these extra \emph{equity} constraints gives rise to the
\emph{Equitable Coloring Problem} (ECP).
An \emph{equitable $k$-coloring} (or just $k$-eqcol) of $G$ is a $k$-coloring satisfying the \emph{equity constraints},
\ie $\left| |C_i| - |C_j| \right| \leq 1$, for $i, j \in \{1, \ldots, k \}$ or, equivalently,
$\lfloor n / k \rfloor \leq |C_j| \leq \lceil n / k \rceil$
for each $j \in \{1, \ldots, k \}$. The \emph{equitable chromatic number} of $G$, $\chi_{eq}(G)$, is the
minimum $k$ for which $G$ admits a $k$-eqcol. The ECP consists of finding $\chi_{eq}(G)$.

The ECP was introduced in \cite{MEYER}, motivated by an application concerning \emph{garbage collection}
\cite{EXAMPLE2}. 
%In this application, each garbage collection route is represented by a vertex. If
%two routes can not be traversed the same day, those vertices representing them are adjacent.
%Also, it is expected that each day the same number of routes were traversed. Therefore, the problem of
%asignning one of the 6 weekly working days is reduced to find an equitable coloring on the graph thta uses 6 colors.
%
Other applications of the ECP concern \emph{load balancing problems} in multiprocessor machines \cite{EXAMPLE3}
and results in \emph{probability theory} \cite{EXAMPLE1}. An introduction to ECP and some basics results are provided in
\cite{KUBALE}.

Computing $\chi_{eq}(G)$ for arbitrary graphs is proved to be $NP$-hard and
just a few families of graphs are known to be easy such as complete $n$-partite, complete split, wheel and tree graphs \cite{KUBALE}.
In particular, if $G$ has a universal vertex $u$, the cardinality of the color classes of any equitable coloring in $G$ is at most two and the
color classes of exactly two vertices correspond to a matching in the complement of $G$.
In other words, the ECP is polynomial when $G$ has at least one universal vertex.

There exist some remarkable differences between GCP and ECP. Unlike GCP, a graph admiting a $k$-eqcol
may not admit a $(k+1)$-eqcol. This leads us to define the \emph{skip set} of $G$, $\S(G)$, as the set
of $k \in \{ \chi_{eq}(G), \ldots, n\}$ such that $G$ does not admit any $k$-eqcol. For instance, if $G = K_{3,3}$, \ie
the \emph{complete bipartite graph} with partitions of size 3, then $G$ admits a 2-eqcol but does not admit a 3-eqcol.
Here, $\S(K_{3,3}) = \{3\}$. Computing the skip set of a graph is as hard as computing the equitable chromatic number.
If $\S(G) = \varnothing$, we say that $G$ is \emph{monotone}. For instance, trees are monotone graphs \cite{EQTREE}.
% while the skip set of a bipartite graph\cite{BIPARTITE} is $\S(K_{m,m}) = \biggl \{ k : 3 \leq k \leq m \land \biggl\lceil
         %\dfrac{m}{\lfloor k/2 \rfloor} \biggr\rceil - \biggl\lfloor
         %\dfrac{m}{\lceil k/2 \rceil} \biggr\rfloor \geq 2 \biggr \}$

Another drawback emerging from ECP is that the equitable chromatic number of a graph
can be smaller than the equitable chromatic number of one of its induced subgraphs.
In particular, in an unconnected graph, equitable chromatic numbers of each connected component are uncorrelated with the chromatic number of the whole graph. 
%So, we can not restrict ourselves to connected graphs as in the case of GCP.

On the other hand, some useful properties of GCP also hold for ECP. For example, it is known that $G$ admits
$k$-eqcols for $k \geq \Delta(G)+1$, where $\Delta(G)$ is the maximum degree of vertices in $G$.
In \cite{KKALGORITHM} a polynomial time algorithm which produces a $(\Delta(G)+1)$-eqcol is presented.

\emph{Integer linear programming} (ILP) approach together with algorithms which exploit the
polyhedral structure proved to be the best tool for dealing with coloring problems.
Although many ILP formulations are known for GCP, as far as we know, just two of these models
were adapted for ECP. One of them, given in \cite{BYCBRA}, is based on the \emph{asymmetric representatives}
model for the GCP \cite{REPRESENTATIVES}. The other one, proposed by us in \cite{MACI},
is based on the classic \emph{color assignments to vertices} model \cite{AARDAL} with further
improvements stated in \cite{BCCOL}.

The goal of this paper is to study the last model from a polyhedral point of view and determine families
of valid inequalities which can be useful in the context of an efficient cutting-plane algorithm. 

The remainder of the paper is organized as follows. 

In sections \ref{SPOLYT}-\ref{SNEWINEQ}, we study the facial structure of the polytope associated
with the formulation given in \cite{MACI}. We introduce several families of valid inequalities which
always define high dimensional faces.
Section \ref{SCOMPU} is devoted to describe a cutting-plane algorithm for solving ECP. We expose
computational evidence for reflecting the improvement in the performance when the cutting-plane
algorithm uses the new inequalities as cuts. That algorithm is then used to reinforce bounds on a
Branch and Bound enumeration tree. At the end,
% Finally, the paper is closed in Section \ref{SCONCLU}, with
a conclusion is presented.\\

Some definitions and notations will be useful in the following.

Given a graph $G = (V, E)$ we consider $V= \{1,\ldots,n\}$. The complement of $G$ is denoted by $\overline{G}$.
We also denote by $K_n$ the \emph{complete graph} of $n$ vertices.
The percentage of density of $G$ is $\dfrac{100|E|}{|V|(|V| - 1)/2}$. For instance, the percentage of density of any complete graph is 100.
Given $u \in V$, the \emph{degree of $u$} is the number of vertices adjacent to $u$ and is denoted by $\delta(u)$.
%The maximum degree of vertices in $G$ is denoted by $\Delta(G)$.
For any $S \subset V$, $G[S]$ is the graph induced by $S$ and $G - S$ is the graph obtained by the deletion
of vertices in $S$, \ie $G - S = G[V \backslash S]$. In particular, if $S = \{u\}$ we just write $G - u$
instead of $G - \{u\}$. A \emph{stable set} is a set of vertices in $G$, no two of which are adjacent. We denote by $\alpha(G)$ the \emph{stability number} of $G$, \ie the maximum cardinality of a stable set of $G$. Given $S \subset V$, we also denote by $\alpha(S)$ the stability number of $G[S]$. We say that $S$ is $k$-\emph{maximal} if 
$\alpha(S)=k$ and for all $v \in V \backslash S$, $\alpha(S\cup\{v\}) = k + 1$. In particular, if $S$ is 1-maximal, we say that $S$ is a \emph{maximal clique}.
Given $u \in V$, the \emph{neighborhood} of $u$, $N(u)$, is the set of vertices adjacent to $u$, and the
\emph{closed neighborhood} of $u$, $N[u]$, is the set $N(u) \cup \{u\}$. A vertex $u \in V$ is a
\emph{universal vertex} if $N[u] = V$. A \emph{matching} of $G$ is a subset of edges such that no pair of them has a common extreme point.
Whenever it is clear from the context, we will write $\chi_{eq}$
rather than $\chi_{eq}(G)$. The same convention also applies for other operators that depend on $G$
such as $\S$ and $\Delta$.\\

%We denote by $\nu(G)$ the number of edges of a \emph{maximum matching} of $G$. 

Throughout the paper, we consider graphs with at least five vertices and one edge, and not containing universal vertices nor $K_{n-1}$
as an induced subgraph. Thus, for a given graph $G$ we assume that $2 \leq \chi_{eq}(G) \leq n-2$. The remaining cases can be solved in polynomial time.

\section{The polytope $\ECP$} \label{SPOLYT}

%Since $n$ colors will always suffice for coloring graph G, 
A straightforward ILP model for GCP can be obtained by modeling colorings with two sets of binary
variables: variables $x_{vj}$ for $v \in V$ and $j\in \{1,\ldots, n\}$ where $x_{vj}=1$ if and only if the coloring assigns color $j$ to vertex $v$, and
variables $w_j$ for $j\in \{1,\ldots, n\}$ where $w_j=1$ if and only if color $j$ is used in the coloring. 
The formulation is shown below:
\begin{align}
  & \sum_{j = 1}^n x_{vj} = 1, & &\forall~ v \in V \label{RASSIGN}\\
  & x_{uj} + x_{vj} \leq w_j,  & &\forall~ (u,v) \in E,~ 1 \leq j \leq n  \label{RADYAC}
\end{align}
Constraints (\ref{RASSIGN}) assert that each vertex has to be colored by a
unique color and constraints (\ref{RADYAC}) ensure that two adjacent vertices can not share the same color.
Hence, the chromatic number can be computed by minimizing $\sum_{j = 1}^n w_j$. 

This formulation presents a disadvantage: the number of integer solutions $(x,w)$ with the same value
$\sum_{j = 1}^n w_j$ is very large. A technique widely used in combinatorial optimization to deal
with this kind of problem is the concept of \emph{symmetry breaking} \cite{MARGOT}. This technique is
applied in \cite{BCCOL}, where the following constraints are added to the previous formulation in order
to remove (partially) symmetric solutions:
\begin{align}
  & w_{j + 1} \leq w_j,        & &\forall~ 1 \leq j \leq n-1  \label{RORDER}
\end{align}
which means that color $j + 1$ may be used only if color $j$ is also used.

Given a partition of $V$ into color classes, let us observe that permutations of colors between those sets
yield symmetric colorings. In \cite{BCCOL}, additional constraints are proposed in order to drop
most of these colorings by sorting the color classes by the minimum label of the vertices
belonging to each set and only considering the coloring that assigns color $j$ to the $j$th color class.
These constraints are:
\begin{align}
	&x_{vj} = 0, & &\forall~ 1 \leq v < j \leq n \label{RREPR1}\\
        &x_{vj} \leq \sum_{u = j-1}^{v-1} x_{u j-1}, & &\forall~ 2 \leq j \leq v \leq n \label{RREPR2}
\end{align}
Even though the formulation consisting of constraints (\ref{RASSIGN})-(\ref{RREPR2}) eliminates a greater amount of
symmetrical solutions, it is difficult to characterize the integer polyhedron associated to that formulation since it
depends on the labeling of vertices \cite{BCCOL}.

From now on, we represent colorings of $G$ as binary vectors $(x,w)$ satisfying constraints (\ref{RASSIGN})-(\ref{RORDER}) and
we call \emph{Coloring Polytope}, $\CP(G)$, to the convex hull of binary vectors $(x,w)$ that represent colorings of $G$.

In order to characterize equitable colorings, we add the following constraints to the model:
\begin{align}
  & x_{vj} \leq w_j, & &\forall~ v~ \textrm{isolated},~ 1 \leq j \leq n \label{RISOL} \\
  & \sum_{v \in V} x_{vj} \geq \sum_{k=j}^n \biggl \lfloor \dfrac{n}{k} \biggr \rfloor
          \bigl( w_k - w_{k+1} \bigr), & &\forall~ 1 \leq j \leq n - 1 \label{RLOWER} \\
  & \sum_{v \in V} x_{vj} \leq \sum_{k=j}^n \biggl \lceil \dfrac{n}{k} \biggr \rceil
          \bigl( w_k - w_{k+1} \bigr), & &\forall~ 1 \leq j \leq n - 1 \label{RUPPER}
\end{align}
where $w_{n+1}$ is a dummy variable set to 0. Constraints (\ref{RISOL}) ensure that isolated vertices
use enabled colors and (\ref{RLOWER})-(\ref{RUPPER}) are precisely the equity constraints. %, since the \rhs equals $\lfloor n/k \rfloor$ or $\lceil n/k \rceil$ when the solution $(x,w)$ represents a $k$-eqcol.
The \emph{Equitable Coloring Polytope} $\ECP(G)$ is defined as the convex hull of binary vectors
$(x,w)$ that represent equitable colorings of $G$, \ie they satisfy constraints (\ref{RASSIGN})-(\ref{RORDER})
and (\ref{RISOL})-(\ref{RUPPER}).

From now on, we present equitable colorings by using mappings, color classes or binary vectors, according to our convenience.

We also work with two useful operators over colorings. The first one is based on the fact that
\emph{swapping} colors in a $k$-eqcol produces a $k$-eqcol indeed.

\begin{tdef}
Let $c$ be a $k$-eqcol of $G$ with color classes $C_1$, $\ldots$, $C_k$ and $L = (j_1, j_2, \ldots, j_r)$ be an ordered
list of different colors in $\{1,\ldots,k\}$. We define $swap_L(c)$ as the $k$-eqcol with color classes $C'_1$, $\ldots$,
$C'_k$ which satisfies $C'_{j_t} = C_{j_{t+1}}~~\forall~1 \leq t \leq r-1$, $C'_{j_r} = C_{j_1}$ and
$C'_i = C_i~~\forall~i \in \{1, 2, \ldots, k \} \backslash \{j_1, j_2, \ldots, j_r\}$.
\end{tdef}

The other operator takes a $k$-eqcol whose color classes have at most 2 vertices and returns a $(k+1)$-eqcol. 
%The new coloring remains equitable.

\begin{tdef}
Let $c$ be a $k$-eqcol of $G$ with $\lceil n/2 \rceil \leq k \leq n-1$ and $v\neq v'$ 
such that $c(v) = c(v')$. We define $intro(c,v)$ as a $(k+1)$-eqcol $c'$ which
satisfies $c'(v) = k + 1$ and $c'(i) = c(i)~~\forall~i \in V \backslash \{v\}$.
\end{tdef}

%We now characterize the dimension of $\ECP$ by adapting the proof presented in Proposition
%1 of \cite{BCCOL} for $\CP$.

\begin{trem} \label{REMPROC}
Let us observe that colorings with $n-1$ and $n$ colors are always equitable. Then, we can use
Proposition 1 of \cite{BCCOL} to prove that the following $n^2 - \chi_{eq} - |\S|$ equitable colorings
are affinely independent:
\begin{enumerate}
\item A $(n-1)$-eqcol $c$ such that $C_{n-1}$ has two vertices, namely $u_1$ and $u_2$.
\item %The $(n-1)$-eqcol 
$swap_{n-1,j}(c)$ for each $j \in \{1,\ldots,n-2\}$.
\item The $n$-eqcol $c' = intro(c,u_1)$.
\item %The $n$-eqcol 
$swap_{n,j,j'}(c')$ for each $j,j' \in \{1,\ldots,n-1\}$ such that $j' \neq j$.
\item %The $n$-eqcol 
$swap_{n,j}(c')$ for each $j \in \{1,\ldots,n-1\}$.
\item An arbitrary $k$-eqcol of $G$ for each $k \in \{\chi_{eq}, \ldots, n - 2\} \backslash \S$.
\end{enumerate}
%Let us observe that the number of colorings defined in items 1-4 plus $c$ is
%$n^2 - n + 1$ and all of them are equitable. Then, if we take in item 5 an arbitrary
%$k$-eqcol of $G$ for $k \in \{\chi_{eq}, \ldots, n-2 \} \backslash \S$, we obtain
%$n^2 - \chi_{eq} - |\S|$ affinely independent equitable colorings.
\end{trem}

\begin{tthm} \label{TDIM}
The dimension of $\ECP$ is $n^2 - (\chi_{eq} + |\S| + 1)$ and a minimal equation
system is defined by:
\begin{align}
	&\sum_{j = 1}^n x_{vj} = 1, &\forall~ v \in V, \label{TDIM1} \\
	&w_j = 1, &\forall~ 1 \leq j \leq \chi_{eq}, \label{TDIM2} \\
	&w_j = w_{j+1}, &\forall~ j \in \S, & \label{TDIM3} \\
	&\sum_{v \in V} x_{vn} = w_n. & \label{TDIM4}
\end{align}
\end{tthm}
\begin{proof}
From Remark \ref{REMPROC}, $dim(\ECP) \geq n^2 - (\chi_{eq} + |\S| + 1)$. We only need to note that
$\ECP \subset \mathbb{R}^{n^2 + n}$ and that every equitable
coloring satisfies $n + \chi_{eq} + \S + 1$ mutually independent equalities given in (\ref{TDIM1})-(\ref{TDIM4}).
\end{proof}

Let us analyze the faces of $\ECP$ defined by restrictions of the formulation.
For non-negativity constraints and inequalities (\ref{RORDER}) we adapt the proofs given in \cite{BCCOL}
for $\CP$.

\begin{tthm} \label{TORIGINAL1}
Let $v \in V$ and $1 \leq j \leq n$. Constraint $x_{vj} \geq 0$ defines a facet of $\ECP$.
\end{tthm}
\begin{proof}
%We follow Proposition 2 of \cite{BCCOL}.
We exhibit $n^2 - \chi_{eq} - |\S| - 1$ affinely independent colorings that lie on the face of $\ECP$ defined by inequality $x_{vj} \geq 0$.
Let us consider the following cases:\\
\textbf{Case $j \leq n - 2$}. Let $u_1, u_2 \in V \backslash \{v\}$ be non adjacent vertices and let $c$ be a $(n-1)$-eqcol such that $c(v) \neq j$
and $C_{n-1} = \{u_1, u_2\}$. We consider the set of colorings given by Remark \ref{REMPROC} starting with $c$ and choosing the arbitrary $k$-eqcols 
in item 6 satisfying that vertex $v$ is not painted with color $j$.
It is clear that all these colorings, except $swap_{n,j,c(v)}(c')$ where $c'=intro(c,u_1)$, lie in the face defined by the inequality.\\
%Thus, we get $n^2 - \chi_{eq} - |\S| - 1$ affinely independent colorings that
%lie on the face, concluding that $x_{vj} \geq 0$ defines a facet of $\ECP$.\\
\textbf{Case $j = n - 1$}. Let $S$ be the set of $n$-eqcols and $(n-1)$-eqcols presented in the previous case for
$j = n-2$. We consider the colorings $swap_{n-1,n-2}(\tilde{c})$ for each $\tilde{c} \in S$ and an arbitrary $k$-eqcol of $G$ for each
$k \in \{\chi_{eq}, \ldots, n - 2\} \backslash \S$.\\
%These colorings  are affinely indepedent and lie on the face, concluding that $x_{vn-1} \geq 0$ defines a facet of %$\ECP$.\\
\textbf{Case $j = n$}. Let $u_2$ be a vertex not adjacent to $v$. We consider the set of colorings given by
Remark \ref{REMPROC} starting with a $(n-1)$-eqcol $c$ such that $C_{n-1} = \{v, u_2\}$. It is clear that all these
colorings, except $intro(c,v)$, lie in the face defined by the inequality. 
%We get $n^2 - \chi_{eq} - |\S| - 1$ affinely independent colorings that satisfy $x_{vn} = 0$, concluding that $x_{vn} %\geq 0$ defines a facet of $\ECP$.
%%Proposition 2 of \cite{BCCOL} gives affinely independent colorings for the non-negativity
%%inequality to define a facet of $\CP$. If we restrict ourselves to take the equitable
%%colorings of the list of colorings given by it, we obtain the needed ones for this
%%inequality to define a facet of $\ECP$.
\end{proof}

Let $1 \leq j \leq n-1$ and $\F$ be the face of $\ECP$ defined by constraint (\ref{RORDER}), \ie $w_{j+1} \leq w_j$.
Let us notice that, if $G$ does not admit a $j$-eqcol, \ie $j \in \{1,\ldots,\chi_{eq}-1\} \cup \S$, then (\ref{RORDER}) is
a linear combination of equations of the minimal system and, therefore, $\F = \ECP$. In addition, if $j = n-1$, the class of color $n-1$
of every coloring $(x,w)$ satisfying $w_n = w_{n-1}$ have at most one vertex and, therefore, $(x,w)$ verifies
$\sum_{v \in V} x_{vn-1} = w_{n-1}$. Then, $\F$ is not a facet of $\ECP$. For the remaining cases, we have the following result. 

\begin{tthm} \label{TORIGINAL2}
If $G$ admits a $j$-eqcol and $j \leq n-2$, constraint $w_{j+1} \leq w_j$ defines a facet of $\ECP$.
\end{tthm}
\begin{proof}
%We follow Proposition 3 of \cite{BCCOL}.
Let us consider the set of colorings from Remark \ref{REMPROC} but excluding the $j$-eqcol from item 6.
Clearly, the remaining colorings lie on the face and (\ref{RORDER}) defines a facet of $\ECP$.
%Proposition 3 of \cite{BCCOL} gives affinely independent colorings for
%the inequality (\ref{RORDER}) to define a facet of $\CP$. If we restrict ourselves to take the equitable colorings of the list of colorings given by it, we obtain the needed
%ones for this inequality to define a facet of $\ECP$.\\
\end{proof}

The following theorems are related to the faces of $\ECP$ defined by the equity constraints.

\begin{tthm} \label{TORIGINAL3}
Let $1 \leq j \leq n-1$. Constraint
\[ \sum_{v \in V} x_{vj} \geq \sum_{k=j}^n \biggl\lfloor \frac{n}{k} \biggr\rfloor (w_k - w_{k+1}) \]
defines a facet of $\ECP$.
\end{tthm}
\begin{proof}
Let $u_1, u_2$ be non adjacent vertices and let $c$ be a $(n-1)$-eqcol $c$ such that $C_{n-1} = \{u_1, u_2\}$.
We consider the set of colorings given by Remark \ref{REMPROC} starting with $c$ and choosing $k$-eqcols in
item 6 satisfying $|C_j|=\lfloor n/k \rfloor$ when $k \geq j$. The proposed colorings, except the $(n-1)$-eqcol
that satisfies $C_j = \{u_1, u_2\}$, lie on the face and therefore (\ref{RLOWER}) defines a facet of $\ECP$.
\end{proof}

Let us observe that if $1 \leq j \leq n-2$, the face of $\ECP$ defined by (\ref{RUPPER}) is not a facet. Indeed, every coloring $(x,w)$ lying on the face satisfies $\sum_{v \in V} x_{vn-1} = w_{n-1}$. 
For the case $j = n-1$, the constraint (\ref{RUPPER}) is $\sum_{v \in V} x_{vn-1} \leq 2 w_{n-1} - w_n$ and we have: 

\begin{tthm} \label{TORIGINAL4}
The inequality $\sum_{v \in V} x_{vn-1} \leq 2 w_{n-1} - w_n$ defines a facet of $\ECP$.
\end{tthm}
\begin{proof}
Since $n \geq 5$ and $\chi_{eq}(G) \leq n-2$ there exist $u_1, u_2, u_3, u_4, u_5 \in V$ such that
$u_1$ is not adjacent to $u_2$ and $u_3$ is not adjacent to $u_4$. Let $c$ be a $(n-1)$-eqcol $c$
such that $c(u_1) = c(u_2) = n-1$.
We consider the colorings from items 1,3,4,5 in Remark \ref{REMPROC} together with the following ones:
\begin{itemize}
\item The $(n-2)$-eqcol $\hat{c}$ such that $\hat{c}(u_1) = \hat{c}(u_2) = c(u_3)$,
$\hat{c}(u_3) = c(u_4)$ and $\hat{c}(i) = c(i) ~~\forall~i \in V \backslash \{u_1, u_2, u_3\}$.
%Let $j_1 = c^5(u_1) = c(u_3)$ and $j_3 = c^5(u_3) = c(u_4)$. In few words, $u_1, u_2$ share the color $j_1$
%and $u_3, u_4$ share the color $j_3$.
\item %The $(n-2)$-eqcol
$swap_{j,c(u_3)}(\hat{c})$ for each $j \in \{1, \ldots, n-2\} \backslash \{c(u_3),c(u_4)\}$. % We have $n-4$ colorings.
\item %The $(n-2)$-eqcol
$swap_{\hat{c}(u_5),c(u_4)}(\hat{c})$.
\item An arbitrary $k$-eqcol of $G$ for each $k \in \{\chi_{eq}, \ldots, n - 3\} \backslash \S$. % We have $n - \chi_{eq} - |\S| - 2$ colorings.
\end{itemize}
The proof for the affine independence of the previous $n^2 - \chi_{eq} - |\S| - 1$ colorings is similar to the one for the colorings generated
in Remark \ref{REMPROC}.
% Hence, (\ref{RUPPER}) is a facet-defining inequality of $\ECP$.\\
\end{proof}

\subsection{Valid inequalities from $\CP$}

Taking into account that valid inequalities for $\CP$ are also valid for $\ECP$, %In some cases, facet-defining
%inequalities of $\CP$ lead to facet-defining inequalities of $\ECP$. Even if a facet-defining inequality
%of $\CP$ did not define a facet of $\ECP$, a \emph{lifting procedure} could be applied in order to
%increase the dimension of the face and reach a facet of $\ECP$ eventually.
%So, it is natural to start analyzing known facet-defining inequalities of $\CP$.
%This is the purpose of the rest of this section.
in this section we analyze the faces of $\ECP$ defined by facet-defining inequalities of $\CP$. 

%\subsubsection{Block inequalities}

One of the families of valid inequalities presented in \cite{BCCOL} is the following. Given a vertex $v$ and a
color $j$, the \emph{$(v,j)$-block inequality} is $\sum_{k=j}^n x_{vk} \leq w_j$.

Let us observe that the $(v,1)$-block inequality is always satisfied by equality since every coloring $(x,w)$
verifies constraints (\ref{RASSIGN}) and $w_1 = 1$. Moreover, the $(v,2)$-block inequality defines the same
facet as inequality $x_{v1} \geq 0$. For the
remaining cases we have:

\begin{tthm} \label{TBLOCK}
% Let $G$ be a monotone graph.
Let $v \in V$ and $3 \leq j \leq n - 2$. The $(v,j)$-block inequality defines
a facet of $\ECP$ if and only if $G$ admits a $(j-1)$-eqcol.
\end{tthm}
\begin{proof}
Let $\F$ be the face of $\ECP$ defined by the $(v,j)$-block inequality.
To prove that $\F$ is a facet of $\ECP$ when $G$ admits a $(j-1)$-eqcol, we can use the same affinely independent
colorings proposed in the proof of Proposition 10 of \cite{BCCOL}, by imposing them to be equitable colorings. 
\\
Now, let us suppose that $G$ does not admit a $(j-1)$-eqcol. We will prove that every equitable coloring lying on
the face satisfies $x_{vj-1} = 0$. Let $(x,w)$ be a $k$-eqcol lying on $\F$.
If $k \leq j-2$, clearly $x_{vj-1} = 0$. Otherwise, $\sum_{k=j}^n x_{vk}=1$ since $k \neq j-1$, and then
$x_{vj-1} = 0$.
\end{proof}

%\subsubsection{Clique and 2-rank inequalities}
Let us consider other family of inequalities studied in \cite{BCCOL}.
Given $S \subset V$ and a color $j$, $\sum_{v\in S} x_{vj} \leq \alpha(S) w_j$ is valid for $\CP$.
The authors of \cite{BCCOL} proved that, by applying a lifting procedure on this inequality for $j \leq n-\alpha(S)$,
we can get
\[ \sum_{v \in S} x_{vj} + \sum_{v \in V} \sum_{k=n-\alpha(S)+1}^{n-1} x_{vk} \leq \alpha(S) w_j
	+ w_{n-\alpha(S)+1} - w_n. \]
We will refer to it as the $(S,j)$-\emph{rank inequality}. 

Let us remark that, if $S$ is not $\alpha(S)$-maximal, \ie if there exists $v\in V\backslash S$ such that $\alpha(S \cup \{v\}) = \alpha(S)$, the $(S,j)$-rank inequality is dominated by the
$(S \cup \{v\})$-rank inequality. Then, from now on, we only consider $(S,j)$-rank inequalities where $S$ is $\alpha(S)$-maximal. 

%Proposition 5 of \cite{BCCOL} gives sufficient conditions for these inequalities to define facets of $\CP$.
%In the case of $\ECP$, we noticed that these inequalities seldom define facets when $\alpha(S) \geq 3$.
%However, if $\alpha(S) \leq 2$, we found sufficient conditions for these inequalities to define facets of $\ECP$.

When $\alpha(S) = 1$, 
% are called clique inequalities.Given a maximal clique $Q$ and a color $j$, 
the $(S,j)$-rank inequality takes the form
$\sum_{v \in S} x_{vj} \leq w_j$ and is called \emph{$(S,j)$-clique inequality}.
If $|S| = 1$, \ie $S = \{v\}$ for some $v$, the $(S,j)$-clique inequality is dominated by
the $(v,j)$-block inequality. If $|S| \geq 2$, Propositions 5 and 6 of \cite{BCCOL} state that the
$(S,j)$-clique inequality defines a facet of $\CP$. The proof of these propositions can be easily
adapted to the equitable case allowing us to prove the following result.

\begin{tthm} \label{TCLIQUE}
% Let $G$ be a monotone graph.
Let $Q$ be a maximal clique of $G$ with $|Q| \geq 2$ and $j \leq n-1$.
The $(Q,j)$-clique inequality defines a facet of $\ECP$.
\end{tthm}

In Theorem \ref{T2RANK1} of \cite{APPENDIX} we give sufficient conditions for the $(S,j)$-rank inequalities
to define facets of $\ECP$ when $\alpha(S) = 2$.

Other valid inequalities can arise when $\alpha(S) = 2$. Let $Q$ be the set of vertices of $S$ that are universal
in $G[S]$, \ie $Q = \{ q \in S ~:~ S \subset N[q] \}$. If $Q$ is not empty, we may apply a different
lifting procedure that one used in \cite{BCCOL}, obtaining new valid inequalities for $\CP$ and $\ECP$:

\begin{tdef}
The \emph{$(S,Q,j)$-2-rank inequality} is defined for a given $S \subset V$ such that $S$ is 2-maximal,
$Q = \{ q \in S ~:~ S \subset N[q] \} \neq \varnothing$ and $j \leq n - 1$, as
\begin{equation} \label{R2RANK2}
\sum_{v \in S\backslash Q} x_{vj} + 2 \sum_{v \in Q} x_{vj} \leq 2 w_j.
\end{equation}
\end{tdef}

\begin{tlem} \label{T2RANK2VALID}
The $(S,Q,j)$-2-rank inequality is valid for $\ECP$.
\end{tlem}
\begin{proof}
If some vertex of $Q$ uses color $j$, no one else in $S$ can be painted with $j$.
Therefore, the value of the \lhs in (\ref{R2RANK2}) is at most 2 when color $j$ is used.
\end{proof}

If $|Q| = 1$, the $(S,Q,j)$-2-rank inequality is dominated by another valid inequality presented in the
next section (see Remark \ref{RANKDOMINATED}).

In Theorem \ref{T2RANK2} and Corollary \ref{T2RANK2COL} of \cite{APPENDIX}, we give sufficient conditions for the $(S,Q,j)$-2-rank
inequalities to define facets of $\ECP$ when $|Q| \geq 2$.

%to define facets
%of $\ECP$ when  are presented in Theorem \ref{T2RANK1} (in the appendix).\\
%Sufficient conditions for the $(S,Q,j)$-2-rank inequalities to define facets
%of $\ECP$ are presented in Theorem \ref{T2RANK2}, Corollary \ref{T2RANK2COL} and Remark \ref{T2RANK2CASEQ1} (in the appendix).

%Some valid inequalities called \emph{odd-hold inequalities} and \emph{path inequalities} are presented
%in \cite{COLL} for a coloring polytope. Then, in the same work, the authors reinforce these inequalities in order to
%obtain other ones called \emph{clique-wheel inequalities} and \emph{clique-fan inequalities}. In both cases, they
%searches vertices from an induced subgraph with a certain structure (odd-hold or path) that are universal in that
%subgraph and applies a lifting procedure to reach the new inequalities.\\
%We apply the same procedure given above to the rank inequality $\sum_{v \in S} x_{vj} \leq 2 w_j$ when
%$\alpha(S) = 2$ and a vertex $v \in S$ is universal in $G[S]$.

\section{New valid inequalities for $\ECP$}  \label{SNEWINEQ}

In this section, we present new families of valid inequalities for $\ECP$ which are not valid for $\CP$.

\subsection{Subneighborhood inequalities}

The \emph{neighborhood inequalities} defined in \cite{BCCOL} for each $u \in V$, \ie
$\alpha(N(u)) x_{uj}$ $+ \sum_{v \in N(u)} x_{vj} \leq \alpha(N(u)) w_j$,
are valid inequalities for $\CP$. Indeed, if $S \subset N(u)$, $\alpha(S) x_{uj} + \sum_{v \in S} x_{vj} \leq \alpha(S) w_j$
is valid for $\CP$. We can reinforce the latter inequality in the context of $\ECP$ to obtain:

\begin{tdef} \label{NEIGHBOR1DEF}
The \emph{$(u,j,S)$-subneighborhood inequality} is defined for a given $u \in V$, $S \subset N(u)$ such that
$S$ is not a clique and $j \leq n - 1$, as
\begin{equation} \label{RNEIGHBOR1}
\gamma_{jS} x_{uj} + \sum_{v \in S} x_{vj} +
\sum_{k = j+1}^n (\gamma_{jS} - \gamma_{kS}) x_{uk} \leq \gamma_{jS} w_j,
\end{equation}
where $\gamma_{kS} = \min \{\lceil n/\chi_{eq} \rceil, \lceil n/k \rceil, \alpha(S)\}$.
\end{tdef}

\begin{tlem} \label{TNEIGHBOR1VALID}
The $(u,j,S)$-subneighborhood inequality is valid for $\ECP$.
\end{tlem}
\begin{proof}
Let $(x,w)$ be an $r$-eqcol of $G$. If $r < j$, both sides of (\ref{RNEIGHBOR1})
are equal to zero. If $r \geq j$ and $x_{uj} = 1$, the value of the \lhs of (\ref{RNEIGHBOR1}) is exactly $\gamma_{jS}$. On the other hand,
if $x_{uj} = 0$, the term $\sum_{v \in S} x_{vj}$ contributes up to $\gamma_{rS}$ and the term
$\sum_{k = j+1}^n (\gamma_{jS} - \gamma_{kS}) x_{uk}$ contributes up to $\gamma_{jS} - \gamma_{rS}$ regardless the
color assigned to $u$. Hence, the \lhs does not exceed $\gamma_{jS}$ and (\ref{RNEIGHBOR1}) is valid.
\end{proof}

Subneighborhood inequalities always define faces of high dimension:

\begin{tthm} \label{TNEIGHBOR1DIM}
Let $\F$ be the face defined by the $(u,j,S)$-subneighborhood inequality.
Then,
\[ dim(\F) \geq %n^2 - \lceil n/2 \rceil - \chi_{eq} - |\S| + |S| - \delta(u) =
     dim(\ECP) - \bigl( \lceil n/2 \rceil -1 - |S| +  \delta(u) \bigr) = o(dim(\ECP)). \]
\end{tthm}
\begin{proof}
Let $s_1, s_2 \in S$ be non adjacent vertices and let $1 \leq r \leq \lceil n/2 \rceil-1$ such that
$r \neq j$. We propose at least $n^2 - \lceil n/2 \rceil - \chi_{eq} - |\S| + |S| - \delta(u) + 1$ affinely independent
colorings lying on $\F$. 
\begin{itemize}
\item A $n$-eqcol $c$ such that $c(u) = j$, $c(s_1) = n$ and $c(s_2) = r$.
\item %The $n$-eqcol
$swap_{n,j_1,j_2}(c)$ for each $j_1, j_2 \in \{1,\ldots,n-1\} \backslash \{j\}$ such that $j_1 \neq j_2$. % We have $(n-2)(n-3)$ colorings.
\item %The $n$-eqcol
$swap_{c(s),n,j}(c)$ for each $s \in S \backslash \{ s_1 \}$. %We have $|S|-1$ colorings.
\item %The $n$-eqcol
$swap_{n,j'}(c)$ for each $j' \in \{1,\ldots,n-1\}$. %We have $n-1$ colorings.
\item The $(n-1)$-eqcol $c'$ such that $c'(s_1) = r$ and $c'(i) = c(i) ~~\forall~i \in V \backslash \{s_1\}$. %($s_1$ and $s_2$ are not adjacent, so they can share color $r$).
\item %The $(n-1)$-eqcol
$swap_{j',r}(c')$ for each $j' \in \{1,\ldots,n-1\} \backslash \{j, r\}$. %We have $n-3$ colorings.
\item %The $(n-1)$-eqcol
$swap_{j,r,j'}(c')$ for each $j' \in \{1,\ldots,n-1\} \backslash \{j, r\}$ and,
if $j \leq \lceil n/2 \rceil-1$ then $j' \geq \lceil n/2 \rceil$. %We obtain at least $\lfloor n/2 \rfloor$ colorings.
\item The $(n-1)$-eqcol $c''$ such that $c''(s_1) = c(v)$, $c''(v) = j$ and
$c''(i) = c(i) ~~\forall~i \in V \backslash \{s_1,v\}$, % ($u$ and $v$ share color $j$),
for each $v \in V \backslash N[u]$. %We have $n - \delta(u) - 1$ colorings.
\item If $j \geq \chi_{eq} + 1$, an arbitrary $k$-eqcol of $G$ for each $k \in \{\chi_{eq},\ldots,j-1\} \backslash \S$.
\item $swap_{j,\hat{c}(u)}(\hat{c})$ where $\hat{c}$ is a $k$-eqcol of $G$, for each $k \in \bigl\{\max\{j, \chi_{eq}\},\ldots,n-2\bigr\} \backslash \S$.
%Then, we have $n - \chi_{eq} - |\S| - 1$ colorings.
\end{itemize}
The proof for the affine independence of the previous colorings is similar to the one for the colorings generated in Remark \ref{REMPROC}.
\end{proof}

Sufficient conditions for a $(u,j,S)$-subneighborhood inequality to be a facet-defining inequality of $\ECP$
are presented in Theorem \ref{TNEIGHBOR1} of \cite{APPENDIX} for the case $\lceil n/j \rceil \leq \lceil n/\chi_{eq} \rceil$ whereas
the following result allows us to study the inequality for the case $\lceil n/j \rceil > \lceil n/\chi_{eq} \rceil$.

%The following result asserts that it is enough to consider inequalities (\ref{RNEIGHBOR1}) that satisfy
%$\lceil n/j \rceil \leq \lceil n/\chi_{eq} \rceil$, since inequalities (\ref{RNEIGHBOR1}) such that
%$\lceil n/j \rceil > \lceil n/\chi_{eq} \rceil$ may be characterized by means of the
%$(u,\chi_{eq},S)$-neighborhood inequality. In other words, $(u,j,S)$-subneighborhood inequalities with
%$1 \leq j < \chi_{eq}$ are ``symmetric versions'' of the $(u,\chi_{eq},S)$-neighborhood inequality.

\begin{tthm} \label{TIGHT}
Let $j$ such that $\lceil n/j \rceil > \lceil n/\chi_{eq} \rceil$, $\F_j$ be the face defined by the
$(u,j,S)$-subneighborhood inequality and $\F_{\chi_{eq}}$ be the face defined by the
$(u,\chi_{eq},S)$-subneighborhood inequality. Then, $dim(\F_j) = dim(\F_{\chi_{eq}})$.
\end{tthm}
\begin{proof}
Clearly, if $\alpha(S) < \lceil n/\chi_{eq} \rceil$, both inequalities coincide. So, let us assume that
$\alpha(S) \geq \lceil n/\chi_{eq} \rceil$. Since $\lceil n/j \rceil > \lceil n/\chi_{eq} \rceil$,
$j < \chi_{eq}$ and $w_j = w_{\chi_{eq}} = 1$.
Then, both inequalities only differ in the coefficients of $x_{vj}$ and $x_{v\chi_{eq}}$ for all $v \in V$. Moreover,
the coefficient of $x_{vj}$ in the $(u,j,S)$-subneighborhood is the same as the one of $x_{v\chi_{eq}}$ in
the $(u,\chi_{eq},S)$-subneighborhood, and conversely.\\
Let $d = dim(\F_{\chi_{eq}})$ and $d' = dim(\F_j)$. If $c^1, c^2, \ldots, c^{d+1}$ are affinely independent equitable colorings
in $\F_{\chi_{eq}}$, colorings $swap_{j,\chi_{eq}}(c^i)$ for $1 \leq i \leq d+1$ are well defined and they
are affinely independent too. Moreover, they lie on $\F_j$. Therefore, $d \leq d'$.\\
To prove that $d' \leq d$, we follow the same reasoning. 
\end{proof}

\begin{trem} \label{RANKDOMINATED}
Let $j \leq n-1$, $S \subset V$ such that $\alpha(S) = 2$ and $Q = \{ v \in S ~:~ S \subset N[v] \} = \{q\}$.
The $(q,j,S\backslash\{q\})$-subneighborhood inequality is
\[ \sum_{v \in S\backslash\{q\}} x_{vj} + 2 x_{qj} + x_{qn} \leq 2 w_j, \]
and dominates the $(S,Q,j)$-2-rank inequality. 
In Corollary \ref{T2RANK2CASEQ1} of \cite{APPENDIX} we give sufficient conditions for it to be a facet-defining inequality of $\ECP$.
\end{trem}

\subsection{Outside-neighborhood inequalities}

\begin{tdef}
The \emph{$(u,j)$-outside-neighborhood inequality} is defined for a given $u \in V$ such that $N(u)$ is not a clique and $j \leq \lfloor n/2 \rfloor$, as
\begin{equation} \label{RNEIGHBOR2}
\biggl(\biggl\lfloor \dfrac{n}{t_j} \biggr\rfloor - 1 \biggr) x_{uj} - \sum_{v \in V \backslash N[u]} x_{vj}
+ \sum_{k = t_j+1}^n b_{jk} x_{uk} \leq \sum_{k = t_j+1}^n b_{jk} (w_k - w_{k+1}),
\end{equation}
where $t_j = \max \{j,\chi_{eq}\}$ and $b_{jk} = \lfloor n/t_j \rfloor - \lfloor n/k \rfloor$.
\end{tdef}

\begin{tlem}
The $(u,j)$-outside-neighborhood inequality is valid for $\ECP$.
\end{tlem}
\begin{proof}
Let $(x,w)$ be an $r$-eqcol of $G$. If $r < j$, both sides of (\ref{RNEIGHBOR2}) are equal to zero.
Let us assume that $r \geq j$ and $C_j$ denotes the color class $j$ of $(x,w)$. We divide the
proof in two cases:\\
\textbf{Case $r = t_j$}. The terms $\sum_{k = t_j+1}^n b_{jk} x_{uk}$ and $\sum_{k = t_j+1}^n b_{jk} (w_k - w_{k+1})$ vanish from the inequality so we
only need to check that $(\lfloor n/t_j \rfloor - 1) x_{uj} - \sum_{v \in V \backslash N[u]} x_{vj}$ is a non positive value.
If $x_{uj} = 0$, the inequality holds. If $x_{uj} = 1$,
$$\sum_{v \in V \backslash N[u]} x_{vj} = |C_j \backslash N[u]| \geq \lfloor n/t_j \rfloor - 1$$ and (\ref{RNEIGHBOR2}) holds.\\
\textbf{Case $r > t_j$}. We need to check that the \lhs of (\ref{RNEIGHBOR2}) is at most $b_{jr}$.
If $x_{uj} = 0$, then $\sum_{k = t_j+1}^n b_{jk} x_{uk} \leq \max \{b_{jk}: t_j+1\leq k \leq r\} = b_{jr}$ and the inequality holds.
If $x_{uj} = 1$, $\sum_{k = t_j+1}^n b_{jk} x_{uk} = 0$ and
$$\sum_{v \in V \backslash N[u]} x_{vj} = |C_j \backslash N[u]| \geq \lfloor n/r \rfloor - 1$$ and (\ref{RNEIGHBOR2}) holds.
\end{proof}

In order to study the faces of $\ECP$ defined by outside-neighborhood inequalities, let us characterize the equitable colorings that belong to those faces.

\begin{trem} \label{NEIGHBOR2POINTS}
Let $\F$ the face of $\ECP$ defined by the $(u,j)$-outside-neighborhood inequality and $c$ be an $r$-eqcol. Let us observe that if $r<j$, $c$ always lies on $\F$. For the case $r\geq j$, let $C_j$ be the color class $j$ of $c$. Then, $c$ lies on $\F$ if and only if the following conditions hold:
\begin{itemize}
\item If $c(u) = j$ then $|C_j| = \lfloor n/r \rfloor$.
\item If $c(u) \neq j$ then
\begin{itemize}
\item $C_j\subset N(u)$ and
\item if $\biggl\lfloor \dfrac{n}{r} \biggr\rfloor < \biggl\lfloor \dfrac{n}{\max \{j,\chi_{eq}\}} \biggr\rfloor$ then 
$c(u)\geq \biggl\lfloor \dfrac{n}{\lfloor n/r \rfloor + 1} \biggr\rfloor + 1$.
\end{itemize}
\end{itemize}
%Moreover, if $r \geq \biggl\lfloor \dfrac{n}{\lfloor n/t_j \rfloor} \biggr\rfloor + 1$ (or equivalently, $\lfloor n/r \rfloor < \lfloor n/t_j \rfloor$ then 
%$c(u)\geq \biggl\lfloor \dfrac{n}{\lfloor n/r \rfloor + 1} \biggr\rfloor + 1$.
\end{trem} 

Like the subneighborhood inequalities, outside-neighborhood inequalities define faces of high dimension:

\begin{tthm} \label{TNEIGHBOR2FACE}
Let $\F$ be the face defined by the $(u,j)$-outside-neighborhood inequality. Then,
\[ dim(\F) \geq % n^2 + \lceil n/2 \rceil - 3n + 3 + \delta(u) =
dim(\ECP) - \bigl( 3n - \lceil n/2 \rceil - |\S| - \chi_{eq} - 4 - \delta(u) \bigr) = o(dim(\ECP)). \]
\end{tthm}
\begin{proof}
Let $v_1 \in V \backslash N[u]$, $v_2, v_3 \in N(u)$ such that $v_2$ is not adjacent to $v_3$ and
$1 \leq r \leq \lfloor n/2 \rfloor$ such that $r \neq j$. We propose
$n^2 + \lceil n/2 \rceil - 3n + 4 + \delta(u)$ affinely independent solutions lying on $\F$:
\begin{itemize}
\item A $n$-eqcol $c$ such that $c(u) = j$, $c(v_1) = n$, $c(v_2) = n-1$ and $c(v_3) = r$.
\item %The $n$-eqcol
$swap_{n,j_1,j_2}(c)$ for each $j_1, j_2 \in \{1,\ldots,n-1\} \backslash \{j\}$ such that $j_1 \neq j_2$.
% We have $(n-2)(n-3)$ colorings.
\item %The $n$-eqcol
$swap_{n,j,c(v)}(c)$ for each $v \in N(u)$. % We have $\delta(u)$ colorings.
\item %The $n$-eqcol
$swap_{j,r,j'}(c)$ for each $j' \in \{\lfloor n/2 \rfloor + 1, \ldots, n-1\}$. % We have $\lceil n/2 \rceil - 1$ colorings.
\item %The $n$-eqcol
$swap_{n,j'}(c)$ for each $j' \in \{1,\ldots,n-1\} \backslash \{j\}$. % We have $n-2$ colorings.
\item The $(n-1)$-eqcol $c'$ such that $c'(v_1) = r$, $c'(v_3) = n-1$ and $c'(i) = c(i) ~~\forall~i \in V \backslash \{v_1, v_3\}$.
\item %The $(n-1)$-eqcol
$swap_{j',n-1}(c')$ for each $j' \in \{1, \ldots, n-2\}$. % We have $n-2$ colorings.
\item A $(n-2)$-eqcol $c''$ such that $c''(v_1) = c''(u) = n-2$ and $c''(v_2) = c''(v_3) = j$.
\end{itemize}
The proof for the affine independence of the previous colorings is similar to the one for the colorings generated in Remark \ref{REMPROC}.
\end{proof}

The following necessary condition for an outside-neighborhood inequality to define a facet of $\ECP$ will be helpful in the design of the separation routine.

\begin{tthm} \label{TNEIGHBOR2NEC}
If the $(u,j)$-outside-neighborhood inequality defines a facet of $\ECP$ then
$\alpha(N(u)) \geq \biggl\lfloor \dfrac{n}{\max \{j,\chi_{eq}\}} \biggr\rfloor$.
\end{tthm}
\begin{proof}
Let $t_j = \max \{j,\chi_{eq}\}$ and $\F$ be the face of $\ECP$ defined by the $(u,j)$-outside-neighborhood inequality.
Let us suppose that $\alpha(N(u)) < \lfloor n/t_j \rfloor$. We will prove that every equitable coloring lying on $\F$ also satisfies
the equality 
\begin{equation} \label{TNEIGHBOR2NECEQ}
\sum_{l=1}^{j-1} x_{ul} + w_j = 1.
\end{equation}
Since this equality can not be obtained as a linear combination of the minimal equation system for $\ECP$ and the $(u,j)$-outside-neighborhood \emph{equality},
$\F$ is not a facet of $\ECP$.\\
Let $c$ be an $r$-eqcol that lies on $\F$. Clearly, if $r < j$, $w_j = 0$ and $c(u) = l$ for some $1 \leq l \leq j-1$ and, consecuently, the
equality (\ref{TNEIGHBOR2NECEQ}) holds. If $r \geq j$, $w_j = 1$ and we have to prove that $\sum_{l=1}^{j-1} x_{ul}=0$, or equivalently, $c(u) \geq j$.
According to Remark \ref{NEIGHBOR2POINTS}, if $c(u) \neq j$ then $C_j \subset N(u)$ and thus $\alpha(N(u)) \geq |C_j|$. 
Observe that this fact implies that $\lfloor n/r \rfloor < \lfloor n/t_j \rfloor$. Indeed, if $\lfloor n/r \rfloor = \lfloor n/t_j \rfloor$,
$|C_j|\geq \lfloor n/t_j \rfloor$ and it contradicts the assumption $\alpha(N(u)) < \lfloor n/t_j \rfloor$.
Then, by Remark \ref{NEIGHBOR2POINTS}, $c(u) \geq \biggl\lfloor \dfrac{n}{\lfloor n/r \rfloor + 1} \biggr\rfloor + 1 > j$ and (\ref{TNEIGHBOR2NECEQ}) holds.
\end{proof}

For the case $j \geq \chi_{eq}$, we present sufficient conditions for the $(u,j)$-outside-neighborhood inequality to define a facet of $\ECP$ in Theorem \ref{TNEIGHBOR2} of \cite{APPENDIX}. For the other case, we have the following result whose proof follows the same ideas than in Theorem \ref{TIGHT}.
% \ie we use the fact that the $(u,j)$-outside-neighborhood
%and the $(u,\chi_{eq})$-outside-neighborhood are symmetric inequalities when $j < \chi_{eq}$.

\begin{tthm} \label{TIGHT2}
Let $j < \chi_{eq}$, $\F_j$ be the face defined by the $(u,j)$-outside-neighborhood inequality and
$\F_{\chi_{eq}}$ be the face defined by the $(u,\chi_{eq})$-outside-neighborhood inequality. Then,
$dim(\F_j) = dim(\F_{\chi_{eq}})$.
\end{tthm}

\subsection{Clique-neighborhood inequalities}

\begin{tdef} \label{NEIGHBOR3DEF}
The \emph{$(u,j,k,Q)$-clique-neighborhood inequality} is defined for a given $u \in V$ ,a clique $Q$ of $G$ such that
$Q \cap N[u]=\varnothing$ and numbers $j,k$ verifying $3 \leq k \leq \alpha(N(u)) + 1$ and
$1 \leq j \leq \biggl\lceil \dfrac{n}{k-1} \biggr\rceil - 1$, as

\begin{multline} \label{RNEIGHBOR3}
(k - 1) x_{uj} +
\sum_{l = \lceil \frac{n}{k-1} \rceil}^{n-2} \biggl(k - \biggl\lceil \dfrac{n}{l} \biggr\rceil \biggr) x_{ul} +
(k - 1) \bigl(x_{u n-1} + x_{un} \bigr) +
\!\!\!\sum_{v \in N(u)\cup Q}\!\!\!x_{vj} \\ + \sum_{v \in V \backslash \{u\}}\!\!\!(x_{v n-1} + x_{vn}) \leq
\sum_{l = j}^n b_{ul} (w_l - w_{l+1}),
\end{multline}
where
\[ b_{ul} = \begin{cases}
 \min\{\lceil n / l \rceil, \alpha(N(u)) + 1\}, &\textrm{if $j \leq l \leq \lceil n/k \rceil - 1$} \\
 k, &\textrm{if $\lceil n/k \rceil \leq l \leq n - 2$} \\
 k + 1, &\textrm{if $l \geq n - 1$}
\end{cases} \]
\end{tdef}

%%%%%%%%%%%%%%%%%%%%%%%%%%%%%
% ACORDARSE QUE LA CONDICION "\lceil n/k \rceil \leq \alpha(N(u)) + 1" PEDIDA EN LA DIMENSION ES NECESARIO YA QUE, SI NO, NO SE
% CUMPLE QUE b_ur - \lceil n/k \rceil >= 0 cuando r < k !!!
%%%%%%%%%%%%%%%%%%%%%%%%%%%%%

\begin{tlem}
The $(u,j,k,Q)$-clique-neighborhood inequality is valid for $\ECP$.
\end{tlem}

\begin{proof}
Let $(x,w)$ be an $r$-eqcol of $G$. If $r < j$, both sides of (\ref{RNEIGHBOR3}) are zero.
Let us assume that $r \geq j$ and observe that the \rhs of (\ref{RNEIGHBOR3}) is $b_{ur}$.
Let $C_j$, $C_{n-1}$ and $C_n$ be the color class $j$, $n-1$ and $n$ of $(x,w)$ respectively.
We divide the proof in the following cases:\\
\textbf{Case $r\leq \lceil n / k \rceil -1$}. We have to prove that $(x,w)$ verifies
\[ (k - 1) x_{uj} + \!\!\!\sum_{v \in N(u)\cup Q}\!\!\! x_{vj} \leq b_{ur} = \min \biggl\{ \biggl\lceil \dfrac{n}{r} \biggr\rceil, \alpha(N(u)) + 1 \biggr\}. \]
If $x_{uj}=1$, $\sum_{v \in N(u)} x_{vj}=0$ and $\sum_{v \in Q} x_{vj} \leq 1$. Since $b_{ur} \geq k$, the inequality holds.
If instead $x_{uj}=0$, $\sum_{v \in N(u)\cup Q} x_{vj} = |C_j \cap (N(u)\cup Q)| \leq \min \bigl\{\lceil n/r \rceil, \alpha(N(u) \cup Q)\bigr\} \leq
                \min \bigl\{\lceil n/r \rceil, \alpha(N(u)) + 1 \bigr\}$.\\
\textbf{Case $\lceil n / k \rceil \leq r \leq n-2$}. We have to prove that $(x,w)$ verifies
\[ (k - 1) x_{uj} +
\sum_{l = \lceil \frac{n}{k-1} \rceil}^{n-2} \biggl(k - \biggl\lceil \dfrac{n}{l} \biggr\rceil \biggr) x_{ul} +
\!\!\!\sum_{v \in N(u)\cup Q}\!\!\! x_{vj} \leq k. \]
If $x_{uj}=1$, $\sum_{l = \lceil \frac{n}{k-1} \rceil}^{n-2} (k - \lceil n/l \rceil) x_{ul} = 0$ and $\sum_{v \in N(u)\cup Q} x_{vj} \leq 1$. Therefore, the inequality holds.\\
If instead $x_{uj}=0$, $\sum_{l = \lceil \frac{n}{k-1} \rceil}^{n-2} (k - \lceil n/l \rceil) x_{ul} \leq k - \lceil n / r \rceil$ and
$\sum_{v \in N(u)\cup Q} x_{vj} \leq |C_j| \leq \lceil n / r \rceil$ and the inequality holds.\\
\textbf{Case $r \geq n-1$}. Let us first notice that $|C_j| + |C_{n-1}| + |C_n|\leq 3$. We have to prove that $(x,w)$ satisfies
\[ L(x) + \!\!\!\sum_{v \in N(u)\cup Q}\!\!\! x_{vj} + \sum_{v \in V \backslash \{u\}}\!\!\!(x_{v n-1} + x_{vn}) \leq k+1. \]
where
\[ L(x) = (k - 1) x_{uj} + \sum_{l = \lceil \frac{n}{k-1} \rceil}^{n-2} \biggl(k - \biggl\lceil \dfrac{n}{l} \biggr\rceil \biggr) x_{ul} + (k - 1) \bigl(x_{u n-1} + x_{un} \bigr). \]
Let us observe that $L(x) \leq k-1$ and $L(x) = k-1$ if and only if $u \in C_j \cup C_{n-1} \cup C_n$. 
Then, if $L(x) = k-1$, since $u \in C_j \cup C_{n-1} \cup C_n$ we have
$\sum_{v \in N(u)\cup Q} x_{vj} + \sum_{v \in V \backslash \{u\}} (x_{v n-1} + x_{vn}) \leq |C_j|+|C_{n-1}|+|C_n|-1\leq 2$,
and the inequality holds.\\
If $L(x) \leq k-2$, the inequality holds since $\sum_{v \in N(u)\cup Q} x_{vj} + \sum_{v \in V \backslash \{u\}} (x_{v n-1} + x_{vn}) \leq |C_j|+|C_{n-1}|+|C_n|\leq 3$.
\end{proof}

Let us remark that, if $Q$ is not maximal in $G - N[u]$, the $(u,j,k,Q)$-clique-neighborhood inequality is dominated by a
$(u,j,k,Q')$-clique-neighborhood, with $Q'$ a clique such that $Q \subsetneqq Q' \subset G-N[u]$.

%ESTO OCURRIA CON LA ujkQ ANTERIOR:
%It can be shown that $(u,j,k,Q)$-clique-neighborhood inequalities with $k \geq \lceil n/2 \rceil$ are dominated
%by other inequalities called $S$-color (see Remark \ref{REMARKONLY}).

In order to analyze the faces of $\ECP$ defined by clique-neighborhood inequalities, we first explore the colorings
that belong to those faces.

\begin{trem} \label{NEIGHBOR3POINTS}
Let $\F$ be the face of $\ECP$ defined by the $(u,j,k,Q)$-clique-neighborhood inequality and $c$ be an $r$-eqcol.
Let us observe that, if $r<j$, $c$ always lies on $\F$. For the case $r\geq j$, let $C_j$, $C_{n-1}$ and $C_n$ be the color class $j$, $n-1$ and $n$ of $c$ respectively. Then, $c$ lies on $\F$ if and only if the following conditions hold:
\begin{itemize}
\item If $r\leq \lceil n / k \rceil -1$ then:
\begin{itemize}
\item If $c(u) = j$ then $|C_j \cap Q| = 1$ and $k = \alpha(N(u)) + 1$.\\
Otherwise, $|C_j \cap (N(u)\cup Q)| = \min\{\lceil n / r \rceil, \alpha(N(u)) + 1\}$.
\end{itemize}
\item If $\lceil n / k \rceil \leq r \leq n-2$ then:
\begin{itemize}
\item If $c(u) = j$ then $|C_j \cap Q| = 1$. Otherwise,
\begin{itemize}
\item $|C_j \cap (N(u) \cup Q)| = \lceil n/r \rceil$ and
\item if $r \geq \biggl\lceil \dfrac{n}{k-1} \biggr\rceil$ then $c(u) \geq \biggl\lceil \dfrac{n}{\lceil n/r \rceil} \biggr\rceil$.
%Moreover, if $r \geq \biggl\lceil \dfrac{n}{k-1} \biggr\rceil$ then $c(u) \geq \biggl\lceil \dfrac{n}{\lceil n/r \rceil} \biggr\rceil$.
\end{itemize}
\end{itemize}
\item If $r \geq n-1$ then:
\begin{itemize}
\item If $c(u) \in \{ j, n-1, n \}$ then $|C_j \cap Q| + |C_{n-1} \backslash \{u\}| + |C_n \backslash \{u\}| = 2$.\\
Otherwise, $c(u) \geq \lceil n/2 \rceil$ and $|C_j \cap (N(u) \cup Q)| + |C_{n-1}| + |C_n| = 3$.
\end{itemize}
\end{itemize}
\end{trem}

Clique-neighborhood inequalities also define high dimensional faces in $\ECP$.

\begin{tthm}
Let $\F$ be the face defined by the $(u,j,k,Q)$-clique-neighborhood inequality. Then,
\[ dim(\F) \geq % n^2 + \lfloor n/2 \rfloor + 3 - 3n + \delta(u) + |Q| = 
   dim(\ECP) - \bigl( 3n - |\S| - \chi_{eq} - \lfloor n/2 \rfloor - \delta(u) - |Q| - 4 \bigr) = o(dim(\ECP)). \]
\end{tthm}
\begin{proof}
Let $v_1, v_2 \in N(u)$ be non adjacent vertices, and $q \in Q$.
We propose $n^2 + \lfloor n/2 \rfloor + 4 - 3n + \delta(u) + |Q|$ affinely independent solutions lying on $\F$:
\begin{itemize}
\item A $n$-eqcol $c$ such that $c(u) = j$ and $c(q) = n$.
\item %The $n$-eqcol
$swap_{n,j_1,j_2}(c)$ for each $j_1, j_2 \in \{1,\ldots,n-1\} \backslash \{j\}$ such that $j_1 \neq j_2$. % We have $(n-2)(n-3)$ colorings.
\item %The $n$-eqcol
$swap_{n,j,c(v)}(c)$ for each $v \in (N(u) \cup Q) \backslash \{q\}$. % We have $\delta(u) + |Q| - 1$ colorings.
\item %The $n$-eqcol
$swap_{j',j,n}(c)$ for each $j' \in \{\lceil n/2 \rceil, \ldots, n-1\}$. % We have $\lfloor n/2 \rfloor$ colorings.
\item %The $n$-eqcol
$swap_{j',n}(c)$ for each $j' \in \{1, \ldots, n-1\}$. % We have $n - 1$ colorings.
\item A $(n-1)$-eqcol $c'$ such that $c'(u) = j$, $c'(v_1) = c'(v_2) = n - 1$.
\item %The $(n-1)$-eqcol
$swap_{j,n-1}(c')$.
\item A $(n-2)$-eqcol $c''$ such that $c''(u) = c''(q) = j$ and $c''(v_1) = c''(v_2) = n - 2$.
\item %The $(n-2)$-eqcol
$swap_{j',n-2}(c'')$ for each $j' \in \{1, \ldots, n-3\} \backslash \{j\}$.
%We have $n - 4$ colorings.
\end{itemize}
The proof for the affine independence of the previous colorings is similar to the one for the colorings generated in Remark \ref{REMPROC}.
\end{proof}

%PARA LO QUE SIGUE HAY QUE SUPONER ADEMAS QUE $N[u] \cup Q \neq V$ PORQUE SINO YA DEFINE FACETA (ES UNA S-COLOR):
%(b)~We prove that the clique-neighborhood inequality can not define a facet if
%$k \geq \lceil n/2 \rceil$. So, let us suppose that $k \geq \lceil n/2 \rceil$ and
%let $\hat{v} \in V \backslash (N[u]\cup Q)$. Then, the \lhs of (\ref{RNEIGHBOR3}) is
%$\sum_{v \in N[u]\cup Q} x_{vj}$ + $\sum_{v \in V} x_{vn-1}$ + $\sum_{v \in V} x_{vn}$. It is easy to
%see that any equitable coloring $c$ such that $c(\hat{v}) = j$ makes the \lhs of (\ref{RNEIGHBOR3})
%be less than the \rhs Thus, every solution $(x,w)$ that satisfies (\ref{RNEIGHBOR3}) by equality also
%satisfies $x_{\hat{v}j} = 0$. Therefore (\ref{RNEIGHBOR3}) can not define a facet of $\ECP$.\\

Sufficient conditions for the clique-neighborhood inequalities to define facets of $\ECP$ are presented
in Theorem \ref{TNEIGHBOR3} and Corollary \ref{TNEIGHBOR3COR} of \cite{APPENDIX}.
%A condition that will be considered in the separation routine of clique-neighborhood inequalities is given
%by the following result.
%
%\begin{tthm}
%If the $(u,j,k,Q)$-clique-neighborhood defines a facet of $\ECP$ then $Q$ is maximal in $G - N[u]$.
%\end{tthm}
%\begin{proof}
%Let us suppose that $Q$ is not maximal in $G - N[u]$. Then, there exists a vertex $q \in V \backslash (N[u]\cup Q)$
%adjacent to every vertex of $Q$. Therefore, the $(u,j,k,Q)$-clique-neighborhood is dominated by the
%$(u,j,k,Q \cup \{q\})$-clique-neighborhood and can not define a facet of $\ECP$.
%\end{proof}

\subsection{$S$-color inequalities}

Given a set of colors $S$, let us analyze how many vertices can be painted with colors from $S$. %That means to obtain an upper bound of $\sum_{j \in S} |C_j|$ for a given coloring.
Let $(x,w)$ be a $k$-eqcol and $d_{Sk}$ be the number of colors in $S$ with non-empty color class in $(x,w)$, \ie $d_{Sk} = |S \cap \{1,\ldots,k\}|$. 
%\[ \sum_{j \in S} |C_j| = \!\!\!\sum_{j \in S \cap \{1,\ldots,k\}}\!\!\! |C_j| \leq d_{Sk} \biggl\lceil \dfrac{n}{k} \biggr\rfloor + \biggl(n - k \biggl\lfloor \dfrac{n}{k} \biggr\rceil \biggr) \]
It is straighforward to see that $(x,w)$ has $n - k \lfloor \frac{n}{k}\rfloor$ classes of size $\lfloor \frac{n}{k}\rfloor+1$ and $k - (n - k \lfloor \frac{n}{k}\rfloor)$ classes of size $\lfloor \frac{n}{k}\rfloor$. Then, the number of classes of color in $S$ having size $\lfloor \frac{n}{k}\rfloor+1$ is at most
$\min\{ d_{Sk}, n - k \lfloor \frac{n}{k}\rfloor \}$.
%Therefore,
%\[ \sum_{j \in S} |C_j| = \!\!\!\sum_{j \in S \cap \{1,\ldots,k\}}\!\!\! |C_j| \leq d_{Sk} \biggl\lfloor \dfrac{n}{k} \biggr\rfloor + \biggl(n - k \biggl\lfloor \dfrac{n}{k} \biggr\rfloor \biggr) \]
%where  
%$d_{Sk} = |S \cap \{1,\ldots,k\}|$. Let us observe that the following inequality is also verified by the coloring
%\[ \sum_{j \in S} |C_j| = \!\!\!\sum_{j \in S \cap \{1,\ldots,k\}}\!\!\! |C_j| \leq d_{Sk} \biggl( \biggl\lfloor \dfrac{n}{k} \biggr\rfloor + 1 \biggr) \]
Denoting by $b_{Sk} = d_{Sk} \lfloor \frac{n}{k} \rfloor + \min \{ d_{Sk}, n - k \lfloor \frac{n}{k}\rfloor \}$ we have that $\sum_{j \in S} |C_j| \leq b_{Sk}$, which motivates the following definition.
\begin{tdef}
Let $S \subset \{1,\ldots,n\}$. The \emph{$S$-color inequality} is defined as
\begin{equation} \label{RONLYCOLORS}
\sum_{j \in S} \sum_{v \in V} x_{vj} \leq \sum_{k=1}^{n} b_{Sk} (w_k - w_{k+1}),
\end{equation}
where $d_{Sk} = |S \cap \{1,\ldots,k\}|$ and
$b_{Sk} = d_{Sk} \lfloor \frac{n}{k} \rfloor + \min \{ d_{Sk}, n - k \lfloor \frac{n}{k}\rfloor \}$.
\end{tdef}

\begin{tlem}
The $S$-color inequality is valid for $\ECP$.
\end{tlem}
\begin{proof}
Let $(x,w)$ be a $k$-eqcol. If $k < j$, both sides of (\ref{RONLYCOLORS}) are zero.
If instead $k \geq j$, the \rhs of (\ref{RONLYCOLORS}) is $b_{Sk}$ which is an upper bound of
$\sum_{j \in S} |C_j| = \sum_{j \in S} \sum_{v \in V} x_{vj}$.
\end{proof}

\begin{trem} \label{REMARKONLY}
Let us present some useful facts about $S$-color inequalities.
\begin{enumerate}
\item Given $S \subset \{1,\ldots,n-1\}$, the $(S \cup \{n\})$-color inequality can be obtained by adding the $S$-color
inequality and equation (\ref{TDIM4}) from the minimal system. Then, both inequalities define the same face of $\ECP$.
\item Constraints (\ref{RLOWER}) and (\ref{RUPPER}) are both $S$-color inequalities with
$S = \{1,\ldots,n-1\}\backslash\{j\}$ and $S = \{j\}$ respectively.
\item It is not hard to see that the $(S,j)$-rank inequality with $\alpha(S) = 2$ and $j \geq \lceil n/2 \rceil$,
and (\ref{RNEIGHBOR3}) with $k = 2$ are both dominated by the $\{j,n-1\}$-color inequality.
\item If for every $k$ such that $G$ admits a $k$-eqcol, we have that either $k$ divides $n$ or
$n - k \lfloor \frac{n}{k}\rfloor \geq d_{Sk}$, then the $S$-color inequality is obtained by
adding constraints (\ref{RUPPER}), \ie $\sum_{v \in V} x_{vj} \leq \sum_{k=j}^n \lceil n / k \rceil (w_k - w_{k+1})$, for $j \in S$.
Thus, an $S$-color inequality can cut off a fractional solution of the linear relaxation of the formulation only if $2 \leq |S \backslash \{n\}| \leq n-3$
and there exists $k \in \{ \chi_{eq}, \ldots, n-1 \} \backslash \S$ such that $1 \leq n - k \lfloor \frac{n}{k}\rfloor \leq d_{Sk} - 1$.
\end{enumerate}
\end{trem}

The following result shows that $S$-color inequalities define faces of high dimension.

\begin{tthm} \label{TONLYCOLORSFACE}
Let $S \subset \{1,\ldots,n\}$ such that $|S \backslash \{n\}| \geq 1$ and let $\F$ be the face defined by the $S$-color inequality. Then,
\[ dim(\F) \geq %n^2 - \chi_{eq} - |\S| - n + |S \backslash \{n\}| =
             dim(\ECP) - (n - |S \backslash \{n\}| - 1) = o(dim(\ECP)).\]
\end{tthm}
\begin{proof}
From Remark \ref{REMARKONLY}.1 we can assume w.l.o.g.\! that $S \subset \{1,\ldots,n-1\}$.
Let $u_1, u_2$ be non adjacent vertices and $c$ be a $(n-1)$-eqcol such that $c(u_1) = c(u_2) = n-1$.
We consider colorings from Remark \ref{REMPROC} starting from $c$ and choosing those ones that lie in the
face defined by (\ref{RONLYCOLORS}).
That is, by excluding the $(n-1)$-eqcols that assign colors from $\{1,\ldots,n-1\} \backslash S$ to $u_1$ and
$u_2$ simultaneously, and by choosing $k$-eqcols where color classes from $S$ should have as many vertices
as possible, for each $k \in \{\chi_{eq},\ldots,n-2\} \backslash \S$.
Hence, we get $n^2 - \chi_{eq} - |\S| - n + 1 + |S|$ affinely independent colorings.
\end{proof}

%Below, necessary conditions for the $S$-color inequalities to define facets of $\ECP$ are presented. They
%will be helpful in the design of the separation routine.

%\begin{tthm} \label{TONLYCOLORSNEC}
%Let $S \subset \{1,\ldots,n-1\}$ such that $|S| \geq 1$. If the $S$-color inequality
%defines a facet of $\ECP$ then either $S = \{n-1\}$, or $|S| \geq 2$ and there exists a number
%$k \in \{ \chi_{eq}, \ldots, n-1 \} \backslash \S$ such that
%$1 \leq n - k \lfloor \frac{n}{k}\rfloor \leq d_{Sk} - 1$.
%\end{tthm}
%\begin{proof}
%If $|S| = 1$, the $S$-color inequality is equal to the constraint (\ref{RUPPER}), which is a facet-defining
%inequality if and only if $S = \{n-1\}$ as we stated in Theorem \ref{TORIGINAL4}. So, let us suppose that
%$|S| \geq 2$ and for all $k \in \{ \chi_{eq}, \ldots, n-1 \} \backslash \S$ either
%$k|n$ or $n - k \lfloor \frac{n}{k} \rfloor \geq d_{Sk}$. In both cases we have
%$b_{Sk} = d_{Sk} \lceil \frac{n}{k} \rceil$. Then, the $S$-color inequality can be obtained as the sum of
%$\{j\}$-color inequalities for all $j \in S$ and, consecuently, it does not define a facet of $\ECP$.
%\end{proof}

Finally, sufficient conditions for the $S$-color inequalities to define facets of $\ECP$ are presented
in Theorem \ref{TONLYCOLORS} of \cite{APPENDIX}.

\section{Implementation and computational experience} \label{SCOMPU}

We present computational results concerning the efficiency of valid
inequalities studied in the previous sections when they are used as cuts in
a cutting-plane algorithm for solving ECP.

The main elements of our implementation are described below.

%In this section, we experiment with valid inequalities presented in this
%work (besides the clique inequalities) and gather information about their
%behaviour and its impact in the performance when they are used as cuts.

\subsection{Initialization}

%As we mentioned before, we consider two possibilities of symmetry breaking
%in the integer programming formulation.
%From the polyhedral perspective, we have already pointed out the intractability
%of using constraints (\ref{RREPR1})-(\ref{RREPR2}) in our study.
%However, since the polytope that contemplates those constraints is included in
%$\ECP$, valid inequalities arising from our study are also valid for it.
%Moreover, according to our computational experiments \cite{MACI}, the
%formulation that uses those constraints performs better. Then, we
%decided to include constraints (\ref{RREPR1})-(\ref{RREPR2}) in the
%formulation used for experimentation. Indeed, we delete variables $x_{vj}$
%such that $v < j$. We also consider inequalities (\ref{RREPR2}) as lazy constraints.
%This means they are not part of the initial relaxation, but they are added later as cuts
%whenever necessary.

According to our computational experience reported in \cite{MACI}, the ILP
formulation of ECP consisting of constraints (\ref{RASSIGN})-(\ref{RUPPER}) performs
much better than the one defining $\ECP$, \ie without (\ref{RREPR1})-(\ref{RREPR2}).
Since every valid inequality of $\ECP$ is also valid for equitable colorings satisfying
constraints (\ref{RASSIGN})-(\ref{RUPPER}), we use this tighter formulation for
computational experiments, with inequalities (\ref{RREPR2}) handled as lazy
constraints in the implementation. This means they are not part of the initial relaxation,
but they are added later as cuts whenever necessary.

%Let us note that, from the above, the labeling of the vertices matters. We tested
%several criteria for labels the vertices and the criterion which has proved
%to do well in practice consists of assigning numbers $1,\ldots,q$ to vertices of
%a maximal clique in $G$, where $q$ is the size of that clique, and assigning
%numbers $q+1,\ldots,n$ to the remaining vertices in a decreasing order of
%degree, \ie satisfying $\delta(v) \geq \delta(v+1)$ for all $q+1 \leq v \leq n$.

We tested several criteria for labeling vertices and the one which has
proved to be the best in practice is the following. We first find a maximal
clique $Q$. Denoting by $q$ the size of $Q$, we assign the first $q$ natural numbers to
vertices of $Q$. The labels of remaining vertices are assigned in decreasing order of
degree, \ie satisfying $\delta(v) \geq \delta(v+1)$ for all $v \in \{ q+1, \ldots, n \}$.

%At the initial phase, we run an heuristic called \emph{Naive} \cite{KUBALE} to
%find an upper bound $\overline{\chi_{eq}}$ of the equitable chromatic number.

To find an initial upper bound of $\overline{\chi_{eq}}$, we use the
heuristic \emph{Naive} presented in \cite{KUBALE}. This allows us to eliminate variables
$x_{vj}$ and $w_j$ with $j > \overline{\chi_{eq}}$ from the model.

%Moreover, we compute a lower bound of the equitable chromatic number with the following
%formula, proposed in KUBALE:
%\[ r = \max \biggl\{ \biggl\lceil \dfrac{n + 1}{\hat{\theta}(G - N[v]) + 2}
%        \biggr\rceil : v \in V \biggr\}, \]
%where $\hat{\theta}(G)$ is the cardinal of a clique partition of $G$ we find greedily.
%Since $q$, \ie the size of the maximal clique in $G$ used in the labeling of vertices, is also a lower bound of
%$\chi_{eq}$, we consider $\underline{\chi_{eq}} = \max \{q, r\}$.

In addition, a lower bound $\underline{\chi_{eq}}$ is obtained by considering the maximum between the
size of the maximal clique $Q$ and the value
\[ \max \biggl\{ \biggl\lceil \dfrac{n + 1}{\theta(G - N[v]) + 2}
        \biggr\rceil : v \in V \biggr\}, \]
also proposed in \cite{KUBALE}, where $\theta(G)$ is the cardinal of a clique partition of
$G$ found greedily.

We compute bounds of the stability number of $N(u)$ for all $u \in V$ (via heuristic procedures),
which will be useful for the separation routines. We denote the upper bound as $\overline{\alpha}(N(u))$
and the lower bound as $\underline{\alpha}(N(u))$.

%In \cite{ALIO,MACI}, we showed that a cutting-plane algorithm based on clique inequalities is an
%effective way of strengthening the linear relaxation of the problem. In this section, we experiment
%with valid inequalities presented in this work (besides the clique inequalities) and gather information about their
%behaviour and its impact in the performance when they are used as cuts.

%Let us call $ECF$ to the formulation defined in Section \ref{SPOLYT} which consists of constraints
%(\ref{RASSIGN})-(\ref{RORDER}) and (\ref{RISOL})-(\ref{RUPPER}), and $ECF'$ to the formulation
%consisting of constraints (\ref{RASSIGN})-(\ref{RUPPER}). It is clear that $ECF'$ is a well-defined
%formulation for ECP. Moreover, a valid inequality of $\ECP$ remains valid in the integer polyhedron
%associated to $ECF'$ since it is contained in $\ECP$.

%Preliminary experiments indicate that $ECF'$ performs better than $ECF$ \cite{MACI}. Since $ECF'$ is
%sensitive to the labeling of vertices, we have to settle a criterion. The following has proved to do
%well in practice: let assign numbers $1,\ldots,q$ to vertices of a maximal clique in $G$, where $q$ is the
%size of that clique, and assign numbers $q+1,\ldots,n$ to the remaining vertices in a decreasing order of
%degree, \ie satisfying $\delta(v) \geq \delta(v+1)$ for all $q+1 \leq v \leq n$.

%Our experiments will use $ECF'$ with the labeling criterion given before. 

\subsection{Description of the cutting-plane algorithm} \label{SSEPAR}

%The cutting-plane stage is related to a piece of software embeeded within a Branch \& Cut algorithm whose main
%purpose is to strengthen the current linear relaxation through the addition of cutting planes.
%Here we describe the design of the cutting-plane algorithm and the separation routines behind it.
The design of the separation routines for each family of valid inequalities is described below.
%First of all, we must be aware that the separation routines should not spend a significant amount of
%time. That is why we decided to implement simple and fast heuristics. In some cases, certain
%hard-to-compute parameters (such as $\chi_{eq}$) must be known, and since we do not know the exact
%value of them, lower and upper bounds take its place to perform the separation.
%Such is the case of separation routines of subneighborhood, outside-neighborhood and clique-neighborhood,
%which require to known $\alpha(N(u))$ for a given vertex $u$. So, before starting the optimization, we
%heuristically compute bounds of this parameter for each vertex $u$.\\
Given a fractional solution $(x^*, w^*)$ of the linear relaxation, we look for violated inequalities as follows:
\begin{itemize}
\item \emph{Clique and Block inequalities}. They are handled in the same way as in \cite{BCCOL}.
\item \emph{Clique-neighborhood inequalities}. For each maximal clique $Q$ we found during the clique separation procedure and for each
$u \in V \backslash \bigl(\cup_{q \in Q} N[q]\bigr)$, $j \in \{1, \ldots, \overline{\chi_{eq}}\}$ and $k$ such that
$$\max \{3,\lceil n/\overline{\chi_{eq}} \rceil\} \leq k \leq
    \min \bigl\{\lceil n/j \rceil, \lceil n/\underline{\chi_{eq}} \rceil, \overline{\alpha}(N(u)) + 1 \bigr\},$$
we verify whether $(x^*,w^*)$ violates a weaker version of the $(u,j,k,Q)$-clique-neighborhood which consists of
 replacing $\alpha(N(u))$ by $\overline{\alpha}(N(u))$ to compute $b_{ul}$ in Definition \ref{NEIGHBOR3DEF}.
\item \emph{2-rank inequalities}. For each $j \in \{1,\ldots,\overline{\chi_{eq}}\}$, we find a pair of vertices $v_1$ and $v_2$ such that $x^*_{v_1 j} + x^*_{v_2 j}$
has the highest value, but less than 1, and we initialize $S = \{v_1, v_2\}$ and $Q = \varnothing$. Then,
we fill sets $S$ and $Q$ by adding vertices, one by one, with the following rule.
Let $v$ be a vertex with largest fractional value of $x^*_{vj}$, adjacent to every vertex of $Q$ and such that
$S \cup \{v\}$ is 2-maximal. If $S \subset N[v]$ we add $v$ to the set $Q$. Otherwise, we add it to $S$.
When it is not possible to add more vertices to $S$ or $Q$, we check whether the $(S,Q,j)$-2-rank inequality
cuts off $(x^*,w^*)$.\\
We also implement an additional mechanism that prevents from generating violated cuts with similar support.
Each time a $(S,Q,j)$-2-rank inquality is found (not necessarily violated by the fractional solution), we mark every
vertex of $S$ as \emph{forbidden}, to mean that those vertices can not take part of upcoming $(S,Q,j)$-2-inequalities.
The procedure is performed over and over, until not more than 5 vertices are not forbidden. Then, we unmark all the
forbidden vertices and start over with the next value of $j$.
\item \emph{$S$-color inequalities}. We first find $t$ such that $0 < w_t < 1$ and $w_{t+1} = 0$. If $t$
does not exist (meaning that $w^* \in \mathbb{Z}^n$), we do not generate any cut. Otherwise, we order in
decreasing way the color classes $j \in \{ 1, \ldots, t \}$ according to the number of fractional
variables $x^*$, \ie $|\{ v : x^*_{vj} \notin \mathbb{Z}^n ~\forall~v \in V \}|$. Then, for each $s \in \{ 2, \ldots, t-2 \}$ such that
$$s \geq 1 + \min \{ n - k \lfloor n/k \rfloor : k \in \{ 1,\ldots, t\} \land \textrm{$k$ does not divide $n$} \}$$
(see Remark \ref{REMARKONLY}.4), we scan $S$-color inequalities with $|S| = s$ and $S$ having the most fractional classes,
looking for the inequality that maximizes violation. Once the best $S$-color inequality is identified we check whether
it cuts off $(x^*,w^*)$.\\
The procedure given before allows us to produce only one inequality. In order to generate more inequalities we do the
following. Each time a $S$-color inequality is identified (regardless of the inequality is violated or not), we
mark one color class belonging to $S$ as \emph{forbidden}, to mean that it can not take part of upcoming $S$-color
inequalities. Then, we repeat the procedure until only two color classes are not forbidden.
\item \emph{Subneighborhood and Outside-neighborhood inequalities}.
They are handled by enumeration: for each $j \in \{1,\ldots,\overline{\chi_{eq}}\}$ and $u$
such that $\underline{\alpha}(N(u))$ $\geq 3$ (because vertices $u$ with $\alpha(N(u)) \leq 2$
lead to clique and 2-rank cuts), we check whether $(x^*,w^*)$ violates a weaker version of these inequalities,
defined as follows. For the subneighborhood inequalities, we compute
$\xi_{k} = \min \{ \lceil n/\underline{\chi_{eq}} \rceil, \lceil n/k \rceil, \overline{\alpha}(N(u))\}$,
and then we consider inequalities of the form:
\[ \xi_j x_{uj} + \sum_{v \in N(u)} x_{vj} + \sum_{k = j+1}^n (\xi_j - \xi_k) x_{uk} \leq \xi_j w_j. \]
For the outside-neighborhood inequalities, we first check the condition of Theorem \ref{TNEIGHBOR2NEC}, \ie
$\underline{\alpha}(N(u)) \geq \lfloor n/\max \{j,\underline{\chi_{eq}}\} \rfloor$ and then we use the inequality that
results from replacing $t_j$ with $\max \{ j, \underline{\chi_{eq}} \}$ in (\ref{RNEIGHBOR2}).
\end{itemize}

\subsection{Performance of cuts at root node} \label{SCOMPU2}

In order to evaluate the quality of a cutting-plane algorithm, we analyze the increase of the lower bound
when cuts are added progressively to the LP-relaxation.

%An indirect way to assess the quality of a cutting-plane algorithm is to analyze the ratio between the
%growth of the lower bound when cuts are added progressively to the LP-relaxation and the amount of
%cuts added. We know that greater increases of that ratio mean added constraints have better
%quality. On the other hand, the greater number of cuts are added, the heavier the further
%relaxation is in terms of CPU time, so we must be careful not to incorporate too many cuts that hamper
%the computation of relaxations.

In this experiment, we compare the performance of seven strategies given in Table \ref{strategies}, where each one is a
combination of separation routines that determine the behaviour of the cutting-plane algorithm.

\begin{table}[h]
\begin{center} \footnotesize
\begin{tabular}{|@{\hspace{3pt}}c@{\hspace{3pt}}|@{\hspace{3pt}}c@{\hspace{3pt}}|@{\hspace{3pt}}c@{\hspace{3pt}}|@{\hspace{3pt}}c@{\hspace{3pt}}|@{\hspace{3pt}}c@{\hspace{3pt}}|@{\hspace{3pt}}c@{\hspace{3pt}}|@{\hspace{3pt}}c@{\hspace{3pt}}|@{\hspace{3pt}}c@{\hspace{3pt}}|}
\hline
Strategy & Clique & 2-rank & Block & S-color & Sub- & Outside- & Clique- \\
Name & & & & & neighbor. & neighbor. & neighbor. \\
\hline
S1 & $\bullet$ & & & & & & \\
S2 & $\bullet$ & $\bullet$ & & & & & \\
S3 & $\bullet$ & $\bullet$ & $\bullet$ & & & & \\
S4 & $\bullet$ & $\bullet$ & $\bullet$ & $\bullet$ & & & \\
S5 & $\bullet$ & $\bullet$ & $\bullet$ & $\bullet$ & $\bullet$ & & \\
S6 & $\bullet$ & $\bullet$ & $\bullet$ & $\bullet$ & $\bullet$ & $\bullet$ & \\
S7 & $\bullet$ & $\bullet$ & $\bullet$ & $\bullet$ & $\bullet$ & $\bullet$ & $\bullet$ \\
\hline
\end{tabular}
\end{center}
  \caption{Strategies}
  \label{strategies}
\end{table}
%Notice that we only add block, subneighborhood, outside-neighborhood and $S$-color
%inequalities only if few clique and 2-rank inequalities are found (in the last strategy).
%This decision is based on the fact that block, subneighborhood, outside-neighborhood and
%$S$-color do not produce significant changes in the objective function but they allow
%to explore fractional solutions that provide feedback to the separation routine for clique and 2-rank
%inequalities yielding these latter cuts in a greater amount. 

The experiment was carried out on a server equipped with an Intel i5 2.67 Ghz over Linux Operating
System.
%\footnote{To allow a meaningful (although approximate) comparison of our results against other enviroments, the time spent by the program \emph{dfmax}, used for benchmarking purposes in the literature, is 6.31 seconds (user time) for the benchmark instance \emph{r500.5.b}. Both are available at \texttt{http://mat.gsia.cmu.edu/COLOR03}.}.
The server also has the well known general-purpose IP-solver CPLEX 12.2 which is used for solving linear
relaxations. We consider 50 randomly generated graphs with 150 vertices and different densities of
edges.
%\footnote{Given a graph $G = (V, E)$, the percentage of density of $G$ is $\dfrac{100|E|}{|V|(|V| - 1)/2}$.}.
For each graph and each strategy, we ran 30 iterations of the cutting-plane algorithm.

In order to compare the strategies involved, we call $LB_i$ to the objective value of the linear relaxation after
the $i$th iteration and we compute:
\begin{itemize}
\item Improvement in the lower bound, \ie the difference between the lower bound of $\chi_{eq}$ obtained after and
before the execution of the cutting-plane algorithm: $Impr = \lceil LB_{30} \rceil - \lceil LB_0 \rceil$.
%\item Minimum number of iterations needed to reach the best lower bound: $Iter = \min \{ i : \lceil LB_i \rceil = \lceil LB_{30} \rceil \}$.
\item Time elapsed up to reach the best lower bound, \ie at iteration $\min \{ i : \lceil LB_i \rceil = \lceil LB_{30} \rceil \}$. We denote it as $Time$.
\item Number of cuts generated up to reach the best lower bound. We denote it as $Cuts$.
\end{itemize}

For graphs having 10\% of density, all the strategies showed no improvement in the lower bound. For graphs having
at least 30\% of density, all the strategies except S1 reaches the same bound in every instance, while S1 attains worse bounds.
In Figure \ref{IMPRFIG}, we display the average of $Impr$ over instances having the
same density. 

\begin{figure}[h]
  \begin{center}
    \includegraphics[width=9cm]{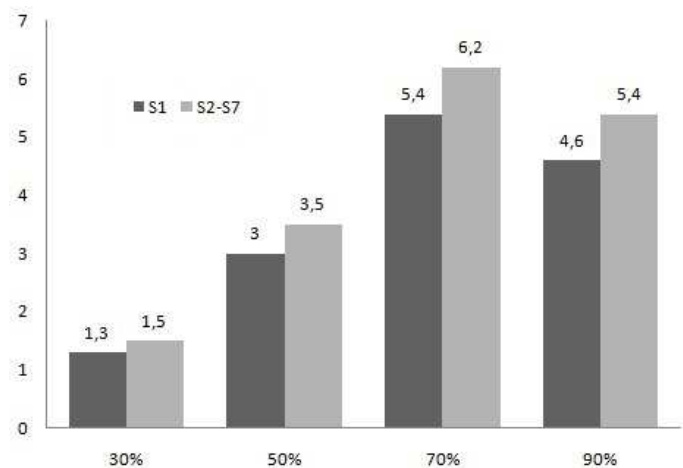}
  \end{center}
  \caption{Average of $Impr$ for strategies S1 and S2-S7}
  \label{IMPRFIG}
\end{figure}

As we have mentioned, strategies S2-S7 reached the same bound in every instance. One way to
tie them is by inspecting the average of $Time$, \ie the time elapsed, and $Cuts$, \ie
the number of cuts generated. The smaller $Time$ is, the sooner the algorithm reaches the best bound.
On the other hand, the less $Cuts$ is, the better the quality of the cuts involved are.
Table \ref{TABLECUTS} resumes these results. Best values are emphasized with bold fond.

\begin{table}[h]
\begin{center}
\begin{tabular}{c|c@{\hspace{4pt}}c@{\hspace{4pt}}c@{\hspace{4pt}}c@{\hspace{4pt}}c@{\hspace{4pt}}c|c@{\hspace{4pt}}c@{\hspace{4pt}}c@{\hspace{4pt}}c@{\hspace{4pt}}c@{\hspace{4pt}}c}
\%Density & \multicolumn{6}{c|}{Time} & \multicolumn{6}{c}{Cuts} \\
Graph & S2 & S3 & S4 & S5 & S6 & S7 & S2 & S3 & S4 & S5 & S6 & S7 \\
\hline
30 & 77 & \textbf{75} & \textbf{74} & 82 & 82 & 98 & \textbf{2034} & 2053 & 2053 & 2093 & 2093 & 3203 \\
50 & \textbf{241} & 248 & 248 & 267 & 267 & 252 & \textbf{3694} & 3796 & 3796 & 4065 & 4065 & 3944 \\
70 & 648 & \textbf{601} & 632 & 700 & 738 & 735 & 6182 & 5805 & \textbf{5670} & 6306 & 6405 & 6377 \\
90 & 720 & 763 & \textbf{612} & 658 & \textbf{610} & \textbf{610} & 5443 & 5493 & \textbf{5065} & 5187 & 5143 & 5143
\end{tabular}
\end{center}
  \caption{Average of $Time$ and $Cuts$ for strategies S2-S7}
  \label{TABLECUTS}
\end{table}
%\begin{table}[h]
%\begin{center}
%\begin{tabular}{c|c@{\hspace{4pt}}c@{\hspace{4pt}}c@{\hspace{4pt}}c@{\hspace{4pt}}c@{\hspace{4pt}}c|c@{\hspace{4pt}}c@{\hspace{4pt}}c@{\hspace{4pt}}c@{\hspace{4pt}}c@{\hspace{4pt}}c}
%\%Density & \multicolumn{6}{c}{Iter} & \multicolumn{6}{c}{Cuts} \\
%Graph & S2 & S3 & S4 & S5 & S6 & S7 & S2 & S3 & S4 & S5 & S6 & S7 \\
%\hline
%30 & \textbf{7,4} & \textbf{7,4} & \textbf{7,4} & 7,5 & 7,5 & 7,6 & \textbf{2034} & 2053 & 2053 & 2093 & 2093 & 3203 \\
%50 & 12,7 & 13,0 & 13,0 & 13,0 & 13,0 & \textbf{12,4} & \textbf{3694} & 3796 & 3796 & 4065 & 4065 & 3944 \\
%70 & 20,8 & 19,4 & \textbf{19,2} & 19,6 & 19,8 & 19,7 & 6182 & 5805 & \textbf{5670} & 6306 & 6405 & 6377 \\
%90 & 20,2 & 20,0 & 19,5 & 19,8 & \textbf{19,4} & \textbf{19,4} & 5443 & 5493 & \textbf{5065} & 5187 & 5143 & 5143 \\
%\hline
%\end{tabular}
%\end{center}
%  \caption{Average of $Iter$ and $Cuts$ for strategies S2-S7}
%  \label{TABLECUTS}
%\end{table}

As we can see from Table \ref{TABLECUTS}, strategy S4 reaches the best lower bound with fewer cuts for graphs having
at least 70\% of density and the amount of cuts generated is relatively acceptable for graphs having at most 50\% of
density. Strategy S4 also has the best balance between number of cuts generated and time consumed. Therefore, this strategy is a good candidate for our cutting-plane algorithm.
% largest increase in the lower bound is produced
%by the strategy that generates all the cuts (except for graphs with low density), and the total
%amount of cuts generated is relatively small compared with the size of the initial relaxation.
%It is known that random graphs of medium density are harder to solve than others, so we should pay more
%attention to the results of graphs with 50\% of density. In our case, there is at least 3\% of
%difference in the average between the strategy ``ALL'' and any of the others.

%We conclude that it is worth considering as cuts the new inequalities obtained by the polyhedral study.  

From the previous results we conclude that the cuts obtained from the polyhedral study are indeed effective. They
appear to be strong in practice, increasing significantly the initial lower bound.

%\endinput

Nevertheless, the long-term efficiency of cuts can not be appreciated here and require further
experimentation. This topic is covered in the next section.

%\subsection{Comparison between different strategies in a Cut and Branch algorithm} \label{SCOMPU3}
\subsection{Long-term efficiency of cuts} \label{SCOMPU3}

The purpose of the following experiment is to compare the Branch and
Bound (B\&B) algorithm of CPLEX with a Cut and Branch.
The algorithm consists of applying 30 iterations of the cutting-plane algorithm to the initial
relaxation. Then, we run a Branch and Bound enumeration until the
optimal solution is found or a time limit of 2 hours is reached.

In order to do that, we apply both algorithms to 40 randomly generated graphs with different number of vertices
and densities of edges. Since instances having 10\% and 90\% of density are easier to solve, we increased the
number of vertices of them. %We carried out the runs in the same server as the one employed before.
%Preliminary experiments showed that S2, S3, S5 and S6 have the worse performance.

Preliminary experiments showed that strategies S2-S6 have a similar behaviour each other, although S4 presents the best performance among them.
This led us to deepen the analysis of strategies S1, S4 and S7. Table \ref{TABLERANDOM} reports:
\begin{itemize}
\item Percentage of solved instances within 2 hours of execution.
%\item Average of percentage of relative gap over all instances, where the relative gap is
%$100(UB - LB)/LB$ and $LB$, $UB$ are the best lower and upper bound of $\chi_{eq}$ achieved.
\item Average of nodes evaluated over solved instances.
\item Average of total CPU time in seconds over solved instances.
\end{itemize}

\begin{table} \footnotesize
\begin{center}
\begin{tabular}{c@{\hspace{3pt}}c|c@{\hspace{3pt}}c@{\hspace{3pt}}c@{\hspace{3pt}}c|c@{\hspace{3pt}}c@{\hspace{3pt}}c@{\hspace{3pt}}c|c@{\hspace{3pt}}c@{\hspace{3pt}}c@{\hspace{3pt}}c}
Num. of & \%Density & \multicolumn{4}{c|}{\% Solved} & \multicolumn{4}{c|}{Nodes} & \multicolumn{4}{c}{Time} \\
Vertices & Graph & B\&B & S1 & S4 & S7 & B\&B & S1 & S4 & S7 & B\&B & S1 & S4 & S7 \\
\hline
90 & 10 & 100 & 100 & 100 & 100 & 2933 & 3050 & \textbf{1718} & \textbf{1718} & 33 & 33 & \textbf{21} & \textbf{21} \\
60 & 30 & 100 & 100 & 100 & 100 & 7515 & 2976 & \textbf{1050} & 6567 & 129 & 52 & \textbf{35} & 130 \\
60 & 50 & 100 & 100 & 100 & 100 & 29490 & 20639 & 21232 & \textbf{15786} & 974 & 1065 & \textbf{812} & \textbf{812} \\
60 & 70 & 87.5 & \textbf{100} & \textbf{100} & \textbf{100} & 19811 & 12891 & \textbf{5330} & 6454 &	734 & 508 & \textbf{327} & 340 \\
90 & 90 & 62.5 & 62.5 & \textbf{100} & \textbf{100} & 52545 & 35538 & \textbf{12645} & 15536 & 2332 & 2404 & \textbf{689} & 1088
\end{tabular}
\end{center}
  \caption{Performance of different strategies}
  \label{TABLERANDOM}
\end{table}
%\begin{table} \footnotesize
%\begin{center}
%\begin{tabular}{c@{\hspace{3pt}}c|c@{\hspace{3pt}}c@{\hspace{3pt}}c|c@{\hspace{3pt}}c@{\hspace{3pt}}c|c@{\hspace{3pt}}c@{\hspace{3pt}}c|c@{\hspace{3pt}}c@{\hspace{3pt}}c}
%Num. of & \%Density & \multicolumn{3}{c|}{\% Solved} & \multicolumn{3}{c|}{\% Rel.Gap}
%          & \multicolumn{3}{c|}{Nodes} & \multicolumn{3}{c}{Time} \\
%Vertices & Graph & S1 & S4 & S7 & S1 & S4 & S7 & S1 & S4 & S7 & S1 & S4 & S7 \\
%\hline
%90 & 10 & 100 & 100 & 100 & 0 & 0 & 0 & 2608 & \textbf{1819} & \textbf{1819} & 29 & \textbf{24} & \textbf{24} \\
%60 & 30 & 100 & 100 & 100 & 0 & 0 & 0 & 6568 & \textbf{6208} & 7034 & \textbf{130} & 177 & 149 \\
%60 & 50 & 100 & 100 & 100 & 0 & 0 & 0 & 25484 & 26529 & \textbf{22268} & 1088 & \textbf{998} & 1024 \\
%60 & 70 & 100 & 100 & 100 & 0 & 0 & 0 & 10330 & \textbf{4618} & 5187 & 409 & \textbf{272} & 275 \\
%90 & 90 & 60 & \textbf{100} & \textbf{100} & 0.8 & \textbf{0} & \textbf{0} & 34445 & \textbf{20140} & 22453 & 2418 & \textbf{1206} & 1512 \\
%\end{tabular}
%\end{center}
%  \caption{Performance of different strategies}
%  \label{TABLERANDOM}
%\end{table}

%While strategy S1 outperforms the others in graphs having 30\% of density, it is not able to solve some instances having 90\% of density.
%Instead, strategy S4 shows an improvement against S1 of around 20\%, 9\%, 50\% and 100\% in graphs having 10\%, 50\%, 70\% and 90\% of density respectively.

The new inequalities show again a substantial improvement and, in particular, strategy S4 is established as the best
one. It is worth mentioning that strategy S7 evaluated fewer nodes than S4 when solving instances of 50\% of density, but this reduction on the number of nodes
was not enough to counteract the CPU time elapsed.

From the latter results we conclude that the new inequalities used as cuts are good enough to
be considered as part of the implementation of a further competitive Branch and Cut algorithm that solves the ECP.

%\subsection{Conclusions and further works}
%
%*** ACA VAN LAS CONCLUSIONES Y FURTHER WORKS ***

%It can be seen an incremental improvement among the algorithms CPX$_1$, CPX$_2$,
%B\&C$_{\textrm{CLQ}}$ and B\&C$_{\textrm{ALL}}$. In particular, B\&C$_{\textrm{ALL}}$ evaluates much fewer
%nodes and uses around half of the time than B\&C$_{\textrm{CLQ}}$ for graphs with 50\%-70\% of density.
%This makes it worth considering new cuts proposed in Section \ref{SSEPAR}, though they behave a
%little worse in graphs with 30\% of density.
%
%By the other way, B\&C-LF$_2$ only seems to work better than B\&C$_{\textrm{ALL}}$ on graphs with 70\%
%of density in the sense that our algorithm is not capable of solving every instance. However, our
%algorithm seems to outperform B\&C-LF$_2$ in CPU time.

%%%%%%%%%%%%%%%%%%%%%%%%%%%%%%%%%%%%%%%%%%%%%%%%%%%%%%%%%%%%%%%%%%%%%%%%%%%%%%%%

%%%%%%%%%%%%%%%%%%%%%%%%%%%%%%%%%%%%%%%%%%%%%%%%%%%%%%%%%%%%%%%%%%%%%%%%%%%%%%%%

\newpage

\appendix

\begin{center} {\Large
\textbf{A polyhedral approach for the Equitable Coloring Problem}}\\[12pt]
Isabel M\'endez-D\'iaz,~~Graciela Nasini,~~Daniel Sever\'in\\[12pt]
\textsc{Online Appendix}
\end{center}

\section{Introduction} % \label{SAPP}

In this appendix we present sufficient conditions for some valid inequalities related to the \emph{Equitable Coloring Problem}
to be facet-defining inequalities.

%This appendix is devoted to present theoretical results related to valid
%inequalities for $\ECP$ already presented in the paper \emph{A polyhedral approach for the Equitable Coloring Problem}.
%In particular, for each family, we give sufficient conditions in order to be facet-defining
%inequalities.

All the proofs are based in the same technique, frequently used in the
literature for this kind of results, which is described in the following
remark.

\begin{trem} \label{TECHNIQUE}
Let $\pi^X x + \pi^W w \leq \pi_0$ be a valid inequality for $\ECP$ defining a proper face $\F'$.
%, \ie there exist solutions $(x', w'), (x'', w'') \in \ECP$ such that $\pi^X x' + \pi^W w' = \pi_0$ and $\pi^X x'' + \pi^W w'' < \pi_0$ respectively.
In order to prove that $\F'$ is a facet of $\ECP$ we have to show that, given any face $\F = \{ (x,w) \in \ECP ~:~ \lambda^X x + \lambda^W w = \lambda_0\}$ such that $\F' \subset \F$, $\lambda^X x + \lambda^W w = \lambda_0$ can be written as a linear combination of $\pi^X x + \pi^W w = \pi_0$ and the minimal equation system for $\ECP$ given in Theorem (\ref{TDIM}). %\footnote{As stated in the book \emph{Integer and Combinatorial Optimization} from G. Nemhauser and L. Wolsey (John Wiley \& Sons, 1988).}.
This last condition becomes equivalent to prove that $(\lambda^X, \lambda^W)$ verifies an equation system of $dim(\ECP)-1$ equalities.
The validity of each equality in the system is derived from the condition 
$\lambda^X x^1 + \lambda^W w^1 = \lambda_0 = \lambda^X x^2 + \lambda^W w^2$ applied on a suitable pair of equitable colorings $(x^1, w^1), (x^2,w^2)$ lying on $\F'$.
\end{trem}

For the sake of simplicity, we directly present the corresponding equation system on $(\lambda^X, \lambda^W)$ and the proposed equitable colorings used to derive each equation, bypassing how to get that equation system from the minimal equation system given in Theorem (\ref{TDIM}) and the inequality at hand.

As we have mentioned in Section \ref{SPOLYT}, we present equitable colorings by using mappings, color classes or binary vectors, according to our convenience.

% In this work, we particularly focus on rank inequalities with $\alpha(S) = 2$,
% called \emph{2-rank inequalities} and we give sufficient conditions and necessary conditions
% (but not both together) for these inequalities to define facets of $\ECP$.

%\begin{tlem} \label{T2RANK1VALID}
%The $(S,j)$-2-rank inequality is valid for $\ECP$.
%\end{tlem}
%\begin{proof}
%Since the inequality (\ref{R2RANK1}) is valid for $\CP$ (see Proposition 5 of \cite{BCCOL}), it is also
%valid for $\ECP$.
% Given a $k$-eqcol $(x,w)$, we verify that inequality (\ref{R2RANK1}) is valid for $(x,w)$ as follows:\\
% \textbf{If $k \leq n-2$}. Straightforward.\\
% \textbf{If $k = n-1$}. The value of the \lhs is at most 3 because one color class has 
% two vertices and the others have one.\\
% \textbf{If $k = n$}. Every color class is a singleton so the value of the \lhs is at most 2.
%\end{proof}

% \begin{tlem} \label{T2RANK1NEC}
% If the $(S,j)$-2-rank inequality defines a facet of $\ECP$, then $S$ is 2-maximal.
% \end{tlem}
% \begin{proof}
% If $S$ were not 2-maximal, the $(S,j)$-2-rank-inequality would be dominated by the
% $(S \cup \{v\},j)$-2-rank-inequality, where $v$ would be a vertex such that $\alpha(S \cup \{v\}) = 2$. 
% \end{proof}

\subsection{2-rank inequalities}

\begin{tthm} \label{T2RANK1}
Let $G$ be a monotone graph, $S \subset V$ such that $\alpha(S) = 2$ and $j \leq \lceil n/2 \rceil - 1$. If
\begin{enumerate}
\item[(i)] there exists a stable set $H$ of size 3 in $G$ such that:
\begin{itemize}
\item if $n$ is odd, the complement of $G-H$ has a perfect matching $M$ and both endpoints of some edge
of $M$ belong to $S$,
\item if $n$ is even, there exists another stable set $H'$ of size 3 in $G$ such that
$H \cap H' = \varnothing$, the complement of $G - (H \cup H')$ has a perfect matching $M$, both
endpoints of some edge of $M$ belong to $S$ and there exist vertices $h \in H$, $h' \in H'$ not
adjacent each other, 
\end{itemize}
\item[(ii)] for all $v \in V \backslash S$, there exist different vertices $s, s' \in S$
and a stable set $H_v = \{v, s, s'\}$ in $G$ such that:
\begin{itemize}
\item if $n$ is odd, the complement of $G-H_v$ has a perfect matching,
\item if $n$ is even, there exists another stable set $H'_v$ of size 3 in $G$ such that
$H_v \cap H'_v = \varnothing$ and the complement of $G - (H_v \cup H'_v)$ has a perfect matching,
\end{itemize}
\item[(iii)] for all $k$ such that $\max \{\chi_{eq}, j\} \leq k \leq \lceil n/2 \rceil - 2$, there exists a $k$-eqcol
where two vertices of $S$ share the same color,
\end{enumerate}
then the $(S,j)$-rank inequality, \ie
\begin{equation} \label{R2RANK1}
\sum_{v \in S} x_{vj} + \sum_{v \in V} x_{vn-1} \leq 2 w_j + w_{n-1} - w_n,
\end{equation}
defines a facet of $\ECP$.
\end{tthm}
\begin{proof}
Let $\F'$ be the face of $\ECP$ defined by (\ref{R2RANK1}) and
$\F = \{ (x,w) \in \ECP ~:~ \lambda^X x + \lambda^W w = \lambda_0\}$ be a face such that $\F' \subset \F$.
According to Remark \ref{TECHNIQUE}, we have to prove that $(\lambda^X, \lambda^W)$ verifies the following equation system: 
\begin{enumerate}
\item[(a)] $\lambda^X_{vj} = \lambda^X_{vn} + \lambda^W_n,~~~ \forall~ v \in S$.
\item[(b)] $\lambda^X_{vn-1} = \lambda^X_{vn} + \lambda^W_n,~~~ \forall~ v \in V$.
\item[(c)] $\lambda^X_{vk} = \lambda^X_{vn-1} + \lambda^W_{n-1},~~~
                   \forall~ v \in V,~ 1 \leq k \leq n-2,~ k \neq j$.
\item[(d)] $\lambda^X_{vj} = \lambda^X_{vn-1} + \lambda^W_{n-1},~~~
                                            \forall~ v \in V \backslash S$.
\item[(e)] $\lambda^W_k = 0,~~~ \forall~ \chi_{eq} + 1 \leq k \leq n-2,~ k \neq j$.
\item[(f)] If $j \geq \chi_{eq} + 1$ then $\lambda^W_j = - 2 \lambda^W_{n-1}$.
\end{enumerate}
We present pairs of equitable colorings lying on $\F'$ that allow us to
prove the validity of each equation in the previous system.
\begin{enumerate}
\item[(a)] Let $s, s' \in S$ be non adjacent vertices.\\
\textbf{Case $v = s$}. Let $c^1$ be a $(n-1)$-eqcol such that $c^1(s) = c^1(s') = j$ and
$c^2 = intro(c^1,s)$. Then, $\lambda^X_{sj} = \lambda^X_{sn} + \lambda^W_n$.\\
\textbf{Case $v \neq s$}. 
Let $c^1$ be a $n$-eqcol such that $c^1(v) = j$, $c^1(s) = n$ and $c^2 = swap_{j,n}(c^1)$. Then, 
$\lambda^X_{vj} + \lambda^X_{sn} = \lambda^X_{vn} + \lambda^X_{sj}$. Since $\lambda^X_{sj} = \lambda^X_{sn} + \lambda^W_n$, we obtain $\lambda^X_{vj} = \lambda^X_{vn} + \lambda^W_n$.
\item[(b)] \textbf{Case $v \notin S$}. 
By hypothesis (ii), there exist $s, s' \in S$ such that $\{v, s, s'\}$ is a stable set.
Let $c^1$ be a $(n-1)$-eqcol such that $c^1(v) = c^1(s) = n-1$, $c^1(s') = j$ and
$c^2 = intro(c^1,v)$. Therefore, $\lambda^X_{vn-1} = \lambda^X_{vn} + \lambda^W_n$.\\
\textbf{Case $v \in S$ and $|S| = 2$}. By hypothesis (ii), there exist $u \in V \backslash S$ and $v' \in S$ such
that  $\{u, v, v'\}$ is a stable set. Let $c^1$ be a $(n-1)$-eqcol such that $c^1(u) = c^1(v) = n-1$, $c^1(v') = j$
and $c^2 = intro(c^1,v)$. Therefore, $\lambda^X_{vn-1} = \lambda^X_{vn} + \lambda^W_n$.\\
\textbf{Case $v \in S$ and $|S| \geq 3$}. Let $s, s' \in S$ be non adjacent vertices,
$c^1$ be a $(n-1)$-eqcol such that $c^1(s) = c^1(s') = n-1$ and other vertex
of $S$ is painted with color $j$, and $c^2 = intro(c^1,s)$.
Then, $\lambda^X_{sn-1} = \lambda^X_{sn} + \lambda^W_n$ and the condition is proved for the case $v = s$.
If instead $v \neq s$, let $c^1$ be a $n$-eqcol such that $c^1(v) = n-1$, $c^1(s) = n$, other vertex of $S$ is
painted with color $j$ and $c^2 = swap_{n,n-1}(c^1)$. Then, 
$\lambda^X_{vn-1} + \lambda^X_{sn} = \lambda^X_{vn} + \lambda^X_{sn-1}$.
Since $\lambda^X_{sn-1} = \lambda^X_{sn} + \lambda^W_n$, we conclude that $\lambda^X_{vn-1} = \lambda^X_{vn} + \lambda^W_n$.
\item[(c)] Let $H$ and $M$ be the stable set and the matching given by hypothesis (i). Let $s, s' \in S$ be
the endpoints of an edge of $M$ and let $u, u' \in H$.\\
\textbf{Case $v = u$}. Let $c^1$ be a $(n-2)$-eqcol such that $c^1(u) = c^1(u') = k$, $c^1(s) = c^1(s') = j$ and
$c^2 = intro(c^1,u)$. We conclude that $\lambda^X_{uk} = \lambda^X_{un-1} + \lambda^W_{n-1}$.\\
\textbf{Case $v \neq u$}. Let $c^1$ be a $n$-eqcol such that $c^1(u) = k$, $c^1(v) = n-1$, a vertex of $S$ is
painted with color $j$ and $c^2 = swap_{k,n-1}(c^1)$. Then,
$\lambda^X_{u k} + \lambda^X_{v n-1} = \lambda^X_{u n-1} + \lambda^X_{vk}$.
Since $\lambda^X_{uk} = \lambda^X_{un-1} + \lambda^W_{n-1}$, we conclude that
$\lambda^X_{vk} = \lambda^X_{vn-1} + \lambda^W_{n-1}$.
\item[(d)] Let $H_v = \{v, s, s'\}$, $H'_v$ (if $n$ is even) and $M_v$ be the stable sets and the matching given
by hypothesis (ii), and let $H$, $H'$ (if $n$ is even) and $M$ be the stable sets and the matching given by hypothesis
(i). Let $c^1$ be a $(\lceil n/2 \rceil - 1)$-eqcol such that the color class $j$ is $H_v$ and the
remaining color classes are $H'_v$ (if $n$ is even) and the endpoints of edges of $M_v$.
Let $\hat{s}, \hat{s}' \in S$ be the endpoints of an edge of $M$ and
let $c^2$ be a $(\lceil n/2 \rceil - 1)$-eqcol such that the color class $j$ is $\{\hat{s}, \hat{s}'\}$ and the
remaining color classes are $H$, $H'$ (if $n$ is even) and the endpoints of edges of $M$ except
$(\hat{s}, \hat{s}')$. These colorings imply
\[ \lambda^X_{vj} + \lambda^X_{sj} + \lambda^X_{s'j} +
   \sum_{w \in V \backslash \{v, s, s'\}} \lambda^X_{w c^1(w)}
	= \lambda^X_{\hat{s}j} + \lambda^X_{\hat{s}'j} +
   \sum_{w \in V \backslash \{\hat{s}, \hat{s}'\}} \lambda^X_{w c^2(w)}. \]
Applying conditions (a)-(c), this last equality becomes
\[ \lambda^X_{vj} + \sum_{w \in V \backslash \{v\}} \lambda^X_{w n} + (n-3) \lambda^W_{n-1}
+ (n-1) \lambda^W_n = \sum_{w \in V} \lambda^X_{w n} + (n-2) \lambda^W_{n-1} + n \lambda^W_n,\]
giving rise to the desired result.
\item[(e)] Let us observe that from any $k$-eqcol $(x^1,w^1)$ and any $(k-1)$-eqcol
$(x^2,w^2)$ lying on $\F'$ we get $\lambda^X x^1 + \lambda^W_k = \lambda^X x^2$.
Then, applying conditions (a)-(d) yields $\lambda^W_k = 0$.\\
Thus we only need to prove that, for any $r$ such that $\chi_{eq} \leq r \leq n-2$, there exists an $r$-eqcol $c$
lying on $\F'$.\\
\textbf{Case $r < j$}. The existence of $c$ is guaranteed by the monotonicity of $G$.\\
\textbf{Case $j \leq r \leq \lceil n/2 \rceil-2$}. Hypothesis (iii) guarantees the existence of an $r$-eqcol $c'$
where two vertices $s, s' \in S$ satisfy $c'(s) = c'(s')$. Then, $c = swap_{c'(s),j}(c')$ is an $r$-eqcol that lies
on $\F'$.\\
\textbf{Case $r = \lceil n/2 \rceil-1$}. $c$ may be one of the colorings given in condition (d).\\
\textbf{Case $r = \lceil n/2 \rceil$}. Let $H$, $H'$ (if $n$ is even) and $M$ be the stable sets and the matching
given by hypothesis (i). Let $s, s' \in S$ be the endpoints of an edge of $M$ and let $h \in H$, $h' \in H'$ (if $n$
is even) be non adjacent vertices.\\
If $n$ is odd, color classes of $c$ are $\{h\}$, $H \backslash \{h\}$ and the endpoints of edges of $M$ where
$\{s,s'\}$ is the class $j$. If instead $n$ is even, color classes of $c$ are $\{h,h'\}$, $H \backslash \{h\}$,
$H' \backslash \{h'\}$ and the endpoints of edges of $M$ where $\{s,s'\}$ is the class $j$.\\
\textbf{Case $r \geq \lceil n/2 \rceil+1$}. Let us consider the $\lceil n/2 \rceil$-eqcol yielded in the
previous case and let $v_1, v_2$ be vertices sharing a color different from $j$.
In order to generate a $(\lceil n/2 \rceil + 1)$-eqcol $c$, we introduce a new color on $v_1$, \ie $c = intro(c', v_1)$
where $c'$ is the $\lceil n/2 \rceil$-eqcol. By repeating this procedure, we can generate a
$(\lceil n/2 \rceil + 2)$-eqcol and so on.
\item[(f)] Let $c^1$ be a $j$-eqcol such that $c^1(s) = c^1(s') = j$ for some $s, s' \in S$ and
$c^2$ be a $(j-1)$-eqcol (the existence of these colorings is proved above). Hence,
\[ \lambda^X_{sj} + \lambda^X_{s'j} +
   \sum_{v \in V \backslash \{s, s'\}} \lambda^X_{v c^1(v)} + \lambda^W_j
	= \sum_{v \in V} \lambda^X_{v c^2(v)}. \]
Application of conditions (a)-(d) yields $\lambda^W_j = - 2 \lambda^W_{n-1}$.
\end{enumerate}
\end{proof}

Let us present an example where the previous theorem is applied.\\

\noindent \textbf{Example.}
Let $G$ be the graph presented in Figure \ref{minigraph1}. We have that $G$ is monotone and $\chi_{eq}(G) = 5$. 
If $S = \{1,2,\ldots,7\}$, $\alpha(S) = 2$ and $H = \{4, 7, 8\}$ is a stable set such that $\overline{G - H}$ has the perfect matching $\{(1,10), (2,11), (3,5), (6,9)\}$ with $\{3,5\}\subset S$. Moreover, it is not hard to verify that for all $v \in \{8,9,10,11\}$ there exists a stable set $H_v = \{ 4, 7, v\}$ such that $\overline{G - H_v}$ has a perfect matching.
Then, if $1 \leq j \leq 5=\lceil 11/2 \rceil-1$, the $(S,j)$-rank inequality is a facet-defining inequality of $\ECP(G)$.  
\begin{figure}[h]
  \centering
\begin{graph}(4, 2)(-2, -1)
	\roundnode{v1}(-2.3,-0.1)
	\autonodetext{v1}[w]{1}
	\roundnode{v2}(-1.7,0.1)
	\autonodetext{v2}[se]{2}
	\roundnode{v3}(-1.2,0)
	\autonodetext{v3}[n]{3}
	\roundnode{v4}(-1.7528,0.7608)
	\autonodetext{v4}[n]{4}
	\roundnode{v5}(-2.6472,0.4702)
	\autonodetext{v5}[nw]{5}
	\roundnode{v6}(-2.6472,-0.4702)
	\autonodetext{v6}[sw]{6}
	\roundnode{v7}(-1.7528,-0.7608)
	\autonodetext{v7}[s]{7}
	\roundnode{v8}(-0.4,0)
	\autonodetext{v8}[n]{8}
	\roundnode{v9}(0.4,0)
	\autonodetext{v9}[n]{9}
	\roundnode{v10}(1.2,0)
	\autonodetext{v10}[n]{10}
	\roundnode{v11}(2,0)
	\autonodetext{v11}[n]{11}
	\edge{v1}{v2}
	\edge{v1}{v3}
	\edge{v1}{v4}
	\edge{v1}{v5}
	\edge{v1}{v6}
	\edge{v1}{v7}
	\edge{v2}{v3}
	\edge{v2}{v4}
	\edge{v2}{v5}
	\edge{v2}{v6}
	\edge{v2}{v7}
	\edge{v3}{v7}
	\edge{v3}{v4}
	\edge{v4}{v5}
	\edge{v5}{v6}
	\edge{v6}{v7}
	\edge{v3}{v8}
	\edge{v8}{v9}
	\edge{v9}{v10}
	\edge{v10}{v11}
\end{graph}
  \caption{}
  \label{minigraph1}
\end{figure}

\begin{tthm} \label{T2RANK2}
Let $G$ be a monotone graph, $S \subset V$ such that $\alpha(S) = 2$ and $Q = \{ q \in S ~:~ S \subset N[q] \}$.
If $|Q| \geq 2$ and
\begin{enumerate}
\item[(i)] no connected component of the complement of $G[S \backslash Q]$ is bipartite,
\item[(ii)] for all $v \in V \backslash S$ verifying $Q \subset N(v)$, there exist two
vertices $s, s' \in S \backslash Q$ and a stable set $H_v = \{v, s, s'\}$ in $G$ such that:
\begin{itemize}
\item if $n$ is odd, the complement of $G-H_v$ has a perfect matching,
\item if $n$ is even, there exists another stable set $H'_v$ of size 3 in $G$ such that
$H_v \cap H'_v = \varnothing$ and the complement of $G - (H_v \cup H'_v)$ has a perfect matching,
\end{itemize}
\end{enumerate}
then, for all $j \leq \lceil n/2 \rceil - 1$, the $(S,Q,j)$-2-rank inequality, \ie
\begin{equation} \label{R2RANK2AGAIN}
\sum_{v \in S\backslash Q} x_{vj} + 2 \sum_{v \in Q} x_{vj} \leq 2 w_j,
\end{equation}
defines a facet of $\ECP$.
\end{tthm}
\begin{proof}
Let $q, q'$ be different vertices of $Q$.

Let $\F'$ be the face of $\ECP$ defined by (\ref{R2RANK2AGAIN}) and
$\F = \{ (x,w) \in \ECP ~:~ \lambda^X x + \lambda^W w = \lambda_0\}$ be a face such that $\F' \subset \F$.
According to Remark \ref{TECHNIQUE}, we have to prove that $(\lambda^X, \lambda^W)$ verifies the following equation system: 
\begin{enumerate}
\item[(a)] $\lambda^X_{vj} = \lambda^X_{vn} + \lambda^W_n,~~~
	      \forall~ v \in V \backslash S ~\textrm{such that}~ Q \backslash N(v) \neq \varnothing$.
\item[(b)] $\lambda^X_{vk} = \lambda^X_{vn} + \lambda^W_n,~~~
	      \forall~ v \in V,~ 1 \leq k \leq n-1,~ k \neq j$.
\item[(c)] $\lambda^X_{qn} + \lambda^X_{vj} = \lambda^X_{qj} + \lambda^X_{vn},~~~
	      \forall~ v \in Q \backslash \{q\}$.
\item[(d)] $\lambda^X_{qn} + 2\lambda^X_{vj} = 2\lambda^X_{vn} + \lambda^X_{qj} + \lambda^W_n,~~~
	      \forall~ v \in S \backslash Q$.
\item[(e)] $\lambda^X_{vj} = \lambda^X_{vn} + \lambda^W_n,~~~
	      \forall~ v \in V \backslash S ~\textrm{such that}~ Q \subset N(v)$.
\item[(f)] $\lambda^W_k = 0,~~~ \forall~ \chi_{eq} + 1 \leq k \leq n-1,~ k \neq j$.
\item[(g)] If $j \geq \chi_{eq} + 1$ then
           $\lambda^X_{qn} + \lambda^W_n = \lambda^X_{qj} + \lambda^W_j$.
\end{enumerate}
We present pairs of equitable colorings lying on $\F'$ that allow us to
prove the validity of each equation in the previous system.
\begin{enumerate}
\item[(a)] Let $\hat{q} \in Q \backslash N(v)$ and let $c^1$ be a $(n-1)$-eqcol such that
$c^1(\hat{q}) = c^1(v) = j$ and $c^2 = intro(c^1,v)$.
We conclude that $\lambda^X_{vj} = \lambda^X_{vn} + \lambda^W_n$.
\item[(b)] Let $s, s' \in S \backslash Q$ be non adjacent vertices.\\
\textbf{Case $v = s$}. Let $c^1$ be a $(n-1)$-eqcol such that $c^1(s) = c^1(s') = k$ and $c^1(q) = j$,
and $c^2 = intro(c^1,s)$. Then, $\lambda^X_{sk} = \lambda^X_{sn} + \lambda^W_n$.\\
\textbf{Case $v \neq s$}. Let $c^1$ be a $n$-eqcol such that $c^1(v) = k$, $c^1(s) = n$. If $v = q$, we make
$c^1(q') = j$. Otherwise, we make $c^1(q) = j$. From the coloring $c^2 = swap_{k,n}(c^1)$ we have
$\lambda^X_{vk} + \lambda^X_{sn} = \lambda^X_{vn} + \lambda^X_{sk}$ and since $\lambda^X_{sk} = \lambda^X_{sn} + \lambda^W_n$ we obtain $\lambda^X_{vk} = \lambda^X_{vn} + \lambda^W_n$.
\item[(c)] Let $c^1$ be a $n$-eqcol such that $c^1(q) = n$, $c^1(v) = j$ and $c^2 = swap_{j,n}(c^1)$.
Therefore, $\lambda^X_{qn} + \lambda^X_{vj} = \lambda^X_{qj} + \lambda^X_{vn}$.
\item[(d)] Let $J$ be the connected component in the complement of $G[S \backslash Q]$
such that $v$ is a vertex of $J$. Since $\alpha(S) = 2$, $J$ does not have triangles.
By hypothesis (i), $J$ is not bipartite and therefore there exists at least an odd cycle in $J$ of size $p$ with $p \geq 5$.\\
Now, let $d(v)$ be the minimum distance in $J$ between $v$ and all the odd cycles in $J$, where the \emph{distance} from a vertex $v$ to an odd cycle is defined as the minimum number of vertices of a path between $v$ and a vertex of the odd cycle.
Condition (d) is proved by induction on $d(v)$.\\
\textbf{Case $d(v) = 0$}. Then, $v$ belongs to an odd cycle of size $p \geq 5$ in $J$. Let $v_1 = v, v_2, \ldots, v_p \in S \backslash Q$ be the vertices of that odd cycle, and let $k_1, k_2, \ldots, k_{p+1}$
be colors different from $j$.\\
%We propose $\frac{p+1}{2}$ pairs of colorings which allows us to determine
%several equations whose sum is the desired result.\\
We denote by $\oplus$ the sum of two integers modulo $p$.
Let $c^1$, $c^2$, $\ldots$, $c^p$ be $(n-1)$-eqcols such that,
for each $1 \leq i \leq p$, $c^i(v_i) = j$, $c^i(v_{i \oplus 1}) = j$,
$c^i(v_r) = k_r ~~\forall~ r \in \{1,\ldots,p\} \backslash \{i, i \oplus 1\}$, $c^i(q) = k_{p+1}$, and let
$c^{p+1}$ be a $n$-eqcol such that $c^{p+1}(v_1) = n$, $c^{p+1}(v_r) = k_r ~~\forall~ r \in \{2,\ldots,p\}$, $c^{p+1}(q) = j$. For instance, if $p = 5$, colors of $v_1$, $\ldots$, $v_5$ and $q$ would be:
\begin{center} \small
\begin{tabular}{|c|c@{\hspace{2pt}}c@{\hspace{2pt}}c@{\hspace{2pt}}c@{\hspace{2pt}}c@{\hspace{2pt}}c|c|c@{\hspace{2pt}}c@{\hspace{2pt}}c@{\hspace{2pt}}c@{\hspace{2pt}}c@{\hspace{2pt}}c|}
\hline
 \multicolumn{7}{|c|}{$c^i$ with $i$ odd} & \multicolumn{7}{|c|}{$c^i$ with $i$ even} \\
\hline
 size & $v_1$ & $v_2$ & $v_3$ & $v_4$ & $v_5$ & $q$ &
 size & $v_1$ & $v_2$ & $v_3$ & $v_4$ & $v_5$ & $q$ \\
\hline
 $n-1$ & $j$ & $j$ & $k_3$ & $k_4$ & $k_5$ & $k_6$ &
 $n-1$ & $k_1$ & $j$ & $j$ & $k_4$ & $k_5$ & $k_6$ \\
 $n-1$ & $k_1$ & $k_2$ & $j$ & $j$ & $k_5$ & $k_6$ &
 $n-1$ & $k_1$ & $k_2$ & $k_3$ & $j$ & $j$ & $k_6$ \\
 $n-1$ & $j$ & $k_2$ & $k_3$ & $k_4$ & $j$ & $k_6$ &
 $n$ & $n$ & $k_2$ & $k_3$ & $k_4$ & $k_5$ & $j$ \\
\hline
\end{tabular}
\end{center}
We assume that the remaining vertices have the same color in all the colorings.
Thus, we obtain
%if $x^1$, $x^2$, $\ldots$, $x^{p+1}$ are the binary variables representing colorings
%$c^1$, $c^2$, $\ldots$, $c^{p+1}$, equalities $\lambda^X x^1 = \lambda^X x^2$, $\lambda^X %x^3 = \lambda^X x^4$,
%$\ldots$, $\lambda^X x^p = \lambda^X x^{p+1} + \lambda^W_n$ are valid and we can add them in order to obtain 
\[ \sum_{\substack{i = 1\\i~\textrm{odd}}}^{p+1} \sum_{v \in V} \lambda^X_{vc^i(v)} =
\sum_{\substack{i = 1\\i~\textrm{even}}}^{p+1} \sum_{v \in V} \lambda^X_{vc^i(v)} + \lambda^W_n. \]
By condition (b), we get 
$\lambda^X_{qn} + 2\lambda^X_{vj} = 2\lambda^X_{vn} + \lambda^X_{qj} + \lambda^W_n$.\\
\textbf{Case $d(v) \geq 1$}. Let $v' \in J$ be a vertex adjacent to $v$ in $J$ such that
$d(v') = d(v) - 1$. By inductive hypothesis, $\lambda^X_{qn} + 2\lambda^X_{v'j} = 2\lambda^X_{v'n} + \lambda^X_{qj} + \lambda^W_n$.\\
Let $c^1$ be a $(n-1)$-eqcol such that $c^1(v) = c^1(v') = j$ and $c^1(q) = k$, where
$k \neq j$. Let $c^2$ be a $n$-eqcol such that $c^2(v) = k$, $c^2(v') = n$, $c^2(q) = j$ and
$c^2(i) = c^1(i) ~~\forall~i \in V \backslash \{v, v', q\}$. Hence
$\lambda^X_{vj} + \lambda^X_{v'j} + \lambda^X_{qk} = \lambda^X_{vk} + \lambda^X_{v'n} + \lambda^X_{qj}
+ \lambda^W_n$. Multiplying this equality by 2, subtracting
$\lambda^X_{qn} + 2\lambda^X_{v'j} = 2\lambda^X_{v'n} + \lambda^X_{qj} + \lambda^W_n$ and
applying condition (b) yields
$\lambda^X_{qn} + 2\lambda^X_{vj} = 2\lambda^X_{vn} + \lambda^X_{qj} + \lambda^W_n$. 
\item[(e)] By hypothesis (ii), we can establish a $(\lceil n/2 \rceil - 1)$-eqcol $c^1$ such
that color class $j$ is $\{v, s, s'\}$ where $s, s' \in S$ (as we did in condition (d) of Theorem \ref{T2RANK1}).
Let $k$ be the color of $q$ in $c^1$ and $c^2 = swap_{j,k}(c^1)$. We get
$\lambda^X_{vj} = \lambda^X_{vn} + \lambda^W_n$ by applying conditions (a)-(d).
\item[(f)] Since $G$ is monotone, there exist a $k$-eqcol $c$ and a $(k-1)$-eqcol $c'$.
If $k < j$, we consider $c^1 = c$ and $c^2 = c'$. If $k > j$, we consider
$c^1 = swap_{c(q),j}(c)$ and $c^2 = swap_{c'(q),j}(c')$. Then, we apply conditions (a)-(e)
to $\lambda^X x^1 + \lambda^W_k = \lambda^X x^2$, where $x^1$ and $x^2$ are the binary variables representing
colorings $c^1$ and $c^2$ respectively.
\item[(g)] Let $c^1$ be a $j$-eqcol such that $c^1(q) = j$ and $c^2$ be a $(j-1)$-eqcol (the existence of these
colorings is proved above). Then, we apply conditions (a)-(e) to $\lambda^X x^1 + \lambda^W_j = \lambda^X x^2$,
where $x^1$ and $x^2$ are the binary variables representing colorings $c^1$ and $c^2$ respectively.
\end{enumerate}
\end{proof}

Theorem \ref{T2RANK2} states that, among other things, $j \leq \lceil n/2 \rceil - 1$ for the
$(S,Q,j)$-2-rank-inequality to define a facet of $\ECP$. Indeed,
this condition is only used in Theorem \ref{T2RANK2} for proving equations given in (e),
\ie $\lambda^X_{vj} = \lambda^X_{vn} + \lambda^W_n,~ \textrm{for all}~ v \in V \backslash S ~\textrm{such that}~ Q \subset N(v)$. So, if every vertex $v \in V \backslash S$ verifies
$Q \backslash N(v) \neq \varnothing$, these equations
vanish from the equation system on $(\lambda^X, \lambda^W)$ and the inequality (\ref{R2RANK2AGAIN}) defines a facet of $\ECP$ even though $j > \lceil n/2 \rceil - 1$. We have proved the following result.

\begin{tcor} \label{T2RANK2COL}
Let $G$ be a monotone graph, $S \subset V$ such that $\alpha(S) = 2$ and $Q = \{ q \in S ~:~ S \subset N[q] \}$.
If $|Q| \geq 2$, no connected component of the complement of $G[S \backslash Q]$
is bipartite and for all $v \in V \backslash S$, $Q \backslash N(v) \neq \varnothing$, then the $(S,Q,j)$-2-rank
inequality defines a facet of $\ECP$ for all $j \leq n-1$.
\end{tcor}

Let us present an example where the previous result is applied.\\

\noindent \textbf{Example.} Let $G$ be the graph presented in Figure \ref{minigraph1}, $S = \{1,2,\ldots,7\}$ and $Q = \{1,2\}$.
The $(S,Q,j)$-2-rank inequality is a facet-defining inequality of $\ECP(G)$ for $1 \leq j \leq 10$ since the
assumptions of Corollary \ref{T2RANK2COL} are satisfied: vertices $3,\ldots,7$ induce an odd cycle in
$\overline{G}$ and for all $v \in \{8,9,10,11\}$, $Q \backslash N(v) = \{1,2\}$.

\subsection{Subneighborhood inequalities}

\begin{tthm} \label{TNEIGHBOR1}
Let $G$ be a monotone graph, $u \in V$, $j \leq n - 1$ such that $\lceil n/j \rceil \leq \lceil n/\chi_{eq} \rceil$
and $S \subset N(u)$ such that $S$ is not a clique of $G$ and, if $S \neq N(u)$ then $\alpha(S) \leq \lceil n/j \rceil - 1$.\\
If
\begin{enumerate}
\item[(i)] for all $3 \leq i \leq \min \{\lceil n/j \rceil, \alpha(S)\}$, there exists a
$\biggl(\biggl\lceil \dfrac{n}{i-1} \biggr\rceil - 1\biggl)$-eqcol whose color class $C_j$ satisfies $|C_j\cap S| = i$,
\item[(ii)] for all $v \in N(u) \backslash S$, there exists an equitable coloring whose color class $C_j$ satisfies
$|C_j \cap S| = \alpha(S)$ and $(C_j\cap N(u)) \backslash S=\{v\}$,
\end{enumerate}
then the $(u,j,S)$-subneighborhood inequality, \ie
\begin{equation} \label{RNEIGHBOR1AGAIN}
\gamma_{jS} x_{uj} + \sum_{v \in S} x_{vj} +
\sum_{k = j+1}^n (\gamma_{jS} - \gamma_{kS}) x_{uk} \leq \gamma_{jS} w_j,
\end{equation}
defines a facet of $\ECP$, where
$\gamma_{kS} = \min \{\lceil n/k \rceil, \alpha(S)\}$.
\end{tthm}
\begin{proof}
Let $\F'$ be the face of $\ECP$ defined by (\ref{RNEIGHBOR1AGAIN}) and
$\F = \{ (x,w) \in \ECP ~:~ \lambda^X x + \lambda^W w = \lambda_0\}$ be a face such that $\F' \subset \F$.
According to Remark \ref{TECHNIQUE}, we have to prove that $(\lambda^X, \lambda^W)$ verifies the following equation system: 
\begin{enumerate}
\item[(a)] $\lambda^X_{vj} = \lambda^X_{vn} + \lambda^W_n,~~~ \forall~ v \in V \backslash N[u]$.
\item[(b)] $\lambda^X_{vk} = \lambda^X_{vn} + \lambda^W_n,~~~
	      \forall~ v \in V \backslash \{u\},~ 1 \leq k \leq n-1,~ k \neq j$.
\item[(c)] $\lambda^X_{vj} + \lambda^X_{un} = \lambda^X_{vn} + \lambda^X_{uj},~~~
	      \forall~ v \in S$.
\item[(d)] $\lambda^X_{uk} + (\gamma_{jS} - 1) \lambda^X_{uj} = \gamma_{jS} \lambda^X_{un}
              + \gamma_{jS} \lambda^W_n,~~~ \forall~ 1 \leq k \leq j-1$.
\item[(e)] $\lambda^X_{uk} + (\gamma_{kS} - 1) \lambda^X_{uj} =
	    \gamma_{kS} \lambda^X_{un} + \gamma_{kS} \lambda^W_n,~~~
	      \forall~ j+1 \leq k \leq n-1$.
\item[(f)] $\lambda^X_{vj} = \lambda^X_{vn} + \lambda^W_n,~~~ \forall~ v \in N(u) \backslash S$.
\item[(g)] $\lambda^W_k = 0,~~~ \forall~ \chi_{eq} + 1 \leq k \leq n-1,~ k \neq j$.
\item[(h)] If $j \geq \chi_{eq} + 1$ then
           $\gamma_{jS} \lambda^X_{un} + \gamma_{jS} \lambda^W_n = \gamma_{jS} \lambda^X_{uj} + \lambda^W_j$.
\end{enumerate}
We present pairs of equitable colorings lying on $\F'$ that allow us to
prove the validity of each equation in the previous system.
\begin{enumerate}
\item[(a)] Let $c^1$ be a $(n-1)$-eqcol such that $c^1(u) = c^1(v) = j$ and $c^2 = intro(c^1,v)$. We conclude that
$\lambda^X_{vj} = \lambda^X_{vn} + \lambda^W_n$. 
\item[(b)] Let $s, s' \in S$ be non adjacent vertices.\\
\textbf{Case $v = s$}. Let $c^1$ be a $(n-1)$-eqcol such that $c^1(s) = c^1(s') = k$, $c^1(u) = j$
and $c^2 = intro(c^1,s)$. Then, $\lambda^X_{sk} = \lambda^X_{sn} + \lambda^W_n$.\\
\textbf{Case $v \neq s$}. Let $c^1$ be a $n$-eqcol such that $c^1(v) = k$, $c^1(s) = n$, $c^1(u) = j$ and
$c^2 = swap_{k,n}(c^1)$. We have $\lambda^X_{vk} + \lambda^X_{sn} = \lambda^X_{vn} + \lambda^X_{sk}$.
Since $\lambda^X_{sk} = \lambda^X_{sn} + \lambda^W_n$, we conclude that
$\lambda^X_{vk} = \lambda^X_{vn} + \lambda^W_n$.
\item[(c)] Let $c^1$ be a $n$-eqcol such that $c^1(v) = j$, $c^1(u) = n$ and $c^2 = swap_{j,n}(c^1)$.
Therefore, $\lambda^X_{vj} + \lambda^X_{un} = \lambda^X_{vn} + \lambda^X_{uj}$.
\item[(d)] \textbf{Case $\gamma_{jS} = 2$}. Let $c^1$ be a $(n-1)$-eqcol such that $c^1(u) = k$ and
$c^1(s) = c^1(s') = j$ where $s, s' \in S$.\\
\textbf{Case $\gamma_{jS} \geq 3$}. Let $c$ be the $(\lceil \frac{n}{\gamma_{jS} - 1} \rceil - 1)$-eqcol given
by hypothesis (i) and $c^1 = swap_{c(u),k}(c)$.\\
In both cases, $c^1(u) = k$. Now, let $C_j$ and $C_k$ be the color classes $j$
and $k$ of $c^1$ respectively. Considering $c^2 = swap_{j,k}(c^1)$ give rise to
\[ \lambda^X_{uk} + \sum_{v \in C_j} \lambda^X_{vj} + \sum_{v \in C_k \backslash \{u\}} \lambda^X_{vk} =
   \lambda^X_{uj} + \sum_{v \in C_j} \lambda^X_{vk} + \sum_{v \in C_k \backslash \{u\}} \lambda^X_{vj}. \]
Since $|C_j \cap S| = \gamma_{jS}$, we have $C_j \subset S$ and we can apply (a)-(c) in order to
get $\lambda^X_{uk} + (\gamma_{jS} - 1) \lambda^X_{uj} = \gamma_{jS} \lambda^X_{un} + \gamma_{jS} \lambda^W_n$.
\item[(e)] We proceed in the same way as in (d) except that, for the case $\gamma_{jS} \geq 3$, we use the
$(\lceil \frac{n}{\gamma_{kS} - 1} \rceil - 1)$-eqcol given by hypothesis (i) instead of the
$(\lceil \frac{n}{\gamma_{jS} - 1} \rceil - 1)$-eqcol.
\item[(f)] In first place, let us note that $v \in N(u) \backslash S$ implies $S \subsetneqq N(u)$.
Then, $\alpha(S) \leq \lceil n/j \rceil - 1$ and, by hypothesis (ii), there exists a coloring $c^1$ that paints $v$
and $\alpha(S)$ vertices of $S$ with color $j$ but the remaining vertices of $N(u)$ do not use $j$.
Let $k$ be the color used by vertex $u$ in $c^1$ and let $C_j$, $C_k$ be the color classes $j$ and $k$ in
$c^1$ respectively, and $c^2 = swap_{j,k}(c^1)$. We have
\[ \lambda^X_{uk} + \lambda^X_{vj} + \sum_{w \in C_j\backslash \{v\}} \lambda^X_{wj} +
                                    \sum_{w \in C_k \backslash \{u\}} \lambda^X_{wk} =
   \lambda^X_{uj} + \lambda^X_{vk} + \sum_{w \in C_j\backslash \{v\}} \lambda^X_{wk} +
                                    \sum_{w \in C_k \backslash \{u\}} \lambda^X_{wj}. \]
In virtue of conditions (a)-(e), we obtain $\lambda^X_{vj} = \lambda^X_{vn} + \lambda^W_n$.
\item[(g)] Since $G$ is monotone, there exist a $k$-eqcol $c$ and a $(k-1)$-eqcol $c'$.
If $k < j$, we consider $c^1 = c$ and $c^2 = c'$. If $k > j$, we consider
$c^1 = swap_{c(u),j}(c)$ and $c^2 = swap_{c'(u),j}(c')$. Then, we apply conditions (a)-(f)
to $\lambda^X x^1 + \lambda^W_k = \lambda^X x^2$, where $x^1$ and $x^2$ are the binary variables representing
colorings $c^1$ and $c^2$ respectively.
\item[(h)] Let $c^1$ be a $j$-eqcol such that $c^1(u) = j$ and $c^2$ be a $(j-1)$-eqcol (the existence of these
colorings is proved above). Then, we apply conditions (a)-(f) to $\lambda^X x^1 + \lambda^W_j = \lambda^X x^2$,
where $x^1$ and $x^2$ are the binary variables representing colorings $c^1$ and $c^2$ respectively.
\end{enumerate}
\end{proof}

Let us present two examples where the previous theorem is applied.\\

\noindent \textbf{Example.} Let $G$ be the graph given in Figure \ref{minigraph2}(a). We have that $G$ is
monotone and $\chi_{eq}(G) = 3$. Let us consider $u = 1$, $S = N(1)$ and $j = 3$. In order to prove that the $(u,j,S)$-subneighborhood inequality
defines a facet of $\ECP(G)$, it is enough to exhibit a ($\lceil \frac{11}{2} \rceil - 1$)-eqcol
such that $|C_3 \cap S| = 3$ and a ($\lceil \frac{11}{3} \rceil - 1$)-eqcol such that $|C_3 \cap S| = 4$.
Both colorings are shown in Figure \ref{minigraph2} (b) and (c) respectively.

It is not hard to see that the $(u,j,S)$-subneighborhood inequality is also facet-defining for
$4 \leq j \leq 10$.
%\[ 4 x_{1,3} + \sum_{v=2}^6 x_{v,3}
%   + x_{1,4} + x_{1,5} + 2 x_{1,6} + 2 x_{1,7} + 2 x_{1,8} + 2 x_{1,9}
%   + 2 x_{1,10} + 3 x_{1,11} \leq 4 w_3, \]
On the other hand, $(u,j,S)$-subneighborhood inequality with $j \in \{ 1,2 \}$
is facet-defining by Theorem \ref{TIGHT}.

Therefore, the $(u,j,S)$-subneighborhood inequality defines a facet of $\ECP(G)$ for all $1 \leq j \leq 10$.\\

%We can also find facet-defining inequalities by inspecting subsets of $N(1)$ and applying Theorem
%\ref{TNEIGHBOR1}. For instance,
\noindent \textbf{Example.} Let us consider again the graph given in Figure \ref{minigraph2}(a). The $(u,j,S)$-subneighborhood inequality with $u = 1$, $j = 3$ and $S = \{3,4,5\}$
%, \ie $3 x_{1,3} + x_{3,3} + x_{4,3} + x_{5,3} + x_{1,6} + x_{1,7} + x_{1,8} + x_{1,9}
% + x_{1,10} + 2 x_{1,11} \leq 3 w_3$,
is facet-defining since $\alpha(S) \leq \lceil \frac{11}{3} \rceil - 1$ and there exist the following colorings:
a ($\lceil \frac{11}{2} \rceil - 1$)-eqcol such that
$|C_3 \cap S| = 3$, an equitable coloring such that $|C_3 \cap S| = 3$ and $(C_3 \cap N(1)) \backslash S = \{ 2 \}$, and
an equitable coloring such that $|C_3 \cap S| = 3$ and $(C_3 \cap N(1)) \backslash S = \{ 6 \}$.
These colorings are shown in Figure \ref{minigraph2} (b), (c) and (d) respectively.

\begin{figure}[h]
  \centering ~~~~~~~~~~~
\begin{graph}(5.8, 2)(-2, -1)
	\freetext(0.4,-0.75){(a) labeling of $G$}
	\roundnode{v1}(-2,0)
	\autonodetext{v1}[w]{1}
	\roundnode{v2}(-1.2,0)
	\autonodetext{v2}[n]{2}
	\roundnode{v3}(-1.7528,0.7608)
	\autonodetext{v3}[n]{3}
	\roundnode{v4}(-2.6472,0.4702)
	\autonodetext{v4}[nw]{4}
	\roundnode{v5}(-2.6472,-0.4702)
	\autonodetext{v5}[sw]{5}
	\roundnode{v6}(-1.7528,-0.7608)
	\autonodetext{v6}[s]{6}
	\roundnode{v7}(-0.4,0)
	\autonodetext{v7}[n]{7}
	\roundnode{v8}(0.4,0)
	\autonodetext{v8}[n]{8}
	\roundnode{v9}(1.2,0)
	\autonodetext{v9}[n]{9}
	\roundnode{v10}(2,0)
	\autonodetext{v10}[n]{10}
	\roundnode{v11}(2.8,0)
	\autonodetext{v11}[n]{11}
	\edge{v1}{v2}
	\edge{v1}{v3}
	\edge{v1}{v4}
	\edge{v1}{v5}
	\edge{v1}{v6}
	\edge{v2}{v7}
	\edge{v7}{v8}
	\edge{v8}{v9}
	\edge{v9}{v10}
	\edge{v10}{v11}
\end{graph}~~~~~~~
\begin{graph}(5.8, 2)(-2, -1)
	\freetext(0.4,-0.75){(b) 5-eqcol in $G$}
	\roundnode{v1}(-2,0)
	\autonodetext{v1}[w]{5}
	\roundnode{v2}(-1.2,0)
	\autonodetext{v2}[n]{1}
	\roundnode{v3}(-1.7528,0.7608)
	\autonodetext{v3}[n]{3}
	\roundnode{v4}(-2.6472,0.4702)
	\autonodetext{v4}[nw]{3}
	\roundnode{v5}(-2.6472,-0.4702)
	\autonodetext{v5}[sw]{3}
	\roundnode{v6}(-1.7528,-0.7608)
	\autonodetext{v6}[s]{2}
	\roundnode{v7}(-0.4,0)
	\autonodetext{v7}[n]{4}
	\roundnode{v8}(0.4,0)
	\autonodetext{v8}[n]{5}
	\roundnode{v9}(1.2,0)
	\autonodetext{v9}[n]{1}
	\roundnode{v10}(2,0)
	\autonodetext{v10}[n]{2}
	\roundnode{v11}(2.8,0)
	\autonodetext{v11}[n]{4}
	\edge{v1}{v2}
	\edge{v1}{v3}
	\edge{v1}{v4}
	\edge{v1}{v5}
	\edge{v1}{v6}
	\edge{v2}{v7}
	\edge{v7}{v8}
	\edge{v8}{v9}
	\edge{v9}{v10}
	\edge{v10}{v11}
\end{graph}\\~~~~~~~~~~~
\begin{graph}(5.8, 3)(-2, -1)
	\freetext(0.8,-0.75){(c) 3-eqcol with $c(2)=3$}
	\roundnode{v1}(-2,0)
	\autonodetext{v1}[w]{1}
	\roundnode{v2}(-1.2,0)
	\autonodetext{v2}[n]{3}
	\roundnode{v3}(-1.7528,0.7608)
	\autonodetext{v3}[n]{3}
	\roundnode{v4}(-2.6472,0.4702)
	\autonodetext{v4}[nw]{3}
	\roundnode{v5}(-2.6472,-0.4702)
	\autonodetext{v5}[sw]{3}
	\roundnode{v6}(-1.7528,-0.7608)
	\autonodetext{v6}[s]{2}
	\roundnode{v7}(-0.4,0)
	\autonodetext{v7}[n]{1}
	\roundnode{v8}(0.4,0)
	\autonodetext{v8}[n]{2}
	\roundnode{v9}(1.2,0)
	\autonodetext{v9}[n]{1}
	\roundnode{v10}(2,0)
	\autonodetext{v10}[n]{2}
	\roundnode{v11}(2.8,0)
	\autonodetext{v11}[n]{1}
	\edge{v1}{v2}
	\edge{v1}{v3}
	\edge{v1}{v4}
	\edge{v1}{v5}
	\edge{v1}{v6}
	\edge{v2}{v7}
	\edge{v7}{v8}
	\edge{v8}{v9}
	\edge{v9}{v10}
	\edge{v10}{v11}
\end{graph}~~~~~~~~~~~
\begin{graph}(5.8, 3)(-2, -1)
	\freetext(0.8,-0.75){(d) 3-eqcol with $c(6)=3$}
	\roundnode{v1}(-2,0)
	\autonodetext{v1}[w]{1}
	\roundnode{v2}(-1.2,0)
	\autonodetext{v2}[n]{2}
	\roundnode{v3}(-1.7528,0.7608)
	\autonodetext{v3}[n]{3}
	\roundnode{v4}(-2.6472,0.4702)
	\autonodetext{v4}[nw]{3}
	\roundnode{v5}(-2.6472,-0.4702)
	\autonodetext{v5}[sw]{3}
	\roundnode{v6}(-1.7528,-0.7608)
	\autonodetext{v6}[s]{3}
	\roundnode{v7}(-0.4,0)
	\autonodetext{v7}[n]{1}
	\roundnode{v8}(0.4,0)
	\autonodetext{v8}[n]{2}
	\roundnode{v9}(1.2,0)
	\autonodetext{v9}[n]{1}
	\roundnode{v10}(2,0)
	\autonodetext{v10}[n]{2}
	\roundnode{v11}(2.8,0)
	\autonodetext{v11}[n]{1}
	\edge{v1}{v2}
	\edge{v1}{v3}
	\edge{v1}{v4}
	\edge{v1}{v5}
	\edge{v1}{v6}
	\edge{v2}{v7}
	\edge{v7}{v8}
	\edge{v8}{v9}
	\edge{v9}{v10}
	\edge{v10}{v11}
\end{graph}
  \caption{}
  \label{minigraph2}
\end{figure}

\begin{tcor} \label{T2RANK2CASEQ1}
Let $G$ be a monotone graph, $j \leq n - 1$ and $q \in V$ such that $\alpha(N(q)) = 2$. Then, the $(q,j,N(q))$-subneighborhood inequality defines a facet of $\ECP$.\\
Moreover, let $S \subset V$ with $\alpha(S) = 2$ and $S \subsetneqq N(q)$.
%\begin{equation*}
%\sum_{v \in S\backslash\{q\}} x_{vj} + 2 x_{qj} + x_{qn} \leq 2 w_j
%\end{equation*}
If $j \leq \lceil n/2 \rceil - 1$ and
for all $v \in N(q) \backslash S$, there exist different vertices $s, s' \in S$
and a stable set $H_v = \{v, s, s'\}$ in $G$ such that:
\begin{itemize}
\item If $n$ is odd, the complement of $G - H_v$ has a perfect matching,
\item If $n$ is even, there exists another stable set $H'_v$ of size 3 in $G$ such that $H_v \cap H'_v = \varnothing$
and the complement of $G - (H_v \cup H'_v)$ has a perfect matching,
\end{itemize}
then the $(q,j,S)$-subneighborhood inequality defines a facet of $\ECP$.
\end{tcor}
\begin{proof}
\textbf{Case $\lceil n/j \rceil \leq \lceil n/\chi_{eq} \rceil$}. The $(q,j,N(q))$-subneighborhood inequality defines
a facet of $\ECP$ since hypotheses (i) and (ii) from Theorem \ref{TNEIGHBOR1} hold trivially.\\
Now, let us consider the $(q,j,S)$-subneighborhood inequality. Since $j \leq \lceil n/2 \rceil - 1$, we have that
$\alpha(S) = 2 \leq \lceil n/j \rceil - 1$. Moreover, hypothesis (i) from Theorem \ref{TNEIGHBOR1}
holds trivially.\\
Let $v \in N(q) \backslash S$ and $M_v$, $H_v = \{v, s, s'\}$ and $H'_v$ (if $n$ is even) be the matching and the stable sets given by the hypothesis. Consider the $(\lceil n/2 \rceil - 1)$-eqcol such that the
color class $C_j$ is $H_v$ and the remaining color classes are $H'_v$ (if $n$ is even) and the endpoints of edges of $M_v$. Then, $|C_j \cap S| = 2$, $(C_j \cap N(q)) \backslash S = \{v\}$ and hypothesis (ii) from Theorem \ref{TNEIGHBOR1} holds. Therefore, the $(q,j,S)$-subneighborhood inequality defines a facet of $\ECP$.\\
\textbf{Case $\lceil n/j \rceil > \lceil n/\chi_{eq} \rceil$}. In virtue of the previous case, we know that the
$(q,\chi_{eq},N(q))$-subneighborhood and the $(q,\chi_{eq},S)$-subneighborhood are facet-defining inequalities of
$\ECP$. Hence, the $(q,j,N(q))$-subneighborhood and the $(q,j,S)$-subneighborhood inequality define facets of $\ECP$ due to Theorem \ref{TIGHT}.
\end{proof}

\subsection{Outside-neighborhood inequalities}

\begin{tthm} \label{TNEIGHBOR2}
Let $G$ be a monotone graph, $u \in V$ such that $N(u)$ is not a clique and $\chi_{eq} \leq j \leq \lfloor n/2 \rfloor$.
If
\begin{enumerate}
\item[(i)] there exists $\hat{v} \in V \backslash N[u]$ not universal in $G - u$, 
\item[(ii)] if $n$ is odd, the complement of $G - u$ has a perfect matching,
\item[(iii)] for all $v \in V \backslash N[u]$, the following conditions hold:
\begin{itemize}
\item if $n$ is even, the complement of $G - \{u, v\}$ has a perfect matching,
\item if $n$ is odd, there exists a stable set $H_v \subset V \backslash \{u, v\}$ of size 3 such that
the complement of $G - (H_v \cup \{u, v\})$ has a perfect matching, % there exists a $\lfloor n/2 \rfloor$-eqcol such that $C_j = \{u, v\}$,
\end{itemize}
\item[(iv)] for all $r$ such that $j \leq r \leq \lfloor n/2 \rfloor$, we have the following:
\begin{itemize}
\item if $\biggl\lfloor \dfrac{n}{r} \biggr\rfloor > \biggl\lfloor \dfrac{n}{r+1} \biggr\rfloor$, then there exists an $r$-eqcol such that
      $C_j \subset N(u)$ and an $r$-eqcol such that $u \in C_j$ and $|C_j| = \lfloor n/r \rfloor$,
\item if $\biggl\lfloor \dfrac{n}{r} \biggr\rfloor = \biggl\lfloor \dfrac{n}{r+1} \biggr\rfloor$, then there exists an $r$-eqcol satisfying conditions given in Remark \ref{NEIGHBOR2POINTS}, \ie lying on the face defined by (\ref{RNEIGHBOR2AGAIN}),
\end{itemize}
\end{enumerate}
then the $(u,j)$-outside-neighborhood inequality, \ie
\begin{equation} \label{RNEIGHBOR2AGAIN}
\biggl(\biggl\lfloor \dfrac{n}{j} \biggr\rfloor - 1 \biggr) x_{uj} - \sum_{v \in V \backslash N[u]} x_{vj}
+ \sum_{k = j+1}^n b_{jk} x_{uk} \leq \sum_{k = j+1}^{n} b_{jk} (w_k - w_{k+1}),
\end{equation}
defines a facet of $\ECP$, where $b_{jk} = \lfloor n/j \rfloor - \lfloor n/k \rfloor$.
\end{tthm}
\begin{proof}
%Since $t_j = j$, (\ref{RNEIGHBOR2AGAIN}) can be regarded as follows:
%\[ \biggl(\biggl\lfloor \dfrac{n}{j} \biggr\rfloor - 1\biggr) x_{uj} -
%\sum_{v \in V \backslash N[u]} x_{vj} + \sum_{k = j+1}^n \biggl(\biggl\lfloor \dfrac{n}{j} \biggr\rfloor
%- \biggl\lfloor \dfrac{n}{k} \biggr\rfloor\biggr) x_{uk} + \sum_{k = j+1}^{n}
%\biggl(\biggl\lfloor \dfrac{n}{k} \biggr\rfloor - \biggl\lfloor \dfrac{n}{k-1} \biggr\rfloor\biggr) w_k \leq 0. \]
%In particular, the coefficient of variable $w_{\lfloor n/2 \rfloor + 1}$ is $-1$.
Let $\F'$ be the face of $\ECP$ defined by (\ref{RNEIGHBOR2AGAIN}) and
$\F = \{ (x,w) \in \ECP ~:~ \lambda^X x + \lambda^W w = \lambda_0\}$ be a face such that $\F' \subset \F$.
According to Remark \ref{TECHNIQUE}, we have to prove that $(\lambda^X, \lambda^W)$ verifies the following equation system: 
%Let us consider the following linear combination: multiply $\alpha_v$ by (\ref{TDIM1}) for $v \in V$,
%$\beta_k$ by (\ref{TDIM2}) for $1 \leq k \leq \chi_{eq}$, $\delta$ by (\ref{TDIM4}) and $\pi$ by
%(\ref{RNEIGHBOR2AGAIN}). If the linear combination is equal to $(\lambda^X, \lambda^W, \lambda_0)$
%(see Remark \ref{TECHNIQUE}), we can clear up $\pi = -\lambda^W_{\lfloor n/2 \rfloor + 1}$, $\delta$, $\alpha_v$
%and $\beta_k$ giving rise to the following conditions:
\begin{enumerate}
\item[(a)] $\lambda^X_{vk} = \lambda^X_{vn} + \lambda^W_n,
                           ~~~\forall~v \in V \backslash N[u],~ 1 \leq k \leq n - 1,~ k \neq j$.
\item[(b)] $\lambda^X_{vk} = \lambda^X_{vn} + \lambda^W_n,
                           ~~~\forall~v \in N(u),~ 1 \leq k \leq n - 1$.
\item[(c)] $\lambda^X_{uj} = \lambda^X_{un} + \lambda^W_n$.
\item[(d)] $\lambda^X_{vj} = \lambda^X_{vn} + \lambda^W_n +
\lambda^W_{\lfloor n/2 \rfloor + 1},~~~\forall~v \in V \backslash N[u]$.
\item[(e)] $\lambda^X_{uk} = \lambda^X_{un} + \lambda^W_n +
(\lfloor n/j \rfloor - 1)\lambda^W_{\lfloor n/2 \rfloor + 1},~~~\forall~1 \leq k \leq j - 1$.
\item[(f)] $\lambda^X_{uk} = \lambda^X_{un} + \lambda^W_n +
(\lfloor n/k \rfloor - 1)\lambda^W_{\lfloor n/2 \rfloor + 1},~~~\forall~j + 1 \leq k \leq n - 1$.
\item[(g)] If $j \neq \chi_{eq}$, then $\lambda^W_k = 0,~~~\forall~\chi_{eq} + 1 \leq k \leq j$.
\item[(h)] $\lambda^W_k = \biggl(\biggl\lfloor \dfrac{n}{k-1} \biggr\rfloor - \biggl\lfloor \dfrac{n}{k} \biggr\rfloor \biggr)
\lambda^W_{\lfloor n/2 \rfloor + 1},~~~\forall~ j + 1 \leq k \leq n - 1,~ k \neq \lfloor n/2 \rfloor+1$.
\end{enumerate}
We present pairs of equitable colorings lying on $\F'$ that allow us to
prove the validity of each equation in the previous system.
\begin{enumerate}
\item[(a)] By hypothesis (i), there exist $\hat{v} \in V \backslash N[u]$ and $\hat{v}' \in V \backslash \{u, \hat{v}\}$ not adjacent to $\hat{v}$.\\
\textbf{Case $v = \hat{v}$}. Let $c^1$ be a $(n-1)$-eqcol such that
$c^1(u) = j$, $c^1(\hat{v}) = c^1(\hat{v}') = k$ and $c^2 = intro(c^1,\hat{v})$.
Then, $\lambda^X_{\hat{v} k} = \lambda^X_{\hat{v} n} + \lambda^W_n$.\\
\textbf{Case $v \neq \hat{v}$}. Let $c^1$ be a $n$-eqcol such that $c^1(u) = j$, $c^1(v) = k$, $c^1(\hat{v}) = n$ and $c^2 = swap_{k,n}(c^1)$. We have $\lambda^X_{vk} + \lambda^X_{\hat{v} n} = \lambda^X_{vn} + \lambda^X_{\hat{v} k}$ and
therefore $\lambda^X_{vk} = \lambda^X_{vn} + \lambda^W_n$.
\item[(b)] Let $u_1, u_2 \in N(u)$ be non adjacent vertices.\\
\textbf{Case $v = u_1$}. Let $c^1$ be a $(n-1)$-eqcol such that $c^1(u_1) = c^1(u_2) = k$. If
$k = j$ we set $c^1(u) = n-1$, otherwise $c^1(u) = j$. Let $c^2 = intro(c^1,u_1)$.
Then, $\lambda^X_{u_1 k} = \lambda^X_{u_1 n} + \lambda^W_n$.\\
\textbf{Case $v \neq u_1$}. Let $c^1$ be a $n$-eqcol such that $c^1(v) = k$ and $c^1(u_1) = n$. If $k = j$ we set
$c^1(u) = n-1$, otherwise $c^1(u) = j$. Let $c^2 = swap_{k,n}(c^1)$.
We have $\lambda^X_{vk} + \lambda^X_{u_1 n} = \lambda^X_{vn} + \lambda^X_{u_1 k}$ and therefore
$\lambda^X_{vk} = \lambda^X_{vn} + \lambda^W_n$.
\item[(c)] Let $v \in N(u)$, $c^1$ be a $n$-eqcol such that $c^1(u) = j$, $c^1(v) = n$ and
$c^2 = swap_{j,n}(c^1)$. In virtue of condition (b), we obtain $\lambda^X_{uj} = \lambda^X_{un} + \lambda^W_n$.
\item[(d)]  \textbf{Case $n$ even}. Let $M_v$ be the matching given by hypothesis (iii). Let $c^1$ be the
$\lfloor n/2 \rfloor$-eqcol whose color classes are the endpoints of $M_v$ and $C_j = \{u, v\}$.
Let $c^2 = intro(c^1,v)$. We deduce that
$\lambda^X_{vj} = \lambda^X_{v\lfloor n/2 \rfloor + 1} + \lambda^W_{\lfloor n/2 \rfloor + 1}
= \lambda^X_{vn} + \lambda^W_n + \lambda^W_{\lfloor n/2 \rfloor + 1}$.\\
\textbf{Case $n$ odd}. Let $M_v$ and $H_v$ be the matching and the stable set given by hypothesis (iii). Let
$c^1$ be the $\lfloor n/2 \rfloor$-eqcol whose color classes are $H_v$, the endpoints of $M_v$ and $C_j = \{u, v\}$.
Now, let $M$ be the matching given by hypothesis (ii) and let $v'$ be a vertex such that $(v, v')$ belongs to $M$.
Let $c^2$ be the $(\lfloor n/2 \rfloor + 1)$-eqcol
whose color classes are the endpoints of $M \backslash (v, v')$, $C_j = \{u\}$ and
$C_{\lfloor n/2 \rfloor+1} = \{ v, v' \}$.
Thus,
\[ \lambda^X_{vj} + \sum_{i \in V \backslash \{u, v\}} \lambda^X_{ic^1(i)} =
\lambda^X_{v\lfloor n/2 \rfloor + 1} + \sum_{i \in V \backslash \{u, v\}} \lambda^X_{ic^2(i)}
+ \lambda^W_{\lfloor n/2 \rfloor + 1}. \]
Conditions (a) and (b) allow us to reach the desired result.
\item[(e)] Let us notice that, if $r = \lfloor n/\lfloor n/j\rfloor \rfloor$ then
$\lfloor n/j \rfloor = \lfloor n/r \rfloor$, $j \leq r \leq \lfloor n/2 \rfloor$ and
$\lfloor \frac{n}{r} \rfloor > \lfloor \frac{n}{r+1} \rfloor$. By hypothesis (iv), there exists an $r$-eqcol
$c$ such that $N(u)$ contains all the vertices painted with color $j$. Let $c^1 = swap_{c(u),k}(c)$ and
$c^2$ be the $r$-eqcol that paints vertex $u$ and $\lfloor n/j \rfloor-1$ vertices of
$V \backslash N[u]$ with color $j$ also given by hypothesis (iv).
By condition (c), we have $\lambda^X_{uk} + \sum_{v \in V \backslash \{u\}} \lambda^X_{vc^1(v)} =
\lambda^X_{un} + \lambda^W_n + \sum_{v \in V \backslash \{u\}} \lambda^X_{vc^2(v)}$.
Applying conditions (a), (b) and (d), we get
$\lambda^X_{uk} = \lambda^X_{un} + \lambda^W_n +
(\lfloor n/j \rfloor - 1)\lambda^W_{\lfloor n/2 \rfloor + 1}$.
\item[(f)] \textbf{Case $k \leq \lfloor n/2 \rfloor$}. We proceed in the same way as in (e),
but using $r = \lfloor n/\lfloor n/k\rfloor \rfloor$ instead of $\lfloor n/\lfloor n/j\rfloor \rfloor$.\\
\textbf{Case $k \geq \lfloor n/2 \rfloor + 1$}. Then, $\lfloor n/k \rfloor = 1$. Let $v \in N(u)$, $c^1$ be a
$n$-eqcol such that $c^1(u) = k$, $c^1(v) = j$ and $c^2 = swap_{k,j}(c^1)$.
Conditions (b) and (c) allow us to obtain $\lambda^X_{uk} = \lambda^X_{un} + \lambda^W_n$.
\item[(g)-(h)] This condition can be verified by providing a $k$-eqcol $(x^1,w^1)$ and a $(k-1)$-eqcol
$(x^2,w^2)$ lying on $\F'$ and applying conditions (a)-(f) to equation $\lambda^X x^1 + \lambda^W_k = \lambda^X x^2$.\\
Thus, we only need to prove that, for any $\chi_{eq} \leq r \leq n-1$, there exists an $r$-eqcol $c$ lying on $\F'$ .\\ 
\textbf{Case $r < j$}. The existence of $c$ is guaranteed by the monotonicity of G.\\
\textbf{Case $j \leq r \leq \lfloor n/2 \rfloor$}. The existence of $c$ is guaranteed by hypothesis (iv).\\
\textbf{Case $r = \lfloor n/2 \rfloor + 1$}. $c$ may be the $(\lfloor n/2 \rfloor + 1)$-eqcol yielded by condition (d).\\
\textbf{Case $\lfloor n/2 \rfloor+2 \leq r \leq n-1$}. Let us consider the
$(\lfloor n/2 \rfloor + 1)$-eqcol yielded in the previous case and let $v_1, v_2$ be vertices sharing a color
different from $j$. In order to generate a $(\lfloor n/2 \rfloor + 2)$-eqcol $c$, we introduce a new color on $v_1$,
\ie $c = intro(c',v_1)$ where $c'$ is the $(\lfloor n/2 \rfloor + 1)$-eqcol. By repeating this
procedure, we can generate a $(\lfloor n/2 \rfloor + 3)$-eqcol and so on.
\end{enumerate}
\end{proof}

Let us present an example where the previous theorem is applied.\\

\noindent \textbf{Example.} Let $G$ be the graph given in Figure \ref{minigraph2}(a). Let us recall that $G$ is
monotone and $\chi_{eq}(G) = 3$.
%It is not hard to see that the assumptions of Theorem \ref{TNEIGHBOR2} are satisfied in order to
%assert that the $(1,3)$-outside-neighborhood inequality, \ie
%\[ 2 x_{1,3} - \sum_{v=7}^{11} x_{v,3} + x_{1,4} + x_{1,5} + 2 x_{1,6} + 2 x_{1,7} + 2 x_{1,8} + 2 x_{1,9}
% + 2 x_{1,10} + 2 x_{1,11} \leq w_4 + w_6, \]
%defines a facet of $\ECP(G)$ if $G$ is the graph presented in Figure \ref{minigraph2}(a).
We apply Theorem \ref{TNEIGHBOR2} considering $u = 1$ and $j = 3$. It is not hard to see that the assumptions of
this theorem are satisfied. Below, we present some examples of colorings related to hypothesis (iv) of
Theorem \ref{TNEIGHBOR2}. Figure \ref{minigraph3}(a) shows a 3-eqcol of $G$ such that $C_3 \subset N(1)$
and Figure \ref{minigraph3}(b) shows a 3-eqcol of $G$ such that $1 \in C_3$ and $|C_3| = 3$.
%(a) a $(\lfloor \frac{11}{2} \rfloor)$-eqcol such that $C_3 = \{1,7\}$ *** NO LO NECESITO ***,
%(b) a 4-eqcol that lies on the face,

By Theorem \ref{TIGHT2}, $(1,j)$-outside-neighborhood inequalities with $j \in \{1,2\}$ are also facet-defining.

\begin{figure}[h]
  \centering ~~~~~~~~~~~
\begin{graph}(5.8, 3)(-2, -1)
	\freetext(0.8,-0.75){(a) 3-eqcol, $C_3 \subset N(1)$}
	\roundnode{v1}(-2,0)
	\autonodetext{v1}[w]{1}
	\roundnode{v2}(-1.2,0)
	\autonodetext{v2}[n]{3}
	\roundnode{v3}(-1.7528,0.7608)
	\autonodetext{v3}[n]{3}
	\roundnode{v4}(-2.6472,0.4702)
	\autonodetext{v4}[nw]{3}
	\roundnode{v5}(-2.6472,-0.4702)
	\autonodetext{v5}[sw]{3}
	\roundnode{v6}(-1.7528,-0.7608)
	\autonodetext{v6}[s]{2}
	\roundnode{v7}(-0.4,0)
	\autonodetext{v7}[n]{1}
	\roundnode{v8}(0.4,0)
	\autonodetext{v8}[n]{2}
	\roundnode{v9}(1.2,0)
	\autonodetext{v9}[n]{1}
	\roundnode{v10}(2,0)
	\autonodetext{v10}[n]{2}
	\roundnode{v11}(2.8,0)
	\autonodetext{v11}[n]{1}
	\edge{v1}{v2}
	\edge{v1}{v3}
	\edge{v1}{v4}
	\edge{v1}{v5}
	\edge{v1}{v6}
	\edge{v2}{v7}
	\edge{v7}{v8}
	\edge{v8}{v9}
	\edge{v9}{v10}
	\edge{v10}{v11}
\end{graph}~~~~~~~
\begin{graph}(5.8, 3)(-2, -1)
	\freetext(1.0,-0.75){(b) 3-eqcol, $1 \in C_3$, $|C_3| = 3$}
	\roundnode{v1}(-2,0)
	\autonodetext{v1}[w]{3}
	\roundnode{v2}(-1.2,0)
	\autonodetext{v2}[n]{2}
	\roundnode{v3}(-1.7528,0.7608)
	\autonodetext{v3}[n]{1}
	\roundnode{v4}(-2.6472,0.4702)
	\autonodetext{v4}[nw]{1}
	\roundnode{v5}(-2.6472,-0.4702)
	\autonodetext{v5}[sw]{1}
	\roundnode{v6}(-1.7528,-0.7608)
	\autonodetext{v6}[s]{2}
	\roundnode{v7}(-0.4,0)
	\autonodetext{v7}[n]{3}
	\roundnode{v8}(0.4,0)
	\autonodetext{v8}[n]{2}
	\roundnode{v9}(1.2,0)
	\autonodetext{v9}[n]{3}
	\roundnode{v10}(2,0)
	\autonodetext{v10}[n]{2}
	\roundnode{v11}(2.8,0)
	\autonodetext{v11}[n]{1}
	\edge{v1}{v2}
	\edge{v1}{v3}
	\edge{v1}{v4}
	\edge{v1}{v5}
	\edge{v1}{v6}
	\edge{v2}{v7}
	\edge{v7}{v8}
	\edge{v8}{v9}
	\edge{v9}{v10}
	\edge{v10}{v11}
\end{graph}
  \caption{}
  \label{minigraph3}
\end{figure}

\subsection{Clique-neighborhood inequalities}

\begin{tthm} \label{TNEIGHBOR3}
Let $G$ be a monotone graph, $u \in V$, $Q$ be a clique of $G$ such that $Q \cap N[u]=\varnothing$ and $j,k$
be numbers verifying $3 \leq k \leq \min \{\alpha(N(u)) + 1, \lceil n/\chi_{eq} \rceil\}$ and
$1 \leq j \leq \biggl\lceil \dfrac{n}{k-1} \biggr\rceil - 1$. If
\begin{enumerate}
\item[(i)] for all $v \in V \backslash (N[u] \cup Q)$, there exist
$\lceil n/3 \rceil \leq r \leq \lceil n/2 \rceil - 1$, $q_1, q_2 \in V$ and two $r$-eqcols such that in one of them
$C_j = \{u, v, q_1\}$ and in the other $C_j = \{u , q_2\}$ ($q_1$ and $q_2$ may be the same vertex),
\item[(ii)] for all $t$ such that $max \{ j, \chi_{eq} \} \leq t \leq n - 3$, we have the following:
\begin{itemize}
\item if $\biggl\lceil \dfrac{n}{t} \biggr\rceil > \biggl\lceil \dfrac{n}{t+1} \biggr\rceil$, there exist $q \in Q$ and
a $t$-eqcol such that $C_j \subset N(u)$, $|C_j| = \lceil n/t \rceil$ and $u, q \in C_t$,
\item if $\biggl\lceil \dfrac{n}{t} \biggr\rceil = \biggl\lceil \dfrac{n}{t+1} \biggr\rceil$, there exists a $t$-eqcol
satisfying conditions given in Remark \ref{NEIGHBOR3POINTS}, \ie lying on the face defined by (\ref{RNEIGHBOR3AGAIN}),
\end{itemize} %where $L = \{\lceil \frac{n}{i} \rceil - 1 ~:~ 2 \leq i \leq k - 1\}$, and
\end{enumerate}
then the $(u,j,k,Q)$-clique-neighborhood inequality, \ie
\begin{multline} \label{RNEIGHBOR3AGAIN}
(k - 1) x_{uj} +
\sum_{l = \lceil \frac{n}{k-1} \rceil}^{n-2} \biggl(k - \biggl\lceil \dfrac{n}{l} \biggr\rceil \biggr) x_{ul} +
(k - 1) \bigl(x_{u n-1} + x_{un} \bigr) +
\!\!\!\sum_{v \in N(u)\cup Q}\!\!\!x_{vj} \\ + \sum_{v \in V \backslash \{u\}}\!\!\!(x_{v n-1} + x_{vn}) \leq
\sum_{l = j}^n b_{ul} (w_l - w_{l+1}),
\end{multline}
defines a facet of $\ECP$, where
\[ b_{ul} = \begin{cases}
 \min\{\lceil n / l \rceil, \alpha(N(u)) + 1\}, &\textrm{if $j \leq l \leq \lceil n/k \rceil - 1$} \\
 k, &\textrm{if $\lceil n/k \rceil \leq l \leq n - 2$} \\
 k + 1, &\textrm{if $l \geq n - 1$}
\end{cases} \]
\end{tthm}
\begin{proof}
%Inequality (\ref{RNEIGHBOR3AGAIN}) can be regarded as follows:
%\[ (k - 1) x_{uj} +
%\sum_{l = \lceil \frac{n}{k-1} \rceil}^{n-2} \biggl(k - \biggl\lceil \dfrac{n}{l} %\biggr\rceil \biggr) x_{ul} +
%(k - 1) \bigl(x_{u n-1} + x_{un} \bigr) +
%\!\!\!\sum_{v \in N(u)\cup Q}\!\!\!x_{vj} + \sum_{v \in V \backslash \{u\}}\!\!\!(x_{v n-1} + x_{vn})
% + \sum_{l = 2}^{n-2} \bigl(b_{ul-1} - b_{ul} \bigr) w_l - w_{n-1} \leq 0, \]
%where we set $b_{ul} = 0$ when $l < j$.
Let $\F'$ be the face of $\ECP$ defined by (\ref{RNEIGHBOR3AGAIN}) and
$\F = \{ (x,w) \in \ECP ~:~ \lambda^X x + \lambda^W w = \lambda_0\}$ be a face such that $\F' \subset \F$.
According to Remark \ref{TECHNIQUE}, we have to prove that $(\lambda^X, \lambda^W)$ verifies the following equation system: 
%Let us consider the following linear combination: multiply $\alpha_v$ by (\ref{TDIM1}) for $v \in V$,
%$\beta_j$ by (\ref{TDIM2}) for $1 \leq j \leq \chi_{eq}$, $\delta$ by (\ref{TDIM4}) and $\pi$ by
%(\ref{RNEIGHBOR3AGAIN}). If the linear combination is equal to $(\lambda^X, \lambda^W, \lambda_0)$
%(see Remark \ref{TECHNIQUE}), we can clear up $\pi = -\lambda^W_{n-1}$, $\delta$, $\alpha_v$ and $\beta_j$,
%giving rise to the following conditions:
\begin{enumerate}
\item[(a)] $\lambda^X_{uj} = \lambda^X_{un} + \lambda^W_n$.
\item[(b)] $\lambda^X_{vn-1} = \lambda^X_{vn} + \lambda^W_n,~~~ \forall~ v \in V$.
\item[(c)] $\lambda^X_{vr} = \lambda^X_{vn-1} + \lambda^W_{n-1},~~~ \forall~ v \in V \backslash \{u\},~ 1 \leq r \leq n - 2,~ r \neq j$.
\item[(d)] $\lambda^X_{vj} = \lambda^X_{vn} + \lambda^W_n,~~~ \forall~ v \in N(u) \cup Q$.
\item[(e)] $\lambda^X_{vj} = \lambda^X_{vn-1} + \lambda^W_{n-1},~~~ \forall~ v \in V \backslash (N[u] \cup Q)$.
\item[(f)] $\lambda^X_{ur} = \lambda^X_{un} + (k - 1) \lambda^W_{n-1} + \lambda^W_n,~~~ \forall~ 1 \leq r \leq \lceil \frac{n}{k-1} \rceil - 1,~ r \neq j$.
\item[(g)] $\lambda^X_{ur} = \lambda^X_{un} + (\lceil n / r \rceil - 1) \lambda^W_{n-1} + \lambda^W_n,~~~ \forall~ \lceil \frac{n}{k-1} \rceil \leq r \leq n - 2$.
\item[(h)] $\lambda^W_r = (b_{ur} - b_{ur-1}) \lambda^W_{n-1},~~~ \forall~ \chi_{eq} + 1 \leq r \leq n-2$.
\end{enumerate}
We present pairs of equitable colorings lying on $\F'$ that allow us to
prove the validity of each equation in the previous system.
\begin{enumerate}
\item[(a)] Let $q \in Q$, $c^1$ be a $(n-1)$-eqcol such that $c^1(u) = c^1(q) = j$ and $c^2 = intro(c^1,u)$.
Then, $\lambda^X_{uj} = \lambda^X_{un} + \lambda^W_n$.
\item[(b)] \textbf{Case $v = u$}. Let $q \in Q$, $w \in N(u) \cup Q \backslash \{q\}$, $c^1$ be a $(n-1)$-eqcol
such that $c^1(u) = c^1(q) = n-1$, $c^1(w) = j$ and $c^2 = intro(c^1,u)$. Then, $\lambda^X_{un-1} = \lambda^X_{un} + \lambda^W_n$.\\
Now, let $v_1, v_2 \in N(u)$ be non adjacent vertices.\\
\textbf{Case $v = v_1$}. Let $c^1$ be a $(n-1)$-eqcol such that $c^1(v_1) = c^1(v_2) = n-1$,
$c^1(u) = j$ and $c^2 = intro(c^1,v_1)$. We have $\lambda^X_{v_1 n-1} = \lambda^X_{v_1 n} + \lambda^W_n$.\\
\textbf{Case $v \in V \backslash \{u, v_1\}$}. Let $c^1$ be a $n$-eqcol such that $c^1(u) = j$, $c^1(v) = n-1$, $c^1(v_1) = n$ and $c^2 = swap_{n-1,n}(c^1)$.
We have $\lambda^X_{vn-1} + \lambda^X_{v_1 n} = \lambda^X_{vn} + \lambda^X_{v_1 n-1}$ and, since
$\lambda^X_{v_1 n-1} = \lambda^X_{v_1 n} + \lambda^W_n$, we obtain $\lambda^X_{v n-1} = \lambda^X_{v n} + \lambda^W_n$.
\item[(c)] Let $v_1, v_2 \in N(u)$ be non adjacent vertices and $q \in Q$.\\
\textbf{Case $v = v_1$}. Let $c^1$ be a $(n-2)$-eqcol such that $c^1(v_1) = c^1(v_2) = r$,
$c^1(u) = c^1(q) = j$ and $c^2 = intro(c^1,v_1)$. Then, $\lambda^X_{v_1 r} = \lambda^X_{v_1 n-1} + \lambda^W_{n-1}$.\\
\textbf{Case $v \neq v_1$}. Let $c^1$ be a $n$-eqcol such that $c^1(u) = j$, $c^1(v) = r$, $c^1(v_1) = n-1$ and
$c^2 = swap_{r,n-1}(c^1)$. We have $\lambda^X_{vr} + \lambda^X_{v_1 n-1} = \lambda^X_{vn-1} + \lambda^X_{v_1 r}$ and,
since $\lambda^X_{v_1 r} = \lambda^X_{v_1 n-1} + \lambda^W_{n-1}$, we obtain
$\lambda^X_{vr} = \lambda^X_{vn-1} + \lambda^W_{n-1}$.
\item[(d)] Let $q \in Q$.\\
\textbf{Case $v = q$}. Let $c^1$ be a $(n-1)$-eqcol such that $c^1(u) = c^1(q) = j$ and
$c^2 = intro(c^1,q)$. Then, $\lambda^X_{qj} = \lambda^X_{qn} + \lambda^W_n$.\\
\textbf{Case $v \neq q$}. Let $c^1$ be a $n$-eqcol such that $c^1(u) = n-1$, $c^1(q) = n$, $c^1(v) = j$ and $c^2 = swap_{j,n}(c^1)$.
We have $\lambda^X_{vj} + \lambda^X_{qn} = \lambda^X_{vn} + \lambda^X_{qj}$ and, since $\lambda^X_{qj} = \lambda^X_{qn} + \lambda^W_n$, we obtain $\lambda^X_{vj} = \lambda^X_{vn} + \lambda^W_n$.
\item[(e)] Hypothesis (i) ensures that there exists an equitable coloring $c^1$ such that
$c^1(u) = c^1(v) = c^1(q_1) = j$ and the remaining vertices do not use color $j$, and there exists
another equitable coloring $c^2$ (with the same number of colors) such that $c^2(u) = c^2(q_2) = j$
and the remaining vertices do not use color $j$, where $q_1, q_2 \in Q$. % (may be the same vertex).
We have
\[ \sum_{w \in V\backslash\{u,v,q_1\}} \lambda^X_{w c^1(w)} + \lambda^X_{q_1 j} + \lambda^X_{vj}
= \sum_{w \in V\backslash\{u,v,q_2\}} \lambda^X_{w c^2(w)} + \lambda^X_{q_2 j} + \lambda^X_{vc^2(v)} \]
and, by conditions (b)-(d), we derive
$\lambda^X_{vj} = \lambda^X_{vc^2(v)} = \lambda^X_{vn} + \lambda^W_{n-1}$.
\item[(f)] Let $t = \lceil \frac{n}{k - 1} \rceil - 1$. Clearly, $max \{ j, \chi_{eq} \} \leq t \leq n - 3$ and
$\lceil \frac{n}{t} \rceil > \lceil \frac{n}{t+1} \rceil$. By hypothesis (ii), there exists a $t$-eqcol $c$ whose
class color $C_j$ satisfies $C_j \subset N(u)$ and $|C_j| = \lceil n/t \rceil$ and $u$ and a vertex of $Q$ use
color $t$. Let $c^1 = swap_{j,t}(c)$ and $c^2 = swap_{r,t}(c)$ (since $t \geq j$ and $t \geq r$, both colorings are
well-defined). Hence, $c^1(u) = j$ and $c^2(u) = r$. We apply conditions proved before to
$\lambda^X x^1 = \lambda^X x^2$, where $x^1$ and $x^2$ are the binary variables representing colorings
$c^1$ and $c^2$ respectively, and we conclude that
$\lambda^X_{ur} = \lambda^X_{un} + (k - 1) \lambda^W_{n-1} + \lambda^W_n$.
\item[(g)] \textbf{Case $r \leq \lceil \frac{n}{2} \rceil-1$}. We proceed in the same way as in (f), but using
$t = \lceil \frac{n}{\lceil n/r \rceil - 1} \rceil - 1$ instead of
$\lceil \frac{n}{k - 1} \rceil - 1$.\\
\textbf{Case $r \geq \lceil \frac{n}{2} \rceil$}. Let $v_1, v_2 \in N(u)$ be non adjacent vertices and $q \in Q$. Let $c^1$ be a $(n-2)$-eqcol such that
$c^1(v_1) = c^1(v_2) = r$, $c^1(u) = c^1(q) = j$ and $c^2 = swap_{j,r}(c^1)$. We apply conditions proved before
to $\lambda^X x^1 = \lambda^X x^2$, where $x^1$ and $x^2$ are the binary variables representing colorings
$c^1$ and $c^2$ respectively, and we conclude that $\lambda^X_{ur} = \lambda^X_{un} + \lambda^W_{n-1} + \lambda^W_n$.
\item[(h)] This condition can be verified by providing an $r$-eqcol $(x^1,w^1)$ and an $(r-1)$-eqcol
$(x^2,w^2)$ lying on $\F'$ and applying conditions (a)-(g) to equation $\lambda^X x^1 + \lambda^W_r = \lambda^X x^2$.\\
Thus, we only need to prove that, for any $\chi_{eq} \leq t \leq n-2$, there exists a $t$-eqcol $c$ lying on $\F'$ .\\ 
\textbf{Case $t < j$}. The existence of $c$ is guaranteed by the monotonicity of G.\\
\textbf{Case $j \leq t \leq n-3$}. The existence of $c$ is guaranteed by hypothesis (ii).\\
\textbf{Case $t=n-2$}. $c$ may be the $(n-2)$-eqcol yielded by condition (c).
\end{enumerate}
\end{proof}

\begin{tcor} \label{TNEIGHBOR3COR}
Let $G$ be a monotone graph and let $u$, $j$, $k$, $Q$ be defined as in Theorem \ref{TNEIGHBOR3}.
If hypothesis (ii) of Theorem \ref{TNEIGHBOR3} holds and for all $v \in V \backslash (N[u] \cup Q)$:
\begin{itemize}
\item if $n$ is odd,
\begin{itemize}
\item there exists a vertex $q_1 \in Q$ and a stable set $H^1_v = \{u, v, q_1\}$ such that the complement of $G - H^1_v$
has a perfect matching $M_v$,
\item there exists a vertex $q_2 \in Q$ and two disjoint stable sets $H^2_v = \{u, q_2\}$, $H^3_v$ such that
$|H^3_v| = 3$ and the complement of $G - (H^2_v \cup H^3_v)$ has a perfect matching  $M'_v$,
\end{itemize}
\item if $n$ is even,
\begin{itemize}
\item there exists a vertex $q_1 \in Q$ and two disjoint stable sets $H^1_v = \{u, v, q_1\}$, $H^2_v$ such that
$|H^2_v| = 3$ and the complement of $G - (H^1_v \cup H^2_v)$ has a perfect matching $M_v$,
\item there exists a vertex $q_2 \in Q$ and three disjoint stable sets $H^3_v = \{u, q_2\}$, $H^4_v$, $H^5_v$
such that $|H^4_v| = |H^5_v| = 3$ and the complement of $G - (H^3_v \cup H^4_v \cup H^5_v)$ has a perfect matching  $M'_v$,
\end{itemize}
\end{itemize}
then the $(u,j,k,Q)$-clique-neighborhood inequality defines a facet of $\ECP$.
\end{tcor}
\begin{proof}
Let us suppose that $n$ is odd.
Let $v \in V \backslash (N[u] \cup Q)$ and let $M_v$, $M'_v$, $H^1_v$, $H^2_v$ and $H^3_v$ be the matchings and the
stable sets given in the hypothesis. Consider an $(\lceil n/2 \rceil - 1)$-eqcol such that the color class
$j$ is $H^1_v$ and the remaining color classes are the endpoints of edges of $M_v$, and an
$(\lceil n/2 \rceil - 1)$-eqcol such that the color class $j$ is $H^2_v$ and the remaining color classes are $H^3_v$
and the endpoints of edges of $M'_v$. Therefore, hypothesis (i) of Theorem \ref{TNEIGHBOR3} holds and
the $(u,j,k,Q)$-clique-neighborhood inequality defines a facet of $\ECP$.

The proof for $n$ even is analogous to the previous one.
\end{proof}

Let us present an example where the previous result is applied.\\

\noindent \textbf{Example.} Let $G$ be the graph given in Figure \ref{minigraph2}(a). Let us recall that $G$ is
monotone and $\chi_{eq}(G) = 3$. We apply Corollary \ref{TNEIGHBOR3COR} considering $u = 1$, $j = 1$, $k = 4$ and
$Q = \{7,8\}$. It is not hard to see that the assumptions of this corollary are satisfied.
Below, we present some examples of colorings related to hypothesis (ii) of Theorem \ref{TNEIGHBOR3}.
Figure \ref{minigraph4}(a) shows a 3-eqcol of $G$ such that $1,7 \in C_3$, $C_1 \subset N(1)$, $|C_1| = 4$ and
Figure \ref{minigraph4}(b) shows a 5-eqcol of $G$ such that $1,7 \in C_5$, $C_1 \subset N(1)$, $|C_1| = 3$.
%It is not hard to see that the assumptions of Theorem \ref{TNEIGHBOR3} are satisfied in order to
%assert that the $(1,1,4,\{7,8\})$-clique-neighborhood inequality, \ie
%$3 x_{1,1} + \sum_{v=2}^8 x_{v,1} + x_{1,4} + x_{1,5} + 2 x_{1,6} + 2 x_{1,7} + 2 x_{1,8} + 2 x_{1,9}
%  + 3 x_{1,10} + 3 x_{1,11} + \sum_{v=2}^{11} x_{v,10} + \sum_{v=2}^{11} x_{v,11} \leq 6 w_1 - 2 w_3 + w_{10}$,
%defines a facet of $\ECP(G)$ if $G$ is the graph presented in Figure \ref{minigraph2}(a).
%Some of the colorings required by the theorem are shown in Figure \ref{minigraph4}:
%(c) a 4-eqcol such that $C_1 = \{1, 7, 9\}$, and    *** NO LOS NECESITO ***
%(d) a 4-eqcol such that $C_1 = \{1, 7\}$.

\begin{figure}[h]
  \centering ~~~~~~~~~~~
\begin{graph}(5.8, 2)(-2, -1)
	\freetext(0.6,-0.75){(a) 3-eqcol of $G$}
	\roundnode{v1}(-2,0)
	\autonodetext{v1}[w]{3}
	\roundnode{v2}(-1.2,0)
	\autonodetext{v2}[n]{1}
	\roundnode{v3}(-1.7528,0.7608)
	\autonodetext{v3}[n]{1}
	\roundnode{v4}(-2.6472,0.4702)
	\autonodetext{v4}[nw]{1}
	\roundnode{v5}(-2.6472,-0.4702)
	\autonodetext{v5}[sw]{1}
	\roundnode{v6}(-1.7528,-0.7608)
	\autonodetext{v6}[s]{2}
	\roundnode{v7}(-0.4,0)
	\autonodetext{v7}[n]{3}
	\roundnode{v8}(0.4,0)
	\autonodetext{v8}[n]{2}
	\roundnode{v9}(1.2,0)
	\autonodetext{v9}[n]{3}
	\roundnode{v10}(2,0)
	\autonodetext{v10}[n]{2}
	\roundnode{v11}(2.8,0)
	\autonodetext{v11}[n]{3}
	\edge{v1}{v2}
	\edge{v1}{v3}
	\edge{v1}{v4}
	\edge{v1}{v5}
	\edge{v1}{v6}
	\edge{v2}{v7}
	\edge{v7}{v8}
	\edge{v8}{v9}
	\edge{v9}{v10}
	\edge{v10}{v11}
\end{graph}~~~~~~~
\begin{graph}(5.8, 2)(-2, -1)
	\freetext(0.6,-0.75){(b) 5-eqcol of $G$}
	\roundnode{v1}(-2,0)
	\autonodetext{v1}[w]{5}
	\roundnode{v2}(-1.2,0)
	\autonodetext{v2}[n]{1}
	\roundnode{v3}(-1.7528,0.7608)
	\autonodetext{v3}[n]{1}
	\roundnode{v4}(-2.6472,0.4702)
	\autonodetext{v4}[nw]{1}
	\roundnode{v5}(-2.6472,-0.4702)
	\autonodetext{v5}[sw]{2}
	\roundnode{v6}(-1.7528,-0.7608)
	\autonodetext{v6}[s]{4}
	\roundnode{v7}(-0.4,0)
	\autonodetext{v7}[n]{5}
	\roundnode{v8}(0.4,0)
	\autonodetext{v8}[n]{3}
	\roundnode{v9}(1.2,0)
	\autonodetext{v9}[n]{2}
	\roundnode{v10}(2,0)
	\autonodetext{v10}[n]{3}
	\roundnode{v11}(2.8,0)
	\autonodetext{v11}[n]{4}
	\edge{v1}{v2}
	\edge{v1}{v3}
	\edge{v1}{v4}
	\edge{v1}{v5}
	\edge{v1}{v6}
	\edge{v2}{v7}
	\edge{v7}{v8}
	\edge{v8}{v9}
	\edge{v9}{v10}
	\edge{v10}{v11}
\end{graph}
  \caption{}
  \label{minigraph4}
\end{figure}

\subsection{$S$-color inequalities}

\begin{tthm} \label{TONLYCOLORS}
Let $M$ be a matching of the complement of $G$ such that $2 \leq |M| \leq \lfloor \frac{n-1}{2} \rfloor$
and let $S \subset \{1,\ldots,n\}$ such that $|S| = 2|M| - r$ with $r \in \{1,2\}$ and $S$ contains all the colors
greater than $n - |M|$. Then, the $S$-color inequality, \ie
\begin{equation} \label{RONLYCOLORSAGAIN}
\sum_{j \in S} \sum_{v \in V} x_{vj} \leq \sum_{k=1}^{n} b_{Sk} (w_k - w_{k+1}),
\end{equation}
defines a facet of $\ECP$, where
$$b_{Sk} = |S \cap \{1,\ldots,k\}| \biggl\lfloor \dfrac{n}{k} \biggr\rfloor + \min \biggl\{ |S \cap \{1,\ldots,k\}|, n - k \biggl\lfloor \dfrac{n}{k} \biggr\rfloor \biggr\}.$$
\end{tthm}
\begin{proof}
Let us note that $|S| \geq 2$. If $|S| = 2$, $S = \{n-1,n\}$ and the $S$-color inequality defines the same
face as the $\{n-1\}$-color inequality as stated in Remark \ref{REMARKONLY}.1. But the $\{n-1\}$-color inequality
is the constraint (\ref{RUPPER}) with $j = n-1$ by Remark \ref{REMARKONLY}.2, which is facet-defining by Theorem \ref{TORIGINAL4}. Then, the $S$-color inequality defines a facet of $\ECP$. So, from now on we assume that
$|S| \geq 3$.

For the sake of simplicity, we define $p = n - |M|$.

Now, let $\F'$ be the face of $\ECP$ defined by (\ref{RONLYCOLORSAGAIN}) and
$\F = \{ (x,w) \in \ECP ~:~ \lambda^X x + \lambda^W w = \lambda_0\}$ be a face such that $\F' \subset \F$.
According to Remark \ref{TECHNIQUE}, we have to prove that $(\lambda^X, \lambda^W)$ verifies the following equation system: 
% and consider the following linear combination: multiply $\alpha_v$ by (\ref{TDIM1}) for
%$v \in V$, $\beta_j$ by (\ref{TDIM2}) for $1 \leq j \leq \chi_{eq}$, $\gamma_j$ by (\ref{TDIM3}) for $j \in \S$
%(note that $G$ could not be monotone), $\delta$ by (\ref{TDIM4}) and $\pi$ by (\ref{RONLYCOLORSAGAIN}). 
%If the linear combination is equal to $(\lambda^X, \lambda^W, \lambda_0)$ (see Remark \ref{TECHNIQUE}), we can clear
%up coefficients as follows: $\pi = - \frac{1}{r} \lambda^W_{p+1}$, $\delta = - \lambda^W_n$,
%$\beta_j = \lambda^W_j$, $\alpha_v = \lambda^X_{vn} + \lambda^W_n + \frac{1}{r} \lambda^W_{p+1}$ and
%$\gamma_j = \sum_{k = \tau+1}^j \lambda^W_j + (b_\tau - b_j) \frac{1}{r} \lambda^W_{p+1}$ where $\tau$ is
%the largest integer less than $j$ that does not belong to $\S$ (in other words, $\tau \notin \S$
%but $\tau+1, \ldots, j \in \S$), giving rise to the following conditions:
\begin{align*}
&\textrm{(a)}~~ \lambda^X_{vj} = \lambda^X_{vn} + \lambda^W_n,~~~\forall~v \in V,~
                   j \in S \backslash \{n\}. & \\
&\textrm{(b)}~~ \lambda^X_{vj} = \lambda^X_{vn} + \lambda^W_n + \frac{1}{r} \lambda^W_{p+1},~~~\forall~v \in V,~
                   j \notin S. & \\
&\textrm{(c)}~~ \sum_{k = \theta+1}^j \lambda^W_k = (b_{Sj} - b_{S\theta}) \frac{1}{r} \lambda^W_{p+1}, ~~~\forall~\chi_{eq}+1 \leq j \leq n - 1 ~\textrm{such that} &  \\
&~~~~~~~~~~~~~~~~ j \neq p+1,~ j \notin \S ~\textrm{and}~ \theta = \max \{j' \in \mathbb{Z} : j' \leq j-1,~j' \notin \S\}. &
\end{align*}
We present pairs of equitable colorings lying on $\F'$ that allow us to
prove the validity of each equation in the previous system.
\begin{enumerate}
\item[(a)] Let $v'$ be a vertex not adjacent to $v$. It exists since $G$ does not have universal vertices.
Let $c^1$ be a $(n-1)$-eqcol such that $c^1(v) = c^1(v') = j$ and $c^2 = intro(c^1,v)$. We conclude that
$\lambda^X_{v j} = \lambda^X_{v n} + \lambda^W_n$.
\item[(b)] Since $j \notin S$, we know that $j \leq p$ so we can propose $p$-colorings using $j$. % In addition, those colorings will be equitable if there exist $n-p = |M|$ color classes of size 2.
Let $\{ (u_1, u'_1), (u_2, u'_2), \ldots, (u_{|M|}, u'_{|M|}) \}$ be the matching $M$ of the complement of $G$ and
let $T = S \backslash \{ p + 1, \ldots, n \}$. Since $\{ p + 1, \ldots, n \} \subset S$ and $|S| = 2|M| - r$, we
have $|T| = |S| - (n - p) = |M| - r$. Moreover, $T \neq \varnothing$.\\
In order to prove $\lambda^X_{vj} = \lambda^X_{vn} + \lambda^W_n + \frac{1}{r} \lambda^W_{p+1}$, we consider
three cases:\\
\textbf{Case $v = u_1$ and $r = 1$}. Let us consider that
$T = \{ t_1, t_2, \ldots, t_{|M|-1} \}$. Let $c^1$ be a $p$-eqcol such that $c^1(u_{i+1}) = c^1(u'_{i+1}) = t_i$ for
$1 \leq i \leq |M|-1$, $c^1(u_1) = c^1(u'_1) = j$ and $c^2 = intro(c^1,u_1)$.
Therefore, $\lambda^X_{u_1 j} = \lambda^X_{u_1 p+1} + \lambda^W_{p+1}$. As condition (a) asserts that
$\lambda^X_{u_1 p+1} = \lambda^X_{u_1 n} + \lambda^W_n$, we conclude that
$\lambda^X_{u_1 j} = \lambda^X_{u_1 n} + \lambda^W_n + \lambda^W_{p+1}$.\\
\textbf{Case $v = u_1$ and $r = 2$}. Since $|M| \leq \lfloor \frac{n-1}{2} \rfloor$, we have
$|\{1,\ldots,p\} \backslash T| = p - |M| + 2 \geq 3$ and we can ensure that there exist different colors
$k, l \in \{1,\ldots,p\} \backslash (T \cup \{j\})$. Moreover, there exists a vertex
$w \in V \backslash \{u_1, u'_1, \ldots, u_{|M|}, u'_{|M|}\}$ because $M$ is not perfect.\\
Now, we propose a pair of equitable colorings (namely $c^1$ and $c^2$)
in order to obtain several equalities. Let us consider $T = \{ t_1, t_2, \ldots, t_{|M|-2} \}$ and $c^1$, $c^2$ be equitable colorings such
that $c^1(u_{i+2}) = c^1(u'_{i+2}) = t_i$ for $1 \leq i \leq |M|-2$, $c^2(i) = c^1(i)$ for
$i \in V \backslash \{ u_1, u'_1, u_2, u'_2, w \}$ and
the colors of vertices $u_1$, $u'_1$, $u_2$, $u'_2$ and $w$ are:
\begin{center} \small
\begin{tabular}{|c|c@{\hspace{3pt}}c@{\hspace{3pt}}c@{\hspace{3pt}}c@{\hspace{3pt}}c|c|c@{\hspace{3pt}}c@{\hspace{3pt}}c@{\hspace{3pt}}c@{\hspace{3pt}}c|}
\hline
 \multicolumn{6}{|c|}{$c^1$} & \multicolumn{6}{|c|}{$c^2$} \\
\hline
 size & $u_1$ & $u'_1$ & $u_2$ & $u'_2$ & $w$ &
 size & $u_1$ & $u'_1$ & $u_2$ & $u'_2$ & $w$ \\
\hline
 $p$ & $j$ & $j$   & $k$ & $k$ & $l$ & $p+1$ & $p+1$ & $p+1$ & $k$ & $j$ & $l$ \\
 $p$ & $l$ & $l$   & $j$ & $j$ & $k$ & $p$ & $l$   & $l$   & $k$ & $k$ & $j$ \\
 $n$   & $j$ & $p+1$ & $l$ & $k$ & $n$ & $n$   & $p+1$ & $j$   & $l$ & $k$ & $n$ \\
 $n$   & $l$ & $p+1$ & $k$ & $n$ & $j$ & $n$   & $l$   & $p+1$ & $j$ & $n$ & $k$ \\
\hline
\end{tabular}
\end{center}
Each combination gives us a different equality of the form $\lambda^X x_1 + \lambda^W w_1 = \lambda^X x_2 + \lambda^W w_2$, namely 
\begin{enumerate}
\item[1.] $\lambda^X_{u_1 j} + \lambda^X_{u'_1 j} + \lambda^X_{u'_2 k} =
 \lambda^X_{u_1 p+1} + \lambda^X_{u'_1 p+1} + \lambda^X_{u'_2 j} + \lambda^W_{p+1}$
\item[2.] $\lambda^X_{u_2 j} + \lambda^X_{u'_2 j} + \lambda^X_{w k} =
 \lambda^X_{u_2 k} + \lambda^X_{u'_2 k} + \lambda^X_{w j}$
\item[3.] $\lambda^X_{u_1 j} + \lambda^X_{u'_1 p+1} =
 \lambda^X_{u_1 p+1} + \lambda^X_{u'_1 j}$
\item[4.] $\lambda^X_{u_2 k} + \lambda^X_{w j} =
 \lambda^X_{u_2 j} + \lambda^X_{w k}$
\end{enumerate}
Let us note that the addition
of the previous equalities gives $2 \lambda^X_{u_1 j} = 2 \lambda^X_{u_1 p+1} + \lambda^W_{p+1}$.
Since condition (a) asserts that $\lambda^X_{u_1 p+1} = \lambda^X_{u_1 n} + \lambda^W_n$, we conclude that
$2 \lambda^X_{u_1 j} = 2 \lambda^X_{u_1 n} + 2 \lambda^W_n + \lambda^W_{p+1}$.\\
\textbf{Case $v \neq u_1$}. Let $c^1$ be a $n$-eqcol such that $c^1(v) = j$, $c^1(u_1) = n$ and
$c^2 = swap_{j,n}(c^1)$. The conditions proved recently allows us to conclude that
$\lambda^X_{vj} = \lambda^X_{vn} + \lambda^W_n + \frac{1}{r} \lambda^W_{p+1}$.
\item[(c)] Let $(x^1,w^1)$ be a $j$-eqcol and $(x^2,w^2)$ be a $\theta$-eqcol. If any of these colorings does not
lie on $\F'$, we can always swap its color classes so that it belongs to the face.
Thus $\lambda^X x^1 + \sum_{k = \theta+1}^j \lambda^W_k = \lambda^X x^2$.
In virtue of conditions (a) and (b), the previous equation becomes
\[ \sum_{v \in V} \lambda^X_{vn} + n \lambda^W_n + (n - b_{Sj}) \frac{1}{r} \lambda^W_{p+1} + \sum_{k = \theta+1}^j \lambda^W_k = \sum_{v \in V} \lambda^X_{vn} + n \lambda^W_n + (n - b_{S\theta}) \frac{1}{r} \lambda^W_{p+1}, \]
and this leads to $\sum_{k = \theta+1}^j \lambda^W_k = (b_{Sj} - b_{S\theta}) \frac{1}{r} \lambda^W_{p+1}$.
\end{enumerate}
\end{proof}

Let us present an example where the previous theorem is applied.\\

\noindent \textbf{Example.} We assume that $G$ is the graph presented in Figure \ref{minigraph2}(a). Let us note that $\overline{G}$ has the matching $\{(4,5)$, $(3,6)$, $(1,7)$, $(2,8)$, $(9,11)\}$. So, for all $S$ such that
$8 \leq |S| \leq 9$ and $\{7,\ldots,11\} \subset S$, the assumptions of Theorem \ref{TONLYCOLORS}
hold and the $S$-color inequality defines a facet of $\ECP$ as expected. Furthermore,
since $\overline{G}$ has also matchings of sizes between 2 and 5, the $S$-color
inequality defines a facet for all $S$ such that $3 \leq |S| \leq 9$ and
$\{ 11 - \lceil \frac{|S|+1}{2} \rceil, \ldots, 11\} \subset S$.\\

Unlike the last theorem, Theorems \ref{T2RANK1}, \ref{T2RANK2}, \ref{TNEIGHBOR1}, \ref{TNEIGHBOR2}
and \ref{TNEIGHBOR3} are restricted to monotone graphs. However, these results might be extended to
the general case, but the proofs behind them turn very cryptic.

\end{document}